\documentclass[a4paper,preprintnumbers,floatfix,superscriptaddress,aps,10pt,twocolumn,notitlepage,longbibliography,noarxiv]{quantumarticle}
\pdfoutput=1
\usepackage{amsmath,amsthm,amsfonts,amssymb,mathtools,dsfont}
\usepackage{bm,bbm}
\usepackage{cancel}
\usepackage{marvosym}
\usepackage{graphicx}
\usepackage{xcolor}
\usepackage{xr}
\usepackage{float}
\usepackage{appendix}
\usepackage{graphicx}
\usepackage{dcolumn}
\usepackage{xcolor}
\usepackage{scalerel}
\usepackage{soul}
\usepackage{url}
\usepackage{hyperref}
\PassOptionsToPackage{linktocpage}{hyperref}
\usepackage[capitalize]{cleveref}
\crefname{figure}{Fig.}{Figs.}
\crefname{algorithm}{Protocol}{Protocols}

\usepackage{crossreftools}
\newcommand{\optionaldesc}[2]{%
  \phantomsection
  #1\protected@edef\@currentlabel{#1}\label{#2}%
}

\usepackage{acronym}
\usepackage{algorithm}
\usepackage[noend]{algorithmic}
\floatname{algorithm}{Protocol}

\usepackage{tikz-cd}

\makeatletter
\newtheorem*{rep@theorem}{\rep@title}
\newcommand{\newreptheorem}[2]{%
\newenvironment{rep#1}[1]{%
 \def\rep@title{#2 \ref{##1}}%
 \begin{rep@theorem}}%
 {\end{rep@theorem}}}
\makeatother
\newtheorem{theorem}{Theorem}
\newreptheorem{theorem}{Theorem}

\newtheorem{observation}[theorem]{Observation}

\newtheorem{definition}[theorem]{Definition}
\newtheorem{corollary}[theorem]{Corollary}
\newtheorem{techlemma}{Lemma}[section]
\newtheorem{techcorollary}[techlemma]{Corollary}

\DeclarePairedDelimiter{\bra}{\langle}{\vert}
\DeclarePairedDelimiter{\ket}{\vert}{\rangle}
\DeclarePairedDelimiterX\braket[2]{\langle}{\rangle}%
  {#1\kern0.15ex\delimsize\vert\kern0.15ex\mathopen{}#2}

\DeclarePairedDelimiterX\ketbra[2]{\vert}{\vert}%
  {#1\kern0.15ex\delimsize\rangle\delimsize\langle\kern0.15ex\mathopen{}#2}

\DeclarePairedDelimiterX{\abs}[1]{\lvert}{\rvert}{%
  \ifblank{#1}{\,\cdot\,}{#1}
}   
\DeclarePairedDelimiterX\norm[1]\lVert\rVert{%
  \ifblank{#1}{\,\cdot\,}{#1}
}   

\newcommand{\bk}[2]{\braket{#1}{#2}}
\newcommand{\kb}[2]{\ketbra{#1}{#2}}
\newcommand{\pure}[1]{\kb{#1}{#1}}

\newcommand{\1}{\openone}
\newcommand{\CC}{\mathbb{C}}
\newcommand{\RR}{\mathbb{R}}

\newcommand{\NN}{\mathbb{N}}

\newcommand{\LL}{\mathcal{L}}
\newcommand{\HH}{\mathcal{H}}
\newcommand{\U}{\mathrm{U}}
\newcommand{\ii}{\ensuremath\mathrm{i}}
\newcommand{\e}{\ensuremath\mathrm{e}} 
\newcommand{\dist}{\mathrm{inFid}}
\newcommand{\veps}{\varepsilon}

\newcommand{\vbt}[1]{{\ttfamily #1}} 

\newcommand{\iS}{\mathrm{s}}
\newcommand{\iH}{\mathrm{h}}

\newcommand{\iU}{\mathrm{u}}

\newcommand{\iCX}{\mathrm{cx}\,}
\newcommand{\Cl}{\mathcal{C}l}

\newcommand{\argdot}{\,\cdot\,}
\newcommand{\XX}{\mathrm{X}}
\newcommand{\PP}{\mathbb{P}}
\renewcommand{\d}{\mathrm{d}}
\newcommand{\eutp}{\stackrel{\propto}{=}}
\newcommand{\ins}[2]{#1\oplus_{#2}} 
\newcommand{\wout}[2]{{#1}_{:#2}} 

\providecommand\given{}
\newcommand\SetSymbol[1][]{%
  \nonscript\:#1\vert
  \allowbreak
  \nonscript\:
  \mathopen{}}
\DeclarePairedDelimiterX\Set[1]\{\}{%
\renewcommand\given{\SetSymbol[\delimsize]}   
#1
}

\renewcommand{\AA}{\mathrm{A}}
\renewcommand{\vec}[1]{\mathbf{#1}}

\DeclareMathOperator{\CPTP}{CPTP}
\DeclareMathOperator{\DM}{\mathcal{D}} 

\DeclareMathOperator{\diag}{diag}

\DeclareMathOperator{\Tr}{Tr} 
\DeclareMathOperator{\tr}{Tr}
\DeclareMathOperator{\LandauO}{\mathcal{O}} 
\DeclareMathOperator{\inFid}{\inFid}
\DeclareMathOperator{\F}{F}
\newcommand{\Fid}{\mathrm{F}}
\newcommand{\avg}{\mathrm{avg}}
\newacro{RB}{randomized benchmarking}
\newacro{GST}{gate set tomography}
\newacro{POVM}{positive operator-valued measure}
\newacro{PVM}{projection-valued measure}
\newacro{CP}{completely positive}
\newacro{CPTP}{completely positive trace preserving}
\newacro{PSD}{positive semidefinite}
\newacro{NISQ}{noisy intermediate-scale quantum}
\newacro{SPAM}{state preparation and measurement} 
\newacro{ONB}{orthonormal basis}
\newacro{SDI}{semi-device independent}
\newacro{QSQ}{quantum system quizzing}

\begin{document}

\title{Sound certification of memory-bounded quantum computers}

\author{Jan N\"oller}
\email{jan.noeller@tu-darmstadt.de}
\affiliation{Department of Computer Science, Technical University of Darmstadt, Darmstadt, 64289 Germany}
\author{Nikolai Miklin}
\affiliation{Hamburg University of Technology, Institute for Quantum Inspired and Quantum Optimization, Hamburg, Germany}
\affiliation{Institute for Applied Physics, Technical University of Darmstadt, Darmstadt, Germany}
\author{Martin Kliesch}
\affiliation{Hamburg University of Technology, Institute for Quantum Inspired and Quantum Optimization, Hamburg, Germany}
\author{Mariami Gachechiladze}
\affiliation{Department of Computer Science, Technical University of Darmstadt, Darmstadt, 64289 Germany}

\begin{abstract}
The rapid advancement of quantum hardware calls for the development of reliable methods to certify its correct functioning.
However, existing certification tests often fall short: they either rely on flawless state preparation and measurement or lack soundness guarantees, meaning that they do not rule out incorrect implementations of the target operations by a quantum device.
We introduce an approach, which we call quantum system quizzing, for the certification of quantum gates in a practical server-user scenario, where a classical user tests the results of quantum computation performed by a quantum server by checking its responses to a set of predesigned small-sized computational problems.
Importantly, this approach does not require trusted state preparation and measurement and is thus inherently free from the associated systematic errors.
For a wide range of relevant gate sets, including a universal one, we prove our certification protocol to be \emph{sound}; i.e., it is guaranteed to reject any incorrect gate implementation, under the assumptions of a known Hilbert space dimension and context independence of error.
A major technical challenge that we are first to resolve is recovering the tensor product structure of a multi-qubit system in the memory-bounded single-device setup. 
Finally, we prove the robustness of our protocol and validate its sample and computational efficiency through extensive numerical experiments.
Our protocol is platform-agnostic and introduces a new paradigm for benchmarking and comparing diverse quantum architectures.
\end{abstract}

\maketitle

\makeatletter
\hypersetup{pdftitle = {Sound certification of memory-bounded quantum computers},
       pdfauthor = {Jan Nöller, Nikolai Miklin, Martin Kliesch, Mariami Gachechiladze},
       pdfsubject = {Quantum computing},
       pdfkeywords = { 
              quantum, certification, verification, quizzing, dimension, assumption, 
              self-testing, soundness, sound,
              verifier, user, remote, computer, computation, 
              dynamics, gate, circuit, algorithm, protocol, 
              clifford, universal,
              gauge, freedom, 
              }
      }

\section{Introduction}
As quantum hardware advances rapidly, its certification is becoming increasingly important to ensure reliability and consistency.
However, existing certification methods face certain challenges. 
Many of these methods verify the proper functioning of some components of a device based on the assumption that other elements (e.g., measurement devices) operate essentially perfectly~\cite{Eisert2020Quantumcertification, kliesch2021theory,liu2020efficient}, possibly introducing systematic errors and leading to hardware miscalibration~\cite{kliesch2021theory}. 
For instance, quantum process tomography is particularly susceptible to \ac{SPAM} errors~\cite{mohseni2008quantum}. 
While alternative methods, such as \ac{RB} \cite{EmeAliZyc05,LevLopEme07,EmeSilMou07,KniLeiRei08,DanCleEme09,Helsen2022framework,Heinrich22GeneralGuarantees} and \ac{GST}~\cite{merkel2013self, BluGamNie13, Nielsen2021gatesettomography,brieger2023compressive} avoid \ac{SPAM} errors, they do not guarantee that all erroneous implementations of quantum gates will be correctly identified~\cite{proctor2017WhatRandomizedBenchmarking}. 
This gap highlights the need for improved methods essential for precise calibration and control of scalable quantum devices.

Any certification method must satisfy two key properties: \textit{completeness} and \textit{soundness}~\cite{kliesch2021theory}. 
Completeness means that the certification test should accept a correctly implemented quantum operation, a requirement that the currently used \ac{SPAM}-robust characterization methods meet~\cite{Helsen2022framework}. 
However, these methods, if used for certification, lack a proof of soundness. 
A test is sound if it accepts \textit{only} those operations that are correctly implemented. 
In \ac{RB} and \ac{GST} literature, the lack of soundness for the certification task is hinted to be connected to \emph{the gauge problem}~\cite{proctor2017WhatRandomizedBenchmarking,Nielsen2021gatesettomography} but no solution is provided.

In the field of the foundations of quantum mechanics, a complete and sound approach to certifying quantum states, measurements, and certain channels, known as self-testing, has been developed~\cite{mayers2003self,Supic2020selftestingof,sekatski2018certifying,Tsirelson93, POPESCU1992which,Acin2012Randomness, Yang2013Robust, Yang2014Robust, Bancal2015Physical, coopmans2019Robust, mckague2010selftestinggraphstates, Baccari2020Scalable, Wu2014Self, Supic2018Self-testingMPE, Supic2023Quantumnetworks, McKague2012RobustSinglet, Supic2016chainedBI, Bamps2015SOS, sekatski2018certifying, Wagner2020DICquantuminstruments,Miklin2020a}. 
It is based solely on the observed statistics, making it free from systematic errors, and has found application in various quantum information processing tasks~\cite{Devetak2005Distillation, Acin2012Randomness}. 
However, self-testing is typically based on the distribution of shared entanglement between spatially separated parties, an assumption that is hard to justify within a single noisy device exhibiting crosstalk~\cite{rybotycki2024violation}.  Recently, self-testing on a single device was proven under computational assumptions based on post-quantum cryptographic primitives~\cite{mahadev2018classical,metger2021self}. 
Realizing such a scheme, however, requires large-scale quantum computation and thus seems out of reach for near-term quantum devices~\cite{stricker2022towards}. 
Finally, under the assumption that single-qubit gates are implemented with high accuracy, the so-called accreditation protocols~\cite{Ferracin18AccreditingOutputsOf, Ferracin21ExperimentalAccreditationOf} ascertain the correctness of the outputs of arbitrary-sized noisy quantum computers. 
However, they provide certificates for a different abstraction level, the correct outputs but not the inner workings of the quantum computer. 

In this work, we propose a method called \ac{QSQ} for the simultaneous certification of multi-qubit quantum gates, state preparation, and measurement.
Due to the simultaneous nature of the certification, quantum gates are certified in a way that does not require trusting the pre-calibration of \ac{SPAM}. 
The method does not require auxiliary qubits or post-processing of measurement data. 
Instead, the certification procedure checks consistency with the expected outcomes for a set of predesignated, computationally small problems in each experimental round, thereby eliminating the need to estimate probabilities. The construction of such sets for universal gate sets constitutes one of the technical contributions of this work.
These features make the method well-suited for experimental implementation on current quantum hardware due to its efficient sample complexity, a claim that we substantiate through an extensive numerical study.

As our main technical contribution, we provide soundness and completeness guarantees for the certification of several quantum gate sets, including a universal gate set, under the assumptions of known underlying Hilbert space dimension and context-independence of the gate implementations. 
These assumptions are also among those considered in popular characterization methods~\cite{Helsen2022framework,Nielsen2021gatesettomography}, and in an earlier proposal for self-testing quantum gates in a single device setup~\cite{vanDam2007self-testing}.
Importantly, in contrast to \cite{vanDam2007self-testing}, we do not assume \emph{addressability} of individual qubits, i.e., tensor product structure of the implemented operations, and instead recover it from the results of the certification test by leveraging input-output relations. Finally, we provide general robustness analysis and extensive numerical experiments suggesting an optimal scaling of sample complexity of the \ac{QSQ} protocol.

\begin{figure}
    \centering
    \includegraphics[width=0.9\linewidth]{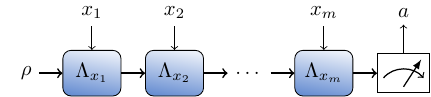}
    \caption{The considered scenario: a quantum system is initialized, undergoes a series of transformations, specified by classical instructions $x_1,x_2,\dots,x_m$, chosen from a finite set $\XX$, and is measured producing an outcome $a\in \AA$.}
    \label{fig:scenario}
\end{figure}

\subsection{Notation \& preliminaries}
The sets of density and unitary operators on a Hilbert space $\HH$ are denoted by $\DM(\HH)$ and $\U(\HH)$, respectively.
The Pauli matrices are denoted by $X, Y$ and $Z$, with their corresponding eigenstates given by $\ket{\pm},\ket{\pm_y},\{\ket{0},\ket{1}\}$, respectively.
We denote the Hadamard gate by $H=(X+Z)/\sqrt{2}$, $S$-gate by $S=\kb{0}{0}+\ii\kb{1}{1}$, and define $C_\iH X\coloneqq \kb{+}{+}\otimes\1+\kb{-}{-}\otimes X$, and a two-qubit non-Clifford phase gate $CS \coloneqq \kb{0}{0}\otimes\1+\kb{1}{1}\otimes S$. 
We use a superscript $U^{(i)}$ to denote the canonical embedding of the unitary $U$ acting on qubit $i$ to the $n$-qubit space.

We consider the quantum computation scenario shown in \cref{fig:scenario}. 
A quantum system is prepared in an initial state, undergoes a series of transformations instructed by an input string $\vec{x}=x_1\dots x_l\in \XX^*$, where $\XX^*$ denotes the set of all strings over the alphabet $\XX$, and is finally measured producing the outcome $a\in \AA$, with $\AA$ being the set of possible outcomes. For a string $\vec{x}\in\XX^\ast$, we abbreviate $\vec x^k\coloneqq\vec x\dots \vec x$ to denote its $k$-fold repetition, with the convention that $\vec{x}^0=\epsilon$ is the \emph{empty string}.
We assume that each instruction $x\in \XX$ is associated with the corresponding operation performed by a quantum computer, described by a quantum channel $\Lambda_x$, which is independent of other operations performed in that experiment, an assumption we refer to as \emph{context independence}. 
Under this assumption, an experiment in the considered scenario is described by \emph{a quantum model}.
\begin{definition}\label{def:model}
    A quantum model $\mathcal{M}$ is a $3$-tuple $\mathcal{M}=(\rho,\{\Lambda_x\}_{x\in \XX},\{M_a\}_{a\in \AA})$, specifying an initial state $\rho$, a set $\{\Lambda_x\}_{x\in \XX}$ of quantum channels, each with an associated label, and a \ac{POVM} $\{M_{a}\}_{a\in \AA}$, all defined over a finite-dimensional Hilbert space $\HH$.
\end{definition}
We refer to $\dim(\mathcal{H})$ simply as the dimension of the model.
For quantum models with unitary channels, $\Lambda_x(\argdot)=U_x(\argdot)U_x^\dagger $, and a pure initial state $\rho=\kb{\psi}{\psi}$, we write $\mathcal{M}=(\ket{\psi},\{U_x\}_{x\in \XX},\{M_a\}_{a\in \AA})$.
A quantum model $\mathcal{M}$ augmented by a set of channels $\{\Lambda_x\}_{x\in \XX'}$ with an additional set of labels $\XX'$, is denoted as
\begin{equation}
\mathcal{M}+\{\Lambda_x\}_{x\in \XX'}\coloneqq(\rho,\{\Lambda_x\}_{x\in \XX\cup \XX'},\{M_a\}_{a\in \AA}). 
\end{equation}
We choose labels matching a target gate that the quantum computer should implement.
For instance, for the $S$ gate we use the label ``$\iS$''.

We consider three classes of quantum models:
\begin{enumerate}\itemsep-1em
    \item the $n$-qubit $S$-gate model,
    \begin{equation}\label{eq:def_model_s}
        \mathcal{S}_n \coloneqq\left(\ket{+}^{\otimes n},\{S^{(k)}\}_{k=1}^n,\{\kb{J}{J}\}_{J\in\{\mathrm{+,-}\}^n}\right),
    \end{equation}
    \item the $n$-qubit Clifford group model,
    \begin{equation}\label{eq:def_model_cl}
        \mathcal{C}l_n \coloneqq\mathcal{S}_n+\{H^{(1)}\} + \{C_\iH X^{(1,k)}\}_{k=2}^n,
    \end{equation}
    \item the $n$-qubit universal model,
    \begin{equation}\label{eq:def_model_u}
        \mathcal{U}_n \coloneqq \mathcal{C}l_n + \{CS^{(1,2)}\}.
    \end{equation}
\end{enumerate}
We explain our choice of $CS$, instead of, for example, $T$-gate as a non-Clifford gate, after presenting our results.
Finally, we need a notion of \emph{equivalence} between quantum models.
 
\begin{definition}\label{def:equiv}
    We say that two quantum models $\mathcal{M} = (\rho,\{\Lambda_x\}_{x\in \XX},\{M_a\}_{a\in \AA})$ and $\mathcal{N}=(\tilde\rho,\{\tilde\Lambda_x\}_{x\in \XX},\{\tilde M_a\}_{a\in \AA})$, both defined over $\HH$, are equivalent (denoted as $\mathcal{M}\sim_u\mathcal{N}$), if there exists a unitary $U$ (possibly followed by a complex conjugation) such that $U\tilde\rho U^\dagger = \rho$, $U\tilde M_a U^\dagger = M_a$ for all $a\in \AA$ and $ U\tilde\Lambda_x(U^\dagger (\argdot) U)U^\dagger=\Lambda_x(\argdot)$ for all $x\in \XX$.
\end{definition}
Equivalent models describe the same experiment~\cite{wigner1931gruppentheorie}, and in the following, we refer to this unitary (or anti-unitary) degree of freedom collectively as \emph{unitary gauge}.

\section{Quantum system quizzing protocol}
In the self-testing literature~\cite{Supic2020selftestingof}, one would say that a quantum model $\mathcal{M}$ is \emph{self-tested} by the observed correlations $\mathbb{P}(a\vert x_1 \dots x_l)=\tr[M_a \Lambda_{x_l}\circ\dots\circ\Lambda_{x_1}(\rho)]$, if any other model defined over the same Hilbert space and producing the same correlations is equivalent to $\mathcal{M}$. 
Here, we propose a stronger notion of certification that avoids reliance on estimating correlations or any figure of merit, and instead tests only for the \emph{compliance} with correct outcomes, which we call quantum system quizzing.
For a quantum model $\mathcal{M}$, we define the \emph{output map} $a_\mathcal{M}$
that assigns to each input string $\vec{x}=x_1\dots x_l$ the set of outcomes that can be observed with non-zero probability, 
\begin{equation}\label{eq:outputMap}
    a_\mathcal{M}(\vec{x}) \coloneqq \Set{a\in \AA \given \tr[M_a \Lambda_{x_l}\circ\dots \circ\Lambda_{x_1}(\rho)]>0}. 
\end{equation}
Now we formally introduce the~\ac{QSQ} protocol:
\begin{algorithm}[H]
\caption{Quantum system quizzing (QSQ)}\label{protocol}
\begin{algorithmic}
\STATE \textbf{Input:} $N$ -- number of repetitions, $\mathcal{X}\subset \XX^*$ -- finite set of gate sequences, $a_\mathcal{M}$ -- output map for a target model $\mathcal{M}$.
\FOR{$i\in \{1,\dots, N\}$}
\STATE draw $\vec{x}$ from $\mathcal{X}$ uniformly at random;
\STATE initialize a system, run the quantum circuit corresponding to $\vec{x}$, record the outcome $a$;
\IF{$a\notin a_\mathcal{M}(\vec{x})$}
\STATE{output \vbt{reject} and end the protocol;}
\ENDIF
\ENDFOR
\STATE output \vbt{accept}.
\end{algorithmic}
\end{algorithm}

Each run of the \ac{QSQ} protocol starts with an input $\vec{x}$ selected at random from the predetermined set of sequences, or \emph{quizzes}, $\mathcal{X}$, after which a quantum system is initialized, the sequence of gates corresponding to $\vec{x}$ is performed, the system is measured with the outcome $a$ recorded.
If the response is not among the measurement outcomes $a_\mathcal{M}(\vec{x})$ expected from the target quantum model $\mathcal{M}$, the protocol aborts by rejecting the implemented quantum model. 
Clearly, \cref{protocol} leads to a complete test: if the implemented model $\mathcal{N}$ is equivalent to the target model $\mathcal{M}$, the protocol outputs \vbt{accept} for any $N$.
On the contrary, if \cref{protocol} accepts for any $N$, then it must be that $a_\mathcal{N}(\vec{x})\subseteq a_\mathcal{M}(\vec{x})$, $\forall \vec{x}\in \mathcal{X}$ for any $\mathcal{N}$.
Investigating whether this implies that $\mathcal{N}$ is equivalent to $\mathcal{M}$ for the set of quizzes $\mathcal{X}$, i.e., whether the test is \emph{sound}, is precisely the focus of the current work.

\begin{definition}
\label{def:self-testing}
    Let $\mathcal{M} = (\rho,\{\Lambda_x\}_{x\in \XX},\{M_a\}_{a\in \AA})$ be a $d$-dimensional quantum model.
    We say that a set of inputs $\mathcal{X}\subset \XX^*$ certifies $\mathcal{M}$, if every $d$-dimensional model $\mathcal{N}=(\tilde\rho,\{\tilde\Lambda_x\}_{x\in \XX},\{\tilde M_a\}_{a\in \AA})$ satisfying $a_\mathcal{N}(\vec{x}) \subseteq a_{\mathcal{M}}(\vec{x})$ for all $\vec{x}\in\mathcal{X} $ is equivalent to $\mathcal{M}$. 
\end{definition}

In practice, \cref{protocol} rejects any implementation if we let $N\to \infty$ due to errors in the gates or \ac{SPAM}. Analogously, in practice, we can only assert the condition $a_\mathcal{N}(\vec{x})\subseteq a_\mathcal{M}(\vec{x})$ up to some level of confidence.
Intuitively, the closer the implementation is to the target one, the higher the probability that \cref{protocol} accepts it for a given $N$, and vice-versa. 
Here, we first answer the question of whether the models in \cref{eq:def_model_s,eq:def_model_cl,eq:def_model_u} can, \emph{in principle}, be certified for $N\to \infty$, and later discuss its robustness and the effects of finite $N$, along with the discussion of the method's sample complexity.

\section{Certification of multi-qubit models}
First, we note that the inputs $\mathcal{X}_1\coloneqq \Set{(\iS\iS)^i\given i\in\Set{0,1,2}}$ certify the $\mathcal{S}_1$ model (see \cref{app:1_qubit}), which is a simplified argument of Ref.~\cite{noller2025classical}.
Generalizing this result to the two-qubit model $\mathcal{S}_2$ is already nontrivial, as we assume only that the overall system is four-dimensional, and the tensor product structure between the qubit subsystems must be reconstructed from the responses in the \ac{QSQ} protocol. 
Specifically, we consider the following quizzes:
\begin{equation}\label{eq:two-qubits-S-inputs}
    \mathcal{X}_2\coloneqq\mathcal{X}_a\cup\mathcal{X}_b\cup\Set*{(\iS_a\iS_b)^2,(\iS_b\iS_a)^2},
\end{equation}
with the sets $\mathcal{X}_a\coloneqq\Set{\iS_b^{2j}\iS_a^{i}\given j\in \{0,1\},0\leq i\leq 4}$ and $\mathcal{X}_b \coloneqq \Set{\iS_a^{2i}\iS_b^{j}\given i\in\{0,1\},0\leq j\leq 4}$. The subscripts $a$ and $b$ of the classical labels $\iS_a$ and $\iS_b$ are used to denote instructions for implementing the $S$-gate on qubit $A$ and qubit $B$, respectively. 
\begin{theorem}\label{res:S_2}
    The input set $\mathcal{X}_2$ defined in \cref{eq:two-qubits-S-inputs} certifies the two-qubit quantum model $\mathcal{S}_2$.
\end{theorem}

We emphasize again that, apart from the labels, \emph{a priori}, we do not assume any structure that would differentiate between the corresponding implemented channels $\Lambda_{\iS_a}$ and $\Lambda_{\iS_b}$.  
We present a sketch of the proof below, while a detailed proof can be found in \cref{app:2_qubits}. 
The analogous $n$-qubit results are stated and proven as \cref{res:S_n} in \cref{app:n_qubits}. 
To start, we need the notion of a \emph{subchannel}. 

\begin{definition}\label{def:subchannels}
For a channel $\Lambda$ on a bipartite system $\HH_A\otimes \HH_B$ we define its subchannel on subsystem $A$, associated to $\rho\in \DM(\mathcal{H}_B)$ as the map $\Lambda\vert_A^\rho:\DM(\HH_A)\to \DM(\HH_A)$  via $\Lambda\vert^\rho_A(\sigma) \coloneqq  \tr_B[\Lambda(\sigma\otimes\rho)]$, for all $\sigma\in \DM(\HH_A)$.
\end{definition}

\begin{proof}[Proof sketch]
Let $\mathcal{N} = (\rho, \{\Lambda_{\iS_a},\Lambda_{\iS_b}\}, \{M_{J}\}_{J\in\{+,-\}^2})$ be a model defined over $\HH\cong\CC^4$ that satisfies $a_{\mathcal{N}}(\vec{x})\subseteq a_{\mathcal{S}_2}(\vec{x})$ for all $\vec{x}\in \mathcal{X}_2$. 
Abbreviate $\Lambda_a\coloneqq \Lambda_{\iS_a},\Lambda_b:=\Lambda_{\iS_b}$. First of all, the measurement distinguishes perfectly between the states $\rho_{ij}\coloneqq \Lambda^{2i}_{a}\circ\Lambda_{b}^{2j}(\rho)$ for $i,j\in\{0,1\}$. 
The dimension constraint then dictates that $\{M_{J}\}_{J\in\{+,-\}^2}$ is a rank-$1$ projective measurement in the \ac{ONB} $\{\ket{+},\ket{-}\}^{\otimes 2}$ by choosing a suitable unitary gauge. 
This induces an effective tensor-product structure $\HH=\HH_A\otimes\HH_B$.

Testing the inputs $\mathcal{X}_a$ shows that the second bit of the outputs is independent of how often we call the instruction $\iS_a$. 
Then, we conclude that $\Lambda_a$ acts trivially on the $B$-subsystem if the latter is in one of the reference basis states $\{\ket{+},\ket{-}\}\in \mathcal{H}_B$. 
We exploit this fact to certify the single-qubit subchannels $\Lambda_{a}\vert_A^{\ket{+}},\Lambda_{a}\vert_A^{\ket{-}}$, since the quizzes $\mathcal{X}_1$ for $\mathcal{S}_1$ certification are contained in $\mathcal{X}_a$. 
We conclude that both subchannels act like the $S$-gate. Similarly, we certify the subchannels $\Lambda_{b}\vert_B^{\ket{+}},\Lambda_{b}\vert_B^{\ket{-}}$ from $\mathcal{X}_b$.

Note that the unitarity of the subchannels does not guarantee the unitarity of the channel acting on the whole Hilbert space.
However, to prove unitarity of, e.g., $\Lambda_a$, it  suffices to find a state $\ket{\varphi}\in \HH_B$ with $\braket{+}{\varphi},\braket{-}{\varphi}\neq 0$, such that the associated subchannel $\Lambda_{a}\vert_A^{\ket{\varphi}}$ has a pure state in its image (see~\cref{lemma:unitarity_from_subchannels} in \cref{app:lemmas}). 
The last two instructions in the set $\mathcal{X}_2$ ensure precisely that such states are prepared in the protocol for both $\Lambda_a$ and $\Lambda_b$.
Lastly, one reconciles the unitary gauges arising in the certification of the two gates.
\end{proof}

We show next that once we have certified a quantum model $\mathcal{M}$, we can also certify augmented quantum models $\mathcal{M}+\{U\}$ for a large class of unitaries $U$ with little effort.
Particularly interesting is that one can add \emph{entangling gates}, which we demonstrate by augmenting the $\mathcal{S}_2$-model with the $C_\iH X$-gate in \cref{res:C_hX}.
To state our result concisely, we formalize a notion of full coherence for \acp{ONB} with more than two elements.
\begin{definition}\label{def:coherence_graph}
    Given an \ac{ONB} $\mathcal{B} = \{\ket{\psi_i}\}_{i=0}^{d-1}$ of a $d$-dimensional Hilbert space $\HH$ and a finite set $\mathcal{C}\subseteq\DM(\HH)$, define the coherence graph of $\mathcal{C}$ with respect to $\mathcal{B}$ as a simple graph with vertices corresponding to the elements of $\mathcal{B}$ and each pair of vertices $(i, j)$ connected by an edge if and only if $\bra{\psi_i}\sigma\ket{\psi_j}\neq 0$ for some $\sigma\in\mathcal{C}$. 
    We say that the set $\mathcal{C}$ is fully coherent with respect to $\mathcal{B}$ if the coherence graph is connected. 
\end{definition}
Let us denote by $\mathcal{D}_\mathcal{M}$ the set of states attainable in a model $\mathcal{M}$ by acting with finite channel compositions on the initial state. 
We are now ready to state our second main result, which discusses the certification of augmented models.
\begin{theorem}\label{cor:adding_inputs}
    Let a $d$-dimensional model $\mathcal{M} = (\ket{\psi_0},\{U_x\}_{x\in \XX}, $ $\{\kb{\psi_a}{\psi_a}\}_{a\in[d]})$ defined over $\HH$ be certified by a set of quizzes $\mathcal{X}$, and let $U \in \U(\HH)$. 
    If the following conditions are satisfied,
    \begin{enumerate}
    \itemsep0em
    \item[{\crtcrossreflabel{\textit{(i)}}[prop:1]}] All gates $\{U_x\}_{x\in \XX}$ are of finite order;
    \item[{\crtcrossreflabel{\textit{(ii)}}[prop:2]}] $U$ has an eigenbasis $\mathcal{B}\,\subseteq\DM_\mathcal{M}$;
    \item[{\crtcrossreflabel{\textit{(iii)}}[prop:3]}] There exists a finite subset $\mathcal{C}\subseteq\mathcal{D}_\mathcal{M}$ of the attainable states which is fully coherent with respect to $\mathcal{B}$ and satisfies $U\mathcal{C} U^\dagger\subseteq \mathcal{D}_\mathcal{M}$;
    \end{enumerate}
    then the augmented model $\mathcal{M}+\{U\}$ is certified by a set of quizzes with cardinality $d+\abs{\mathcal{C}}+\abs{\mathcal{X}}$.
\end{theorem}
\begin{proof}
We denote the channel associated with the unitary $U_x$ as $\Lambda_x$.
Let ``$\iU$'' be the label corresponding to the additional gate $U$, and $\Lambda_\iU$ the corresponding implemented operation. 
We abbreviate the composition of the channels, corresponding to the input string $\vec{x}=x_1\dots x_l\in \XX^*$, as $\Lambda_{\vec{x}}\coloneqq\Lambda_{x_l}\circ\dots\circ\Lambda_{x_1}$.
For simplicity, we assume that the unitary gauge which realizes the equivalence between $\mathcal{M}$ and its implementation is the identity, i.e., that $\mathcal{M}+\{\Lambda_\iU\}$ is the implemented augmented model.
Otherwise, when the gauge unitary $W$ is $W\neq \openone$, the same proof as below can be applied to the augmented model $\mathcal{M}+\{WUW^\dagger\}$.

We begin by choosing sequences $\vec{x}_j\in \XX^*$ for $j\in[d]$, such that $\pure{\phi_j}\coloneqq \Lambda_{\vec{x}_j}(\pure{\psi_0})$ form an \ac{ONB} of eigenstates of $U$, guaranteed by \ref{prop:2}. 
Additionally, we can choose complementary sequences $\vec{x}_j'$ and outcomes $a_j$ such that $a_\mathcal{M}(\vec{x}_j\vec{x}'_j)=\{a_j\}$, since  \ref{prop:1} guarantees that we can implement the inverses of all gates in $\mathcal{M}$. Since $\mathcal{M}$ can be certified in the \ac{QSQ} protocol, we have that $\pure{\psi_{a_j}} = \Lambda_{\vec{x}'_j}(\pure{\phi_j}), \forall j\in [d]$.
Then, from the condition that the sequences $\vec{x}_j \mathrm{u}\vec{x}'_j$ also result in outcomes $a_j$, for $j\in [d]$, we deduce that the channel $\Lambda_\iU$ acts trivially on the states $\{\pure{\phi_j}\}_{j\in [d]}$. 

Let $\mathcal{C}$ be a finite subset of states, as in \ref{prop:3}. 
For every $\sigma\in\mathcal{C}$ we then find $\vec{y}_\sigma\in\XX^*$ such that $\sigma=\Lambda_{\vec{y}_\sigma}(\pure{\psi_0})$. 
From the assumptions, we know also that $U\sigma U^\dagger\in\mathcal{D}_\mathcal{M}$ for all $\sigma\in\mathcal{C}$ and again with \ref{prop:1}, we can choose sequences $\vec{y}'_\sigma$ and outcomes $a_\sigma\in [d]$ which are deterministically observed for the instructions $\vec{y}_\sigma\mathrm{u}\vec{y}_\sigma'$. 
If the model $\mathcal{M}+\{\Lambda_\iU\}$ produces only the correct outputs $a_\sigma$ for the respective sequences, we infer that $\Lambda_\mathrm{u}(\sigma)=U\sigma U^\dagger $ for all $\sigma\in\mathcal{C}$. 
Then we can apply \cref{lemma:channel=U} to conclude that indeed the channel is unitary and $\Lambda_\mathrm{u}(\argdot) = U(\argdot) U^\dagger$. 
\end{proof}
\begin{corollary}\label{res:C_hX} 
The input set $\mathcal{X}_2 \cup \mathcal{X}_{\iCX}$ with $\mathcal{X}_2$ as in \cref{eq:two-qubits-S-inputs} and $
\mathcal{X}_{\iCX}\coloneqq\Set{\iS_a^{2i}\iS_b^{2j}\iCX\given i,j\in\{0,1\}}\cup\Set{\iS_b\iCX\iS_b,\iS_a\iCX\iS_a,\iS_a\iS_b^2\iCX\iS_a}$ certifies the two-qubit quantum model $\mathcal{S}_2+\{C_\iH X\}$.
\end{corollary}

To certify the $\Cl_2$ model, we need to augment the $\mathcal{S}_2+\{C_\iH X\}$ model with the Hadamard gate. 
Since $C_\iH X$ and $S$-gates can implement a SWAP operator, it suffices to certify the $H$-gate on one qubit only. 
This can be proven similarly to the certification of the $S$-gate via subchannels. 
Therefore, we omit the proof here and leave the details to \cref{app:2_qubits} (\cref{res:Clifford_2}). 
Instead, we formulate our main result for a universal $n$-qubit model. 

\begin{figure*}
    \centering
    \includegraphics[width=.32\linewidth]{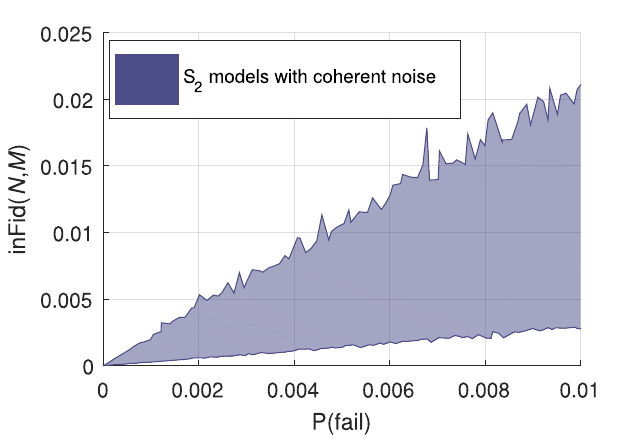}
    \includegraphics[width=.32\linewidth]{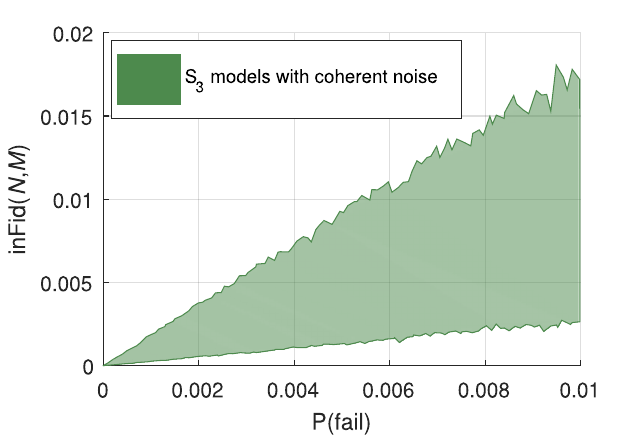}
    \includegraphics[width=.32\linewidth]{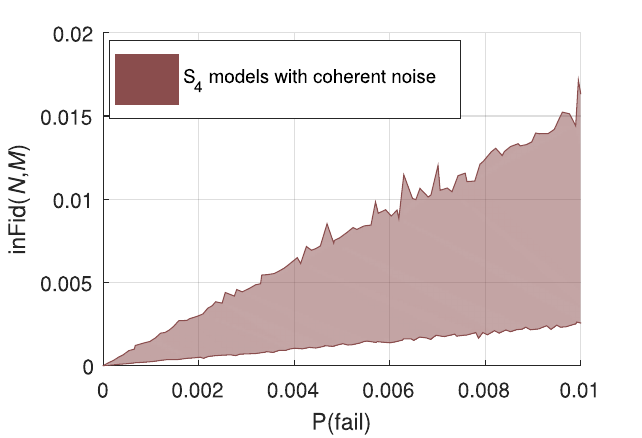}
    \caption{Results of the numerical investigation for $\mathcal{S}_2$ (left),  $\mathcal{S}_3$ (center), and $\mathcal{S}_4$ (right). The coefficients for the soundness (upper borderline) are approx.~$2.6447,2.1993$, and $1.9274$, respectively, and the coefficients for the completeness (lower borderline) are approx.~$0.2448,0.2217$, and $0.2096$, respectively. Each of the plots was created from sampling $10^7$ random noisy models.}
    \label{fig:numSn}
\end{figure*}

\begin{theorem}\label{res:universal}
For every $n\geq 2$, the $n$-qubit model $\mathcal{U}_n = \Cl_n + \{CS^{(1,2)}\}$ can be certified by a quiz set $\mathcal{X}^u_n$ with cardinality $\abs{\mathcal{X}^u_n}\in \LandauO(n 2^n)$.
\end{theorem} 
A set of instructions $\mathcal{X}_n$ to certify $\mathcal{S}_n$ is explicitly stated in \cref{res:S_n} in \cref{app:n_qubits}, while the sets for the Clifford and universal models are obtained by augmenting the $\mathcal{S}_n$ model. 
See the proof in \cref{app:n_qubits}. 
Note that the cardinality of the quiz set is asymptotically optimal in the sense that any certifying set of quizzes, as defined in~\cref{def:self-testing}, requires at least $2^n$ input sequences just to test all the measurement effects $\{M_J\}_{J\in\{+,-\}^n}$.
This, however, does not directly imply that the number of repetitions $N$ in \cref{protocol} must scale exponentially with $n$, as we argue in the robustness analysis section.

Before, let us comment on the choice of universal gate sets. 
The most common universal model to consider is $\mathcal{C}l_n+\{T^{(1)}\}$.
However, this model does not fit exactly in the scope of \ac{QSQ}, since its output map \eqref{eq:outputMap} is equal to the output map of $\mathcal{C}l_n+\{(TZ)^{(1)}\}$, and these two models are \emph{not} equivalent~\cite{schroeder2025certifying}.
This ambiguity can be resolved by estimating probabilities~\cite{noller2025classical}, which requires a relaxation of \cref{def:self-testing}. 
On the other hand, a universal model with $S$, $H$, and $CCZ$-gates can also be certified in the~\ac{QSQ} protocol using~\cref{cor:adding_inputs}. We choose the $CS$ gate in $\mathcal{U}_n$ because it allows universal computation already for $n=2$.

\subsection{Robustness analysis}
Next, we discuss the effect of finite $N$. We first provide a general analytical result, which guarantees a polynomial dependence on the targeted fidelity by applying methods from algebraic geometry. We then supplement these results with extensive numerical experiments. 

In order to discuss robustness, we need to introduce a notion of closeness. For this purpose, we introduce the infidelity between two quantum models, $\mathcal{N} = (\tilde\rho,\Set{\tilde\Lambda_x}_{x\in \XX},\Set{\tilde M}_{a\in \AA})$ and $\mathcal{M} = (\rho,\Set{\Lambda_x}_{x\in \XX},\Set{M_a}_{a\in \AA})$, as the maximum among the infidelity of the states, defined with standard Uhlmann-fidelity, the average gate infidelities for the gates, and the total variation distance between the measurements, globally 
minimized over unitary equivalence:
\begin{equation}\label{eq:distNM}
    \dist(\mathcal{M},\mathcal{N}) \coloneqq \min_{\mathcal{N'}\sim_u\mathcal{N}}\dist^*(\mathcal{M},\mathcal{N'}),
\end{equation}
with $\dist^\ast(\mathcal{M},\mathcal{N}):=\max\{1-\F(\tilde\rho,\rho),\max_x (1-\F_{\avg}(\tilde\Lambda_x,\Lambda_x)),\d_{TV}(\{M_a\},\{\tilde M_a\})\}$. 
Assuming the independence of repetitions, the probability of \cref{protocol} to output \vbt{accept} can be bounded as $\PP(\text{\vbt{accept}}) \leq  (1-\PP(\text{\vbt{fail}}))^N$, where $\PP(\text{\vbt{fail}}) = \frac{1}{\abs{\mathcal{X}}}\sum_{\vec{x}\in\mathcal{X}}\PP(a\notin a_\mathcal{M}(\vec{x})\vert \vec{x})$ is the minimum probability of an implemented model to fail a single repetition of \cref{protocol} among $N$ repetitions. Then, following an approach in \cite{vanDam2007self-testing}, we obtain a general robustness result.

\begin{theorem}\label{th:robustness}
    The \ac{QSQ} protocol is robust in the sense that there exist integer $k_1,k_2\geq 1$, and real $C_1,C_2>0$, such that the following holds
    \begin{equation}\label{eq:robust}
    \begin{split}
       C_1\dist(\mathcal{N},\mathcal{M})^{k_1}\leq \PP(\text{\textnormal{\vbt{fail}}})\leq C_2\dist(\mathcal{N},\mathcal{M})^{\frac{1}{k_2}}.
    \end{split}    
    \end{equation}
\end{theorem}

See \cref{app:numerics} for a proof. In Ref.~\cite{noller2025classical}, we proved the bound in ~\cref{eq:robust} with $k_1=1$ for the model $\mathcal{S}_1$.
Our numerical investigation suggests that this linear scaling in the robustness also holds for the models $\mathcal{S}_n$ (see \cref{fig:numSn} and \cref{fig:S2num}) for $n=2,3,4$, and $\mathcal{C}l_2$ (see \cref{fig:num}).

\subsection{Numerical investigation}
We begin by giving more details about the general approach. 
The general procedure we follow is to sample a noisy model $\mathcal{N}$, calculate its probability $\PP(\text{\vbt{pass}})$ to pass a single repetition of \cref{protocol}, and then optimize over the unitary gauge to find the infidelity with respect to the target model $\mathcal{M}$, $\dist(\mathcal{N},\mathcal{M})$, defined in \cref{eq:distNM}.
For each sampled model, we add a point to a plot with the axes representing the calculated infidelity and the probability $\PP(\text{\vbt{fail}}) = 1-\PP(\text{\vbt{pass}})$ of the model failing a single round.
Such scatter plots allow us to observe the worst-case behaviors of robustness.
In the small-error regime, the region occupied by these points is well approximated by two linear functions, one providing a lower bound and the other an upper bound. The noise model that we consider in our numerical investigation is the following:\\

\noindent\emph{State:} The initial quantum state is subject to white noise,
    \begin{equation}
        \tilde\rho = (1-\eta)\rho+\eta\frac{\1}{2^n},
    \end{equation}
    where $\eta\in [0,1]$ is the noise parameter.\\
    
\noindent\emph{Measurement:} The measurement is subject to a statistical read-out error
    \begin{equation}
        \tilde M_a = \bigotimes_{i=1}^n \Big((1-\mu_{i})\kb{a_i}{a_i}+\mu_{i}\kb{\bar{a}_i}{\bar{a}_i}\Big),
    \end{equation}
    for $a\in \{0,1\}^n$, $\mu_i\in [0,1]$ the noise parameters, and where $\bar{a}_i$ corresponds to the opposite value of the bit $a_i$.\\
    
    \noindent\emph{Channels:} Each gate $\Lambda_x$ is subject to a coherent error and depolarizing noise
    \begin{equation}
        \tilde \Lambda_x(\argdot) = (1-\chi_x)\tilde U_x(\argdot) \tilde U_x^\dagger + \chi_x \frac{\1}{2^n}\Tr[\argdot],
    \end{equation}
    for $x\in \XX$, where $\chi_x\in [0,1]$ is the depolarizing noise parameter, and $\tilde U_x$ is a unitary close to $U_x$.
    
\begin{figure}
    \centering
    \includegraphics[width=\linewidth]{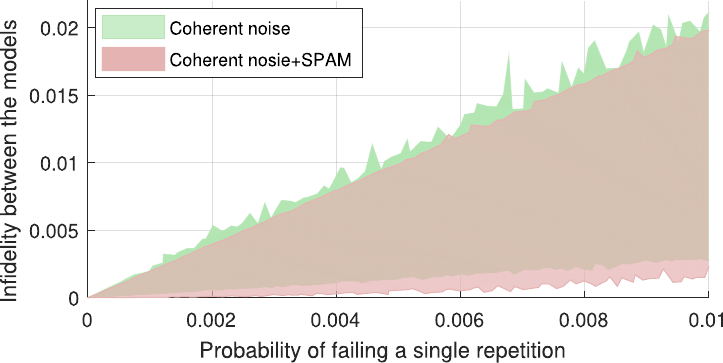}
    \caption{Numerics for robustness for the $\mathcal{S}_2$-model.
    Infidelity between the implemented and target models is plotted against the probability of the former failing a single repetition of \cref{protocol}.
    In the region with higher infidelity values, \ac{SPAM} errors are set to $0$. In the region with lower infidelity values, initial state and measurement are affected by depolarizing noise and statistical errors, respectively, both up to $5\%$. 
    Each region contains $\approx 2\times 10^6$ randomly generated models.
    } 
    \label{fig:S2num}
\end{figure}

If we were to sample the parameters $\eta,\mu_i,\chi_x$, or $\tilde U_x$ uniformly (w.r.t. the Haar measure for the unitary), we would obtain mostly points on the plot that have both high $\PP(\text{\vbt{fail}})$ and high $\dist(\mathcal{N},\mathcal{M})$.
In order to observe the behavior in the low-noise regime, we sample these parameters uniformly in the restricted regions.
For the parameters $\eta,\mu_i,\chi_x$, this is rather trivial to achieve by simply sampling them from some interval $[0,\veps_\mathrm{max}]$.
For $\tilde U_x$, we do the following. If $U_x = V\diag(\lambda_1,\lambda_2,\dots, \lambda_{2^n})V^\dagger\in \U(\HH)$ is the target unitary, then (for each $x$ independently) we draw $\tilde\lambda_i\in \CC$, with $\abs{\tilde\lambda_i}=1$ close to $\lambda_i$, and $W\in \U(\HH)$ close to the identity, and take $\tilde U_x = WV\diag(\tilde\lambda_1,\tilde\lambda_2,\dots, \tilde\lambda_{2^n})V^\dagger W^\dagger$.
These two types of coherent noises represent the over-rotation and misalignment errors. 
Additionally, we define parameters that control the distribution of the values of the above parameters over $x\in \XX$ to account for the edge cases when only certain gates are affected by the noise.
Finally, when calculating the infidelity between the models as given in \cref{eq:distNM}, we optimize (with the gradient descent method) only over the unitary matrices $U$ that are diagonal in the computational basis to avoid the problem of the optimization getting stuck in local minima.
More technical details on the implementation can be found in the programs provided online~\cite{git}.

\begin{figure}
    \centering
    \includegraphics[width=\linewidth]{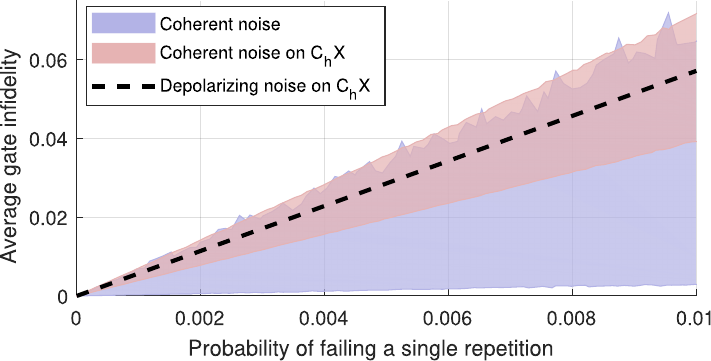}
    \caption{Numerics for robustness for the $\mathcal{C}l_2$-model.
    The worst average gate infidelity among the implemented gates is plotted against the probability of the corresponding noisy model failing a single repetition of \cref{protocol}.
    \ac{SPAM} errors are set to $0$.
    The larger region corresponds to coherent noise applied to all the gates, and contains $\approx 2\times 10^6$ random models, while the smaller region represents the case of coherent noise affecting only the $C_\iH X$ gate and contains $\approx 7\times 10^6$ randomly drawn models.
    The dashed line represents the case where only the $C_\iH X$ gate is affected by depolarizing noise.
    } 
    \label{fig:num}
\end{figure}

The results of our numerical investigation of the robustness for the models $\mathcal{S}_n$ for $n=2,3,4$ are summarized in \cref{fig:numSn}. Additionally, we investigate the two-qubit Clifford model, in particular, the entangling gate in \cref{fig:num}. All plots indicate that, also in the multi-qubit case, we recover the linear robust soundness (i.e., having $k_1=1$ in \cref{th:robustness}) as the analytical bounds obtained in \ac{QSQ}-protocol for  $\mathcal{S}_1$ model~\cite{noller2025classical}.  
In  \cref{fig:S2num}, we analyze the interplay between errors affecting the different components of $\mathcal{S}_2$-model in more detail. Here, the \ac{SPAM} noise only increases the falling probability until it is high enough so that the maximum in \cref{eq:distNM} corresponds to the infidelity of the state or the variational distance for the measurement.
We also observe that the coherent noise corresponding to the over-rotation and misalignment errors leads to worse behavior in terms of soundness than the depolarizing noise in terms of the ratio $\dist(\mathcal{N},\mathcal{M})/\PP(\text{\vbt{fail}})$.
For this reason, we set $\eta=0$, $\mu_i=0$, $\forall i\in \{1,2,\dots, n\}$, and $\chi_x=0$, $\forall x\in \XX$ in the numerical simulations that we present in \cref{fig:numSn}, in order to detect the worst behavior for the robustness of soundness.

\section{Discussions \& outlook}
Quantum certification has traditionally been driven by practical considerations. Here, we revisit the problem from a foundational perspective: How can one certify that a quantum computer implements the intended operations while ruling out all erroneous implementations?  We present the first provably sound protocol for certifying the internal operations of multi-qubit models without assuming trusted \ac{SPAM} pre-calibration. Being platform agnostic, our protocol provides a new paradigm for testing quantum computing architectures and opens multiple directions for future work in quantum system certification and characterization.

This work invites further research within the \ac{QSQ} framework.
First, novel techniques are required to prove analytically the linear robustness observed in the numerics. Such techniques, potentially drawing on the representation theory of groups, could also provide insight into the relationship between a given gate set and the corresponding set of quizzes.
In this work, we design the set of quizzes constructively. Instead of calculating the probabilities in the output map for an arbitrary combination of gates, we design quizzes in a way that the output set has low cardinality. This is easy to achieve for the local phase gate model for $n$ qubits, and extending it to a universal model only requires a few additional quizzes. While this construction works, further research is needed to understand if it is minimal. 

Second, our proof in the limit of infinite repetition hinges on inferring exact addressability of the subsystems. Although the robustness analysis certifies closeness to the tensor product model, in the finite-statistics regime, a more targeted approach is required to certify \textit{approximate} addressability by deriving bounds on an appropriate quantitative measure of crosstalk, a quantity of central importance for current hardware.

Finally, in analyzing the robustness, it is important to investigate the effect of testing longer sequences, that is, deeper circuits. Increasing the sequence length amplifies gate errors and may, in turn, reduce the number of repetitions of the \ac{QSQ} protocol required for certification. We believe this effect underlies the behavior observed in our numerical investigation, where the coefficient in the robustness bounds appearing in the soundness guarantees improves with the number of qubits. Moreover, the analysis of long sequences can be used to partially identify and separate \ac{SPAM} errors from gate errors, following well-established practices in the literature, but here for the first time in a fully sound manner.

\emph{Acknowledgments --} We thank Tommaso Calarco and Markus Heinrich for inspiring discussions.
This research was funded by the Deutsche Forschungsgemeinschaft
(DFG, German Research Foundation), project numbers 441423094, 236615297 - SFB 1119, the German Federal Ministry of Education and Research (BMBF) within the funding program ``quantum technologies - from basic research to market'' via the joint project MIQRO (grant number 13N15522), and the Fujitsu Germany GmbH as part of the endowed professorship ``Quantum Inspired and Quantum Optimization''.

\onecolumngrid
\newpage

\section*{Appendix}
\begin{appendix}

\onecolumngrid
\setcounter{figure}{2}

In this Appendix, we provide technical details that support the claims in the main text.
In \cref{app:1_qubit}, we state and prove a simplified single-qubit certification result from Ref.~\cite{noller2025classical}, in the case of ideal statistics.
In \cref{app:2_qubits}, we start by giving a detailed proof of the first central result of this paper, \cref{res:S_2}.
In the same section, we identify an input set (\cref{res:Clifford_2}) for the certification of a $2$-qubit quantum model containing a set of gates that generates the $2$-qubit Clifford group.
In \cref{app:n_qubits}, we analyze the general case of $n$ qubits.
First, we prove \cref{res:S_n}, which generalizes \cref{res:S_2} from the main text.
This generalization not only makes the proof more technical but also brings some conceptual challenges connected to having to consider various bipartitions. 
Later in the same section, we explicitly construct input sets for $n$-qubit universal quantum models, to give a proof for \cref{res:universal}.
In \cref{app:numerics}, we give details on the analytical robustness analysis, including a proof of \cref{th:robustness}.
Finally, in \cref{app:lemmas} we present a number of supporting lemmata.

\subsection*{Additional notation}
In this section, we use a common shorthand notation $\psi\coloneqq \kb{\psi}{\psi}$ for rank-1 projectors.
An $n$-fold composition of a channel $\Lambda$ we denote by $\Lambda^n(\argdot)\coloneqq\Lambda\circ\Lambda\circ\dots \Lambda(\argdot)$.
The unitary group of a Hilbert space $\HH$ we denote by $\U(\HH)$.
In the proofs, we need to write many equivalences between state vectors and between unitary operators which hold up to a phase, i.e., a complex factor with absolute value $1$.
For that purpose, we use the notation $\eutp$.
Schatten $p$-norms of linear operators are denoted as $\norm{\argdot}_p$, and the diamond norm as $\norm{\argdot}_\diamond$.

\section{Single-qubit results revisited}\label{app:1_qubit}
In the proofs in this paper, we use the following single-qubit result as a ``subroutine''.

\begin{theorem}\label{res:S_1}
    An input set $\mathcal{X}_1\coloneqq \{\epsilon,\iS^2,\iS^4\}$ certifies the single-qubit quantum model $\mathcal{S}_1$.
\end{theorem}
\begin{proof}
Let $\mathcal{M}\coloneqq (\rho,\{\Lambda_\iS\},\{M_0,M_1\})$ be a $2$-dimensional model defined over $\HH\cong \CC^2$ that is accepted by \cref{protocol} for the inputs $\mathcal{X}_1$.
We begin the proof by observing that instructions $\epsilon$ and $\iS^2$ translate to the conditions $\tr[M_0\rho]=\tr[M_1 \Lambda_\iS^2(\rho)]=1$. These necessitate that $M_0=\rho=\psi$, where $\psi$ is a pure state, and that $\Lambda_\iS^2(\psi)=M_1=\psi^\perp$, where by $\ket{\psi^\perp}$ we denote a vector orthogonal to $\ket{\psi}$, i.e., $\bk{\psi^\perp}{\psi}=0$.
From the instruction $\iS^4$, we obtain that $\Lambda_\iS^4(\psi)=\psi$, and, hence, we can infer that $\Lambda_\iS^2(\psi^\perp)=\psi$.
Therefore, we have
\begin{equation}\label{app:eq:proof_S_1_1}
    \Lambda_\iS\circ\Lambda_\iS(\psi) = \psi^\perp, \quad \Lambda_\iS\circ\Lambda_\iS(\psi^\perp) = \psi. \quad
\end{equation}
We notice that $\Lambda_\iS$ maps two states, $\Lambda_\iS(\psi)$ and $\Lambda_\iS(\psi^\perp)$, to two orthogonal ones, and hence, these states must be orthogonal themselves (see \cref{lemma:orthogonality} for a short proof), which also implies that they are pure, since the system we are dealing with is $2$-dimensional.

Let $\phi\coloneqq \Lambda_\iS(\psi)$, and $\phi^\perp=\Lambda_\iS(\psi^\perp)$.
It is not hard to see that $\phi\notin\{\psi,\psi^\perp\}$ as either of the cases would lead to a contradiction with \cref{app:eq:proof_S_1_1}. 
Therefore, $\ket{\phi}$ is coherent in the basis $\{\ket{\psi},\ket{\psi^\perp}\}$ and since $\Lambda_\iS(\phi)=\psi^\perp$ is a pure state, and $\Lambda_\iS$ maps an \ac{ONB} to an \ac{ONB}, we conclude that the channel $\Lambda_\iS$ must be unitary (see \cref{corr:unitarity}). 

With this conclusion, we can set $\Lambda_\iS(\argdot)\eqqcolon U_\iS(\argdot)U_\iS^\dagger$, for a unitary operator $U_\iS$ to be identified.
We have already established that $\Lambda_\iS^4(\psi)=\psi$ and $\Lambda_\iS^4(\psi^\perp)=\Lambda_\iS^2(\psi) = \psi^\perp$. 
We can additionally see that $\Lambda^4_\iS(\phi)=\Lambda^5_\iS(\psi)=\Lambda_\iS\circ \Lambda^4_\iS(\psi)=\Lambda_\iS(\psi)=\phi$.
In other words, $\Lambda^4_\iS$ acts like the identity channel on an \ac{ONB} and a state which is coherent with respect to it, which lets us conclude that $U_\iS ^4 \eutp \openone$ (see \cref{lemma:channel=U}). Let $U_\iS$ be diagonalized as
\begin{equation}\label{app:eq:proof_S_1_2}
    U_\iS \eutp V (\ketbra{0}{0} + \e^{\ii\alpha}\ketbra{1}{1}) V^\dagger,
\end{equation}
up to some global phase, for some unitary $V$ and $\alpha\in [0,2\pi)$. 
From $U_\iS ^4 \eutp \1$, it follows that $\alpha=\frac{k\pi}{2}$, for some $k\in [4]$. 
Since $\ket{\psi} \notin \{U_\iS \ket{\psi},U_\iS^2 \ket{\psi}\}$, we conclude that $k\notin\{0,2\}$ and, thus, either $k=1$ or $k=3$. 
This already establishes that $U_\iS$ is the desired unitary $S = \kb{0}{0}+\ii\kb{1}{1}$, up to the unitary gauge $V$ and possible complex conjugation. 
However, to obtain the model $\mathcal{S}_1$, we still need to adjust the gauge $V$, such that it maps $\kb{+}{+}$ to $\psi$. For either of the cases $k=1$ or $k=3$, we get that $U_\iS^2 \eutp V Z V^\dagger$, and since $U_\iS^2\ket{\psi} = \ket{\psi^\perp}$, we have
$Z V^\dagger \ket{\psi} = V^\dagger\ket{\psi^\perp}$, and, therefore, 
\begin{equation}
V^\dagger\ket{\psi} 
\eutp \frac{1}{\sqrt{2}}\left(\ket{0}+\e^{\ii \theta} \ket{1}\right),
\end{equation}
for some $\theta \in \RR$. 
In the case of $k=1$, we can set $U=V (\kb{0}{0}+\e^{\ii\theta}\kb{1}{1})$ as the new unitary gauge, which clearly does not affect the equivalence in \cref{app:eq:proof_S_1_2}, but achieves $U^\dagger\ket{\psi}\eutp\ket{+}$, thus establishing equivalence of model $\mathcal{M}$ to the model $\mathcal{S}_1$. 
For $k=3$, we set $U=V X(\kb{0}{0}+\e^{\ii\theta}\kb{1}{1})$, with the additional $X$-gate achieving the mapping $S^\dagger\mapsto S$, and eliminating the need to additionally apply the complex conjugation.
This completes the proof. 
\end{proof}
\Cref{res:S_1} is a simplified version of the main result of Ref.~\cite{noller2025classical} proven for a model which includes only the $S$-gate, as opposed to two gates, $S$ and $S^\dagger$, considered there, and it only considers the ideal case.
Note, that in the above proof, we did not need to consider anti-unitary gauge transformations.
Instead, we could resolve the situation which seem to require complex conjugation by applying an additional $X$-gate.
This can be seen as a special instance of the following observation that will be useful later.
\begin{observation}\label{obs:S_1_gauges}
    If two $2$-dimensional quantum models $(\ket{\psi},\{U_\iS\},\{\psi,\psi^\perp\})$ and $(\ket{\psi},\{\tilde U_\iS\},\{\psi,\psi^\perp\})$, with the measurement in the same basis, are both equivalent to $\mathcal{S}_1$, 
    with the corresponding gauge unitaries being $V$ and $\tilde V$, respectively, then $V^\dagger\tilde{V} \eutp X^r\coloneqq\kb{+}{+}+\e^{\ii\pi r}\kb{-}{-}$ for some $r\in\RR$.
\end{observation}
Indeed, it must be that $V\eutp \kb{\psi}{+}+\e^{\ii r_1}\kb{\psi^\perp}{-}$ and $\tilde V\eutp \kb{\psi}{+}+\e^{\ii  r_2}\kb{\psi^\perp}{-}$ for some $r_1,r_2\in\RR$. 
This amounts to $V^\dagger\tilde V \eutp \kb{+}{+}+\e^{\ii(r_2-r_1)}\kb{-}{-}$.

\section{The two-qubit results}
\label{app:2_qubits}
We start by giving a detailed proof of \cref{res:S_2} from the main text.

\begin{reptheorem}{res:S_2}[restated]
    The two-qubit quantum model $\mathcal{S}_2$ is certified by an input set
    \begin{equation}\label{app:eq:2_qubits_inputs}
        \mathcal{X}_2\coloneqq\mathcal{X}_a\cup\mathcal{X}_b\cup\Set*{(\iS_a\iS_b)^2,(\iS_b\iS_a)^2},
    \end{equation}
    where $\mathcal{X}_a\coloneqq\Set{\iS_b^{2j}\iS_a^{i}\given j\in \{0,1\},i\in [5]}$ and $\mathcal{X}_b \coloneqq \Set{\iS_a^{2i}\iS_b^{j}\given i\in\{0,1\},j\in[5]}.$
\end{reptheorem} 
\begin{proof}
Let $\mathcal{M}\coloneqq (\rho,\{\Lambda_{\iS_a},\Lambda_{\iS_b}\},\{M_{ij}\}_{i,j\in\{0,1\}})$ be a quantum model defined over $\HH\cong \CC^4$ satisfying $a_\mathcal{M}(\vec{x})\subseteq a_{\mathcal{S}_2}(\vec{x})$ for all $\vec{x}\in \mathcal{X}_2$.
For brevity, let us denote $\Lambda_a \coloneqq \Lambda_{\iS_a}$ and $\Lambda_b \coloneqq \Lambda_{\iS_b}$.
The proof starts similarly as in \cref{res:S_1}.
From $a_\mathcal{M}(\iS_b^{2j}\iS_a^{2i})\subseteq a_{\mathcal{S}_2}(\iS_b^{2j}\iS_a^{2i})=\{ij\}$, for $i,j\in\{0,1\}$, we find that the four states $\Lambda_a^{2i}\circ\Lambda_b^{2j}(\rho)$ are distinct and perfectly distinguishable by the measurement. 
Therefore, we conclude that these states, and the corresponding measurement effects $M_{ij}$, are rank-1 projector, which we denote as $\psi_{ij}\coloneqq \kb{\psi_{ij}}{\psi_{ij}}$, and moreover, $\Set{\ket{\psi_{ij}}}_{i,j\in\{0,1\}}$ is an \ac{ONB}.
Next, we consider the input sequences $\iS_a^{2i}\iS_b^{2j}$, for $i,j\in\{0,1\}$ which give the same outputs as $\iS_b^{2j}\iS_a^{2i}$, and, thus, the corresponding states should also be equal to $\psi_{ij}$.
We can, therefore, conclude that 
\begin{equation}\label{app:eq:proof_S_2_exchange_AB}
  \Lambda_a^{2i}\circ\Lambda_b^{2j}(\rho_0)=\Lambda_b^{2j}\circ\Lambda_a^{2i}(\rho_0)=\psi_{ij}, \quad i,j\in \{0,1\}.  
\end{equation}
Clearly, from the empty input string we find that $\rho=\psi_{00}$. 

In the next step, we introduce a bipartite tensor product structure on $\mathcal{H}$ in the following way.
Consider two Hilbert spaces $\HH_A\cong \CC^2$ and $\HH_B\cong\CC^2$, each with an \ac{ONB}
$\{\ket{\psi_i}\}_{i\in\{0,1\}}$.
Let us consider the unitary
\begin{equation}\label{app:eq:S_2_u_otimes}
U_\otimes\coloneqq \sum_{i,j\in \{0,1\}}\ket{\psi_i}_A\otimes\ket{\psi_j}_B\bra{\psi_{ij}}.
\end{equation}
In the following, we omit the subscripts for the subsystems $A$ and $B$, and keep their order in the tensor product.
If we apply $U_\otimes$ to all the elements of the model $\mathcal{M}$, it will map the states $\psi_{ij}$, and thus also the measurement effects $M_{ij}$ and the states in \cref{app:eq:proof_S_2_exchange_AB}, to the product states $\psi_i\otimes\psi_j$ of a bipartite system. 
Since we did not yet make any characterization of the channels $\Lambda_a$ and $\Lambda_b$, we will simply assume from now on that the gauge $U_\otimes$ is already applied to the model $\mathcal{M}$ and proceed working in the convention $\psi_{ij}=\psi_i\otimes\psi_j$. 

We have established the bipartition on the level of the measurement and some of the attainable states.
However, propagating this structure to the channels, i.e., proving addressability of subsystems, is much more challenging.
For now, we focus on the channel $\Lambda_{a}$, as the other case then just follows from exchanging the roles of $A$ and $B$ everywhere. 
First, we need to move away from the purely deterministic protocol to \emph{quasi-deterministic},
and allow for ``intermediate'' steps, where we require only the marginal probabilities to be deterministic.
In other terms, in \cref{protocol}, we check only one of the bits ($i$ or $j$) of the outcome $ij$, to be correct, where the other outcome can be arbitrary.
This analysis -- which we conduct rigorously below -- ensures, that the mapping of channels to subchannels is a homomorphism, such that we can reduce the problem to the two-dimensional subchannels in the first step.

We consider the inputs $\iS_b^{2j}\iS_a^{2i+1}$ from $\mathcal{X}_a$ for $i,j\in\{0,1\}$, for which the condition on passing \cref{protocol} reads $a_\mathcal{M}(\iS_b^{2j}\iS_a^{2i+1}) \subseteq \{0j,1j\}$, which, in turn, implies that
\begin{align}
\begin{split}
1=&\ \tr[(M_{0j}+M_{1j})\Lambda_a^{2i+1}\circ \Lambda^{2j}_b(\psi_{00})] =  \tr[(M_{0j}+M_{1j})\Lambda_a(\psi_{ij})] \\
=&\ \tr[(\psi_0\otimes\psi_j + \psi_1\otimes\psi_j)\Lambda_a(\psi_{ij}) ] = \tr[\openone\otimes\psi_j\Lambda_a(\psi_{ij})]=\tr[\psi_j\tr_A[\Lambda_a(\psi_{ij})]].
\end{split}
\end{align}
Since $\Lambda_a(\psi_{ij})$ is normalized, we can conclude that  
\begin{align}\label{app:eq:res_S_2_basis_cond}
 \tr_A[\Lambda_{a}(\psi_i\otimes\psi_j)] =\psi_j\;\text{ for all } i,j\in \{0,1\}.
\end{align}
At this point, we need the concept of subchannels, as given by \cref{def:subchannels}.
From \cref{app:eq:res_S_2_basis_cond} and \cref{lemma:marginal_homomorphism} we can infer that the mapping of channels to subchannels, taken with respect to the basis states $\psi_j$ of the subsystem $B$, is a homomorphism. 
In particular, we have 
\begin{equation}\label{app:eq:res_S_2_subchannel_a}
\Lambda_{a}(\rho\otimes \psi_j) = \Lambda_{a}\vert_A^{j}(\rho)\otimes\psi_j,
\end{equation}
for all $\rho\in \DM(\HH_A)$, $j\in\{0,1\}$, and again from \cref{lemma:marginal_homomorphism}, the above equation hold for any ``power'' of $\Lambda_a$, i.e., composition of the channel $\Lambda_a$ with itself. 
This allows us to study the action of the channel $\Lambda_a$ on the subsystem $A$, as long as the subsystem $B$ remains in the state $\psi_j$.

In particular, for fixed $j\in\{0,1\}$, we reduce the problem to the certification of the reduced single-qubit models $(\psi_0,\{\Lambda_a\vert_A^{j}\},\{M_0,M_1\})$, with $M_i = \psi_i$, for $i\in\{0,1\}$. 
Indeed, the input sequences $\iS_b^{2j}\iS_a^{2i}$ for $i\in\{0,1,2\}$ correspond to the sequences $\{\epsilon,\iS^2,\iS^4\}$ for certification of the single-qubit quantum model $\mathcal{S}_1$ in \cref{res:S_1}.
The proof of \cref{res:S_1} can then be repeated for each of the subchannels $\Lambda_a\vert_A^{j}$, where all the elements of the model are in a tensor-product with the state $\psi_j$ of the subsystem $B$ as in \cref{app:eq:res_S_2_subchannel_a}. 
This allows us to conclude that for each $j\in\{0,1\}$ the model $(\psi_0,\{\Lambda_a\vert_A^{\psi_j}\},\{M_0,M_1\})$ is equivalent to $\mathcal{S}_1$.

All the previous steps can be repeated for $\Lambda_b$ upon exchanging the roles of $A$ and $B$, which is possible because all the input sequences are symmetric with respect to the exchange of the subsystems, and due to \cref{app:eq:proof_S_2_exchange_AB}. 
We conclude that the subchannels of both $\Lambda_a,\Lambda_b$ are unitary, i.e., there exist unitary operators $U_{a,j}$ and $U_{b,i}$, such that 
\begin{equation}
    \Lambda_a\vert_A^{\psi_j}(\argdot) = U_{a,j}(\argdot)U_{a,j}^\dagger,\quad \Lambda_b\vert_B^{\psi_i}(\argdot) = U_{b,i}(\argdot)U_{b,i}^\dagger,
\end{equation}
for $i,j\in\{0,1\}$.
Moreover, there exist unitaries $V_j\in \U(\HH_A)$, for $j\in\{0,1\}$ and $W_i\in\U(\HH_B)$, for $i\in\{0,1\}$, such that
\begin{align}\label{app:eq:certified_subchannels}
\begin{split}
    \psi_0 = V_j\kb{+}{+}V_j^\dagger,& \qquad U_{a,j} \eutp V_j S V_j^\dagger,\quad \forall j\in \{0,1\},\\ 
    \psi_0  = W_i \kb{+}{+}W_i^\dagger,& \qquad U_{b,i} \eutp W_i S W_i^\dagger,\quad \forall i \in \{0,1\}. 
    \end{split}
\end{align}

From \cref{obs:S_1_gauges}, we additionally know that $V_0^\dagger V_1 \eutp X^r$ and $W_0^\dagger W_1 \eutp X^s$, for some $r,s\in [0,2]$.

The unitarity of the subchannels is not sufficient for proving unitarity of the global channel. 
Channel $\Lambda_{a}$, for example, could act decoherently on the subsystem $B$, e.g., we could have
$\Lambda_a(\argdot)=\sum_{j\in\{0,1\}} (\psi_j\otimes U_{a,j})(\argdot)(\psi_j\otimes U_{a,j}^\dagger)$.
To exclude such a possibility, we need to know how $\Lambda_a$ acts on a state, in which the subsystem $B$ is in a coherent superposition of the basis states $\psi_j$. 
If there is any decoherence caused by $\Lambda_a$ on $B$, the resulting state would be mixed.
This intuition is made rigorous by  \cref{lemma:unitarity_from_subchannels}. 
In order to prove the unitarity of $\Lambda_{a}$ it suffices to check the purity of the state 
\begin{equation}\label{app:eq:S_2_proof_ab_psi00}
    \Lambda_{a}\circ\Lambda_{b}(\psi_{00})=\Lambda_{a}(\psi_0\otimes U_{b,0}\psi_0 U_{b,0}^\dagger),
\end{equation} 
as the state on subsystem $B$ satisfies the coherence requirement
$\abs{\bra{\psi_0} U_{b,0}^\dagger\ket{\psi_0}} = \abs{\bra{+}S^\dagger \ket{+}}=\frac{1}{\sqrt{2}}>0.$
Since each $\Lambda_{a}$ and $\Lambda_{b}$ maps an \ac{ONB} to an \ac{ONB}, neither of them can increase the purity of any quantum state, according to \cref{lemma:ONB_channel_purity}. 
At the same time, $\Lambda_{a}\circ\Lambda_{b}\circ\Lambda_{a}\circ\Lambda_{b}(\psi_{00})$ is a pure state, which follows from a deterministic outcome of the input sequence $(\iS_b\iS_a)^2$.
Thus, we conclude that the state in \cref{app:eq:S_2_proof_ab_psi00} must also be pure, and by virtue of 
 \Cref{lemma:unitarity_from_subchannels}, the channel $\Lambda_a$ must be unitary.
From the same lemma, we also know that if $U_a$ is the corresponding unitary operator, it can be represented as
\begin{align}
    \label{eq:preliminary_S_A}U_{a} &= V_0 S V_0^\dagger \otimes\psi_0 + \e^{\ii\pi\alpha}V_1 S V_1^\dagger \otimes\psi_1.
\end{align}
for suitable $\alpha\in \RR$. 
By switching the roles of $A$ and $B$ in the above argumentation, we find that
\begin{align}\label{app:eq:S_2_Ub_subchannels}
    U_{b} &= \psi_0\otimes W_0 SW_0^\dagger + \psi_1\otimes \e^{\ii\pi\beta} W_1 S W_1^\dagger,
\end{align}
for some $\beta\in\RR$, where $U_b$ is the unitary operator of the channel $\Lambda_b$.
Let $\Theta \coloneqq V_{0}^\dagger\otimes W_0^\dagger$, to which we refer as a preliminary gauge unitary.
Since we know that $V_0^\dagger V_1 \eutp X^r$ and $W_0^\dagger W_1\eutp X ^s$ for some $r,s\in[0,2]$, we find that
\begin{align}\begin{split}\label{app:eq:S_2_theta_UaUb}
    \Theta U_{a}\Theta^\dagger &= S \otimes \kb{+}{+} + \e^{\ii\pi\alpha}X^r S X^{-r}\otimes\kb{-}{-},\\
    \Theta U_{b}\Theta^\dagger  & = \kb{+}{+}\otimes S+\e^{\ii\pi\beta}\kb{-}{-}\otimes X^sS X^{-s},
\end{split}\end{align}
where we used the relations from \cref{app:eq:certified_subchannels}.
Now we modify the preliminary gauge by additionally applying $C_\iH X^{-s} = \kb{+}{+}\otimes\1+\kb{-}{-}\otimes X^{-s}$, i.e., we set $\Theta \coloneqq C_\iH X^{-s}(V_0\otimes W_0)^\dagger$.
With the modified gauge we achieve
\begin{align}\begin{split}\label{app:eq:S_2_proof_semifinal}
    \Theta  U_{a} \Theta^\dagger&= S\otimes\kb{+}{+} + \e^{\ii\pi\alpha} X^{t}SX^{t}\otimes\kb{-}{-},\\
    \Theta U_{b}\Theta^\dagger  &= X^{\beta} \otimes S,
\end{split}\end{align}
where we set $t\coloneqq r-s$. 
At this point, we have exhausted our possibilities to structurally simplify the problem by modifying the gauge freedom. 
In order to show that the remaining undetermined parameters vanish, namely that $t=\alpha=\beta=0$, we use the deterministic output of the input sequences in which we transition through states where neither of the two subsystems $A$ or $B$ is in one of the reference basis states.
Specifically, we consider the sequence $(\iS_a\iS_b)^2$, for which the deterministic outcome $a_\mathcal{M}\big((\iS_a\iS_b)^2\big)=\{11\}$ dictates that
\begin{equation}
\begin{tikzcd}[row sep = 4em, column sep = 4em]
    \ket{\psi_{00}}\arrow{r}{U_{b}U_{a}}  &\ket{\phi}\arrow{r}{U_{b}U_{a}} 
    &\ket{\psi_{11}},
\end{tikzcd}\end{equation}
for some appropriate $\ket{\phi}$.
From the above diagram, we find $U_{b} U_{a}\ket{\psi_{00}}\eutp (U_{b} U_{a})^\dagger \ket{\psi_{11}}$ must hold. 
On the other hand, the product $U_{b}U_{a}$, to which the gauge $\Theta$ is applied is given by
\begin{equation}
\Theta U_{b} U_{a} \Theta^\dagger =
X^\beta S\otimes S\kb{+}{+}+ \e^{\ii\pi\alpha}X^{\beta + t}S X^{-t}\otimes S\kb{-}{-}.
\end{equation}
Since $\Theta^\dagger\ket{++}\eutp \ket{\psi_{00}}$ and $\Theta^\dagger\ket{--}\eutp \ket{\psi_{11}}$, we must have
\begin{equation}\label{app:eq:S_2_proof_UbUa_eq}
\Theta U_{b} U_{a} \Theta^\dagger\ket{\mathrm{++}} \eutp \Theta(U_{b} U_{a})^\dagger\Theta^\dagger \ket{\mathrm{--}}.
\end{equation}
The left-hand side of \cref{app:eq:S_2_proof_UbUa_eq} can be easily found to be $\Theta U_{b} U_{a} \Theta^\dagger\ket{\small{++}} \eutp X^{\beta}\ket{+_y}\otimes \ket{+_y}$, and the right-hand side is 
\begin{align}
\begin{split}
\Theta(U_{b} U_{a})^\dagger\Theta^\dagger \ket{\small{--}}&\eutp  \left( S^\dagger X^{-\beta}\otimes\kb{+}{+}S^\dagger + \e^{-\ii\pi\alpha} X^t S^\dagger X^{-(\beta+t)}\otimes \kb{-}{-} S^\dagger\right) \ket{\small{--}}\\
&= \frac{1}{\sqrt{2}}\e^{-\ii\pi\beta}\left(\e^{\frac{\ii\pi}{4}}\ket{+_y}\otimes\ket{+} + \e^{-\ii\pi(\alpha+t+\frac{1}{4})} X^t \ket{+_y}\otimes \ket{-}\right).
\end{split}
\end{align}
Equating the both sides of \cref{app:eq:S_2_proof_UbUa_eq}, and multiplying the equation from the left first by $\openone\otimes\bra{+}$ and then again by $\openone\otimes\bra{-}$, we obtain that 
\begin{equation}
    X^{\beta}\ket{+_y} \eutp X^t\ket{+_y}\eutp\ket{+_y}.
\end{equation}
Clearly, this is satisfied only if $X^\beta =X^t=\openone$. 
Taking this into account, the condition \cref{app:eq:S_2_proof_UbUa_eq} simplifies as
\begin{equation}
    \ket{+_y}\otimes\ket{+_y} \eutp \ket{+_y}\otimes\frac{1}{\sqrt{2}}\left(\ket{+}+\e^{-\ii\pi(\alpha+\frac{1}{2})}\ket{-}\right),
\end{equation}
which only holds if $\alpha=0$.
By inserting $\alpha=\beta=t=0$ in \cref{app:eq:S_2_proof_semifinal} we get that
$\Theta U_{a}\Theta^\dagger = S\otimes \openone,\;\Theta U_{b}\Theta^\dagger = \openone\otimes S$,
which concludes the proof.
\end{proof}

The generalization of the above result to $n$-qubits is presented in \cref{app:n_qubits}. 
Next, we prove a result on the certification of a $2$-qubit model with gates generating the Clifford group. 
This result comes after, and partially builds on, \cref{res:C_hX}, which we present in the main text.

\begin{theorem}\label{res:Clifford_2}
The quantum model $\mathcal{C}l_2=\mathcal{S}_2+\{H^{(1)}\} + \{C_\iH X\}$ is certified by the following input set: 
$\mathcal{X}_{Cl}=\mathcal{X}_{2}\cup\mathcal{X}_\iCX\cup \mathcal{X}_\iH $, where $\mathcal{X}_2$ is defined in \cref{app:eq:2_qubits_inputs}, and 
\begin{align}\label{eq:H_inputs}
    \mathcal{X}_{\iCX}& \coloneqq\bigcup_{i,j\in\{0,1\}}\Set{\iS_a^{2i}\iS_b^{2j}\iCX }\cup\Set{\iS_b\iCX\iS_b,\iS_a\iCX\iS_a,\iS_a\iS_b^2\iCX\iS_a},\\
    \mathcal{X}_\iH & \coloneqq \Set*{\iS_b^{2i}\vec{x} \given \vec{x}\in\{\iH,\iS_a^2\iH,\iH\iH,\iS_a\iH\iS_a,\iS_a^3\iH\iS_a,\iH\iS_a\iH\},\,i\in\{0,1\}}\cup\Set*{\iS_a\iS_b\iH^{j}\iS_a\iS_b\given j\in\{1,2\}}.
\end{align}
\end{theorem}
Note that since we have a Hadamard gate only on the first qubit in the target model, we omit a subscript for the corresponding instruction ``$\iH$''.
\begin{proof}
Let $\mathcal{N}\coloneqq (\rho,\{\Lambda_{\iS_a},\Lambda_{\iS_b},\Lambda_{\iCX},\Lambda_{\iH}\},\{M_{ij}\}_{ij})$ be a 4-dimensional quantum model defined over $\HH\cong\CC^4$ satisfying $a_{\mathcal{N}}(\vec{x}) \subseteq a_{\mathcal{C}l_2}(\vec{x})$ for all $\vec{x}\in\mathcal{X}_{Cl}.$
We begin with the observation that the inputs $\mathcal{X}_{2}\cup\mathcal{X}_{\iCX}$ certify $\mathcal{S}_2+\{C_\iH X\}$ due to \cref{res:C_hX}. We then may directly assume that $\mathcal{N}$ is of the form $\mathcal{N}=\mathcal{S}_2+\{C_\iH X\}+\{\Lambda_\iH\}$.
In particular, when characterizing $\Lambda_{\iH}$, we assume that the unitary gauge freedom in certification of the submodel $\mathcal{S}_2+\{C_\iH X\}$ is already exploited, i.e., we have two Hilbert spaces $\HH_A\cong \CC^2$ and $\HH_B\cong\CC^2$, and $M_{ij}$ is the tensor product of the two $X$-basis measurement on both qubits and $\rho = \kb{\mathrm{++}}{\mathrm{++}}$.

Since the eigenbasis of $H^{(1)}$ is not contained in the attainable states of the model $\mathcal{S}_2+\{C_\iH X\}$, we cannot simply apply \cref{cor:adding_inputs} to include the Hadamard gate to this model. 
Instead, we proceed the same way as in the proof of the $\mathcal{S}_2$ model, with many steps being completely analogous. 
First of all, we use the instructions $\{\iS_b^{2i}\iS_a^{2j}\iH\}_{i,j\in\{0,1\}}\subseteq\mathcal{X}_{\iH}$ to prove that $\Lambda_{\iH}(\rho_A\otimes\kb{\pm}{\pm}) = \Lambda_\iH\vert_A^{\scaleobj{.8}{\ket{\pm}}}\otimes\kb{\pm}{\pm}$ holds. 
Then we can identify the subchannels $\Lambda_\iH\vert_A^{\scaleobj{.8}{\ket{\pm}}}$ using the instructions $\iS_b^{2i}\{\iH\iH,\iS_a\iH\iS_a,\iS_a^3\iH\iS_a,\iH\iS_a\iH\}$ building on the results of Ref.~\cite{noller2025classical}. Before we continue with the proof, we point out two subtleties, which arise when applying the results from Ref.~\cite{noller2025classical}. 
First, due to having fixed a unitary gauge (and not antiunitary) from the certification of the $\mathcal{S}_2$ model, we find that the unitaries associated to $\Lambda_\iH\vert_A^{\scaleobj{.8}{\ket{\pm}}}$ can so far be constrained to be proportional to either of the two candidates $\{H,XHX\}$.
Moreover, in Ref.~\cite{noller2025classical} we determined inputs for the certification of the model $(\ket{+},\{S,S^\dagger,H\},\{\kb{+}{+},\kb{-}{-}\})$. Similar to \cref{res:S_1}, we do not have a dedicated instruction for the $S^\dagger$-gate here, and instead achieve it by implementing $S^3$. 
We do not repeat the construction from Ref.~\cite{noller2025classical} here, and refer the reader to the proof of Theorem 5 therein.

We proceed by proving unitarity of the full channel $\Lambda_\iH$, for which we use the instruction $\vec{x}=\iS_a\iS_b\iH\iS_a\iS_b$. It shows that $\Lambda_{\iH}$ preserves the purity of the state $\ket{+_y}\ket{+_y}$, in which the subsystem $B$ is in a coherent state with respect to the basis $\{\ket{+},\ket{-}\}$, and thus unitarity follows again by applying \cref{lemma:unitarity_from_subchannels}.
Let $U_\iH$ be the corresponding unitary operator.
Therefore, in analogy to \cref{eq:preliminary_S_A} we have  
\begin{equation}\label{eq:preliminary_H}
    U_{\iH} \eutp X^{\alpha_+} H X^{\alpha_+}\otimes \kb{+}{+} + \e^{\ii\pi\beta} X^{\alpha_-}HX^{\alpha_-}\otimes \kb{-}{-},
\end{equation}
for some $\alpha_+,\alpha_-\in \{0,1\}$ and $\beta\in [0,2]$ 
where we took into account the ambiguity of choosing $\{H,XHX\}$ discussed above. 
The instructions $\iS_a\iS_b\iH^j\iS_a\iS_b$ for $j=1,2$, with the deterministic outcomes $01$ and $11$, respectively, show that
\begin{equation}\label{eq:H_gauge_tests}
    U_{\iH}\ket{+_y}\ket{+_y}\eutp\ket{-_y}\ket{+_y},\quad U_{\iH}^2\ket{+_y}\ket{+_y}\eutp\ket{+_y}\ket{+_y}.
\end{equation}
We combine these conditions with \cref{eq:preliminary_H}, and use a relation $X^{\alpha}HX^{\alpha}\ket{+_y}=\e^{\ii\frac{\pi}{4}}\ii^\alpha\ket{-_y}$ for $\alpha\in\{0,1\}$, to conclude that the following needs to hold,
\begin{equation}
    \ii^{\alpha_+}=\ii^{\alpha_-}\e^{\ii\pi\beta},\text{ and } \e^{2\pi\ii\beta}=1,
\end{equation}
from where we 
find $\alpha_+=\alpha_-\eqqcolon \alpha$, and $\beta=0$.
This leaves us with the following form of $U_\iH$,
\begin{equation}
    U_{\iH}=X^{\alpha}HX^{\alpha}\otimes\openone,
\end{equation}
for $\alpha\in\{0,1\}$.
Just like in the single-qubit case, we can eliminate the parameter $\alpha\in\{0,1\}$ by absorbing $X\otimes X$ into the gauge when $\alpha=1$. 
This, of course, leads to interchanging $S\leftrightarrow S^\dagger$ for both subsystems $A$ and $B$ and, therefore, can be counteracted by applying the complex conjugation on the entire system, or in other words, by allowing for antiunitary gauges. 
Note that adding $X\otimes X$ to the gauge leaves $C_\iH X$-gate invariant, and therefore, we have proven that $\mathcal{N}$ is equivalent to the target model $\mathcal{C}l_2$.
\end{proof}

\section{The $n$-qubit results}
\label{app:n_qubits}
Generalizing \cref{res:S_2} to $n$ qubits is not entirely straightforward.
In the bipartite case, after we have identified the subchannels of one subsystem, we had to reconcile only a few relative phases that could affect the other subsystem. 
For $n$-qubit case, there is an exponential in $n$ number of phases that we need to fix.
Adapting the steps of the two-qubit proof after \cref{app:eq:certified_subchannels}, therefore, requires additional types of input sequences, which do not appear in the $n=2$ case. 

For simplicity of the presentation, we replace the target model $\mathcal{S}_n$ by an equivalent model $(\ket{0}^{\otimes n},\{S_y^{(1)},\dots,S_y^{(n)}\},\{\kb{J}{J}\}_{J\in\{0,1\}^n})$ where $S_y \coloneqq \sqrt{Y}$. 
The equivalence can be realized via conjugation with the unitary $\kb{+_y}{0}+\kb{-_y}{1}$.
Note, that due to this change of the reference model, the gauge freedom, which is left after identifying the single-qubit subchannels, as given by \cref{obs:S_1_gauges}, is no longer $X^r$, but $Z^r\coloneqq \kb{0}{0}+\e^{\ii\pi r}\kb{1}{1}$ for $r\in[0,2]$.

\begin{theorem}\label{res:S_n}
For $n\geq 3$ let $\mathcal{X}_n \coloneqq \bigcup_{k=1}^n\mathcal{X}^{\mathrm{loc}}_{n,k}\cup \{\vec{x}_{n,k},\vec{y}_{n,k}\} \cup \bigcup_{m=1}^{n-1}\mathcal{X}^{\mathcal{G}}_{n,m}$, where 
\begin{equation}\label{eq:n-qubits_inputs}
    \begin{split}
      \mathcal{X}^{\mathrm{loc}}_{n,k}=&\Set{\iS_1^{j_1}\dots\cancel{\iS_k^{j_k}}\dots\iS_n^{j_n}\iS_k^{j_k}\given j_k\in[5], j_{i}\in\{0,2\} \text{ for } i\neq k},\\
      \vec{x}_{n,k} =& (\iS_{k-1}\dots\iS_1\iS_n\dots\iS_{k})^2,\\
       \vec{y}_{n,k}=&(\iS_n\dots\iS_1\iS_k)^2,\\
\mathcal{X}_{n,m}^{\mathcal{G}}=&\Set{\iS_n\dots\iS_m\iS_l\given l = m+1,\dots,n}.
    \end{split}
\end{equation}
Then the input set $\mathcal{X}_n$ certifies the $n$-qubit quantum model 
$\left(\ket{0}^{\otimes N},\{S_y^{(k)}\}_{k=1}^n,\{\kb{J}{J}\}_{J\in\{0,1\}^n}\right)$. 
Here $\cancel{\iS_k}$ indicates that the specific element is omitted in the sequence.
\end{theorem}
The different types of sequences in $\mathcal{X}_n$ achieve the following:
$\mathcal{X}_{n,k}^{\mathrm{loc}}$ reduces the problem to qubit subchannels, and then allow to reconstruct these. 
For each $k$, the string $\vec{x}_{n,k}$ is used to subsequently establish unitarity of $\Lambda_k$. 
In the two-qubit case, these are also sufficient to fix the gauge completely.
However, for $n>2$, we also require the additional inputs $\mathcal{X}_{n,m}^{\mathcal{G}}$ which are used to show that no relative gauge shift can occur when reconstructing the subchannels of $\Lambda_k$. 
Finally, the relative phases can be fixed by considering the sequences $\vec{y}_{n,k}$ for all $k\in\{1,2,\dots,n\}$.

\begin{proof}
Let $\mathcal{M}\coloneqq(\rho,\{\Lambda_{\iS_k}\}_{k=1}^n,\{M_J\}_{J\in\{0,1\}^n})$ be a $2^n$-dimensional quantum model defined over Hilbert space $\HH\cong{\CC^{2^n}}$ that passes \cref{protocol} for the input instructions $\mathcal{X}_n$. 
In the following, we set $\Lambda_{k}\coloneqq \Lambda_{\iS_k}$, and we will also denote by $U_k$ the corresponding unitary operator, after we prove the unitarity of $\Lambda_k$.
Analogously to the proof of \cref{res:S_2}, from the inputs of the type $\iS_1^{j_1}\dots\cancel{\iS^{j_k}_k}\dots\iS_n^{j_n}\iS^{j_k}_k\in \mathcal{X}^{\mathrm{loc}}_{n,k}$ for all $(j_1,j_2\dots j_n)\in\{0,2\}^n$ and the dimension constraint, we establish that $\{M_J\}$ is a rank-1 projective measurement, $\rho$ is a pure state, and $\rho=M_{00\dots 0}$.
Let $\HH_{A_k}\cong \CC^2$ for $k=1,\dots,n$ be Hilbert spaces with which we associate qubit subsystems $A_1,\dots,A_n$.
We denote $\HH_{B_k}\coloneqq \bigotimes_{j\neq k} \HH_{A_j}$ as the subsystem complementary to $A_k$.
Let us assume directly that we have fixed the preliminary gauge unitary which maps the basis of the measurement $\{M_J\}_{J\in\{0,1\}}$ to the product computational basis $\{\ket{I}\}_{I\in\{0,1\}^n}$ of $\bigotimes_{k=1}^{n}\HH_{A_k}$, so that $M_I = \kb{I}{I}$, for $I\in\{0,1\}^n$. 
It then follows that also $\rho=M_0=\kb{0}{0}^{\otimes n}$.
We consider the bipartition $A_k\vert B_k$ for each party $A_k$ for $k=1,\dots,n$ in turn.
Note that none of the steps in the proof of \cref{res:S_2} for the $2$-qubit case up until \cref{app:eq:certified_subchannels} explicitly relies on the fact that it is a qubit-qubit system.
For that reason, we state and prove \cref{lemma:orthogonality} for a general case of a qudit-qudit system.
Therefore, we can establish that the mappings to the subchannels are homomorphisms using the inputs $\mathcal{X}_{n,k}^{\mathrm{loc}}$ and then reconstruct the subchannels $\Lambda_{k}\vert_{A_k}^{I}$ with respect to the computational basis states $\ket{I}_{B_k}$ for $I\in\{0,1\}^{n-1}$ on the subsystem $B_k$. 
Moreover, the corresponding unitary operators of the subchannels take the form 
\begin{equation}
U_{k,I}\eutp Z^{\alpha^{(k)}_I}S_y Z^{-\alpha^{(k)}_I}
\end{equation}
with $\alpha^{(k)}_I\in \RR$ for all $k\in \{1,\dots,n\}$, and $I\in\{0,1\}^{n-1}$, where we assume $\alpha^{(k)}_{00\dots 0}\coloneqq 0$ for every $k$ without loss of generality.
This is analogous to the discussion after \cref{app:eq:S_2_Ub_subchannels} in the $2$-qubit proof. In particular, we find that all the Kraus operators $\{K^j\}$ of some fixed $\Lambda_k$ are also block-diagonal, in the sense that for suitable $K_j^I$ we have
\begin{equation}\label{eq:Kraus_structure}
    K_j = \sum_{I\in\{0,1\}^{n-1}}\kb{I}{I}_{B_k}\otimes (K_j^I)_{A_k}.
\end{equation}
This structure of all the Kraus operators is important to deduce unitarity in the next step, which is similar to the two-qubit case; however, it becomes already slightly more involved.
The reason is that we need to ``combine'' the certified unitary subchannels into the global unitary $U_k$ by dealing with multiple bipartitions. 
For that, we consider the input sequences $\vec{x}_{n,k}$ and focus on the case $k=1$ first.
Since it has to hold that $a_\mathcal{M}((\iS_n\iS_{n-1}\dots\iS_1)^2)=\{11\dots 1\}$, we find that 
\begin{equation}
\Lambda_1\circ\dots\circ\Lambda_n\circ\Lambda_1\circ\dots\circ\Lambda_n(\kb{0}{0}^{\otimes n})=\kb{1}{1}^{\otimes n}.
\end{equation}
Since all the subchannels of each $\Lambda_{k}$ w.r.t.~the computational basis states of $B_k$ are unitary, we know that the channels themselves map the \ac{ONB} $\{\ket{J}\}_{J\in\{0,1\}^n}$ to an \ac{ONB}. 
Then \cref{lemma:ONB_channel_purity}
asserts that none of the channels $\Lambda_{1},\Lambda_{2},\dots,\Lambda_{n}$ can increase the purity of any input state.
This means, in particular, that $\Lambda_1\circ\Lambda_2\circ\dots\circ\Lambda_n(\kb{0}{0}^{\otimes n})$ must be a pure state. 
By iteratively applying a technical \cref{lemma:coherent_state}, which uses the structure of the Kraus operators in \cref{eq:Kraus_structure}, we can iteratively prove that the state
$\Lambda_{k}\circ\Lambda_{k+1}\circ\dots\circ\Lambda_{n}(\kb{0}{0}^{\otimes n})$
is fully coherent with respect to the reference \ac{ONB} on parties $A_k,A_{k+1},\dots,A_{n}$. 
Once we prove it for $k=2$ we can apply \cref{lemma:unitarity_from_subchannels}, which then shows that the channel $\Lambda_{1}$ is unitary. 
Naturally, all the above steps can be repeated for other $k\in \{2,3,\dots,n\}$, and we obtain that all the channels $\Lambda_{2},\dots,\Lambda_{n}$ are unitary.
The important property of the sequences $\vec{x}_{n,k}$ is that their first half ends with $\iS_k$. This ensures that the channel $\Lambda_k$ is applied to a fully coherent state on $B_k$, which allows us to reiterate the above argument. 
\Cref{lemma:unitarity_from_subchannels} additionally asserts that there exist phases $\beta^{(k)}_I\in [0,2]$, such that we can write the corresponding unitaries as
\begin{equation}\label{eq:global_channel_form}
  U_{k} \eutp \sum_{I\in\{0,1\}^{n-1}} \e^{\ii\pi\beta^{(k)}_I} \kb{I}{I}_{B_k}\otimes (Z^{\alpha^{(k)}_I} S_y Z^{-\alpha^{(k)}_I})_{A_k}.  
\end{equation}
We directly set $\beta^{(k)}_{0\dots 0}=0$ for all $k$, as the global phases of the $U_k$ do not play any role.
In the following, if the tensor factors appear in the correct order $A_1,A_2,\dots,A_n$, we will drop the subscripts for better readability.
Our goal now is to modify the gauge unitary in a way which forces all the coefficients $\alpha^{(k)}_I,\beta^{(k)}_I$ to be zero. 
We begin by considering
\begin{equation}
 \Theta_1 \coloneqq \sum_{I\in\{0,1\}^{n-1}} Z^{-\alpha^{(1)}_{I}}\otimes\kb{I}{I} =  \sum_{J\in\{0,1\}^{n}} \e^{\ii\pi\gamma_{J}}\kb{J}{J},   
\end{equation}
where $\gamma_{0I}\coloneqq 0$ and $\gamma_{1I} \coloneqq -\alpha^{(1)}_{I}$ for $I\in\{0,1\}^{n-1}$.
At this point, we need to introduce an additional notation for concatenating a bit to a string of bits at some given position. 
For a bit string $I = i_1i_2\dots i_{n-1}\in\{0,1\}^{n-1}$ of length $(n-1)$, a bit $i'\in\{0,1\}$ and $k\in\{1,2,\dots,n\}$, we define $\ins{i'}{k}I\coloneqq i_1i_2\dots i_{k-1}i'i_{k}\dots i_{n-1}\in \{0,1\}^n$, i.e., a string of length $n$ with $i'$ inserted in the position $k$ of $I$.
We keep using the notation $i'I$ for $\ins{i'}{1}I$ and $Ii'$ for $\ins{i'}{n}I$.\\

First, we apply the gauge $\Theta_1$ to $U_{1}$, and see that it removes all the phases $\alpha^{(1)}_{I}$ as intended,
\begin{equation}
    \Theta_1 U_{1}\Theta_1^\dagger \eutp S_y\otimes\sum_{I\in\{0,1\}^{n-1}} \e^{\ii\pi\beta^{(1)}_I}\kb{I}{I}.
\end{equation}
Next, we calculate the action of the gauge $\Theta_1$ on $U_{k}$ for $k\in\{2,3,\dots n\}$,
\begin{align}\label{app:eq:S_n_Theta1Uk}
\begin{split}
    \Theta_1 U_{k}\Theta_1^\dagger &\eutp \sum_{I\in\{0,1\}^{n-1}} \e^{\ii\pi\beta^{(k)}_I}\kb{I}{I}_{B_k}\otimes \Bigg(\Big(\sum_{i\in\{0,1\}}\e^{\ii\pi\gamma_{\ins{i}{k}I}}\kb{i}{i}\Big) Z^{\alpha^{(k)}_I}S_yZ^{-\alpha^{(k)}_I}\Big(\sum_{j\in\{0,1\}}\e^{-\ii\pi\gamma_{\ins{j}{k}I}}\kb{j}{j}\Big)\Bigg)_{A_k}\\
    &=\sum_{I\in\{0,1\}^{n-1}} \e^{\ii\pi\beta^{(k)}_I}\kb{I}{I}_{B_k}\otimes \left(Z^{(\gamma_{\ins{1}{k}I}-\gamma_{\ins{0}{k}I})}Z^{\alpha^{(k)}_I}S_yZ^{-\alpha^{(k)}_I}Z^{-(\gamma_{\ins{1}{k}I}-\gamma_{\ins{0}{k}I})}\right)_{A_k}\\
    &=\sum_{I\in\{0,1\}^{n-1}} \e^{\ii\pi\beta^{(k)}}\kb{I}{I}_{B_k}\otimes \left(Z^{\tilde\alpha^{(k)}_I}S_yZ^{-\tilde\alpha^{(k)}_I}\right)_{A_k},
\end{split}
\end{align}
where $\tilde\alpha^{(k)}_I \coloneqq \alpha^{(k)}_I+\gamma_{\ins{1}{k}I}-\gamma_{\ins{0}{k}I}\in \RR$. 
By redefining $\alpha^{(k)}_I\coloneqq \tilde\alpha^{(k)}_I$ for $k\in\{2,3,\dots,n\}$, we therefore have achieved to eliminate all $\alpha^{(1)}_I$. 
Next consider 
\begin{equation}
\Theta_2 \coloneqq \sum_{I\in\{0,1\}^{n-2}} \openone\otimes Z^{-\alpha^{(2)}_{0I}}\otimes\kb{I}{I}.
\end{equation}
Since $U_1$ commutes with $\Theta_2$, conjugation with $\Theta_2$ leaves the former invariant.
By acting with the gauge $\Theta_2$ on $U_2$ we obtain the following
\begin{equation}
\Theta_{2} U_{2}\Theta_2^\dagger \eutp \sum_{I\in\{0,1\}^{n-2}}\left(\kb{0}{0}\otimes \e^{\ii\pi\beta^{(2)}_{0I}} S_y+\kb{1}{1}\otimes \e^{\ii\pi\beta^{(2)}_{1I}}Z^{\alpha^{(2)}_{1I}-\alpha^{(2)}_{0I}}S_yZ^{-\alpha^{(2)}_{1I}+\alpha^{(2)}_{0I}}\right)\otimes\kb{I}{I}.
\end{equation}
For $k\in\{3,4,\dots,n\}$, conjugation of $U_k$ by $\Theta_2$ can be accounted for by changing the corresponding $\alpha^{(k)}_I$, similar to the calculations in \cref{app:eq:S_n_Theta1Uk}.
That is, additionally to $\alpha^{(1)}_I$, we managed to eliminate $\alpha^{(2)}_{0I}$ for all $I\in\{0,1\}^{n-2}$.
Iterating on this procedure, we find that by choosing an appropriate gauge unitary, we can set 
\begin{equation}\label{eq:preliminary_alpha}
  \alpha^{(k)}_{0\dots 0I}=0,\quad \forall I\in\{0,1\}^{n-k},  
\end{equation} 
for all $k\in \{1,\dots,n-1\}$.
For $k=n$, we can set $\alpha^{(n)}_{0\dots 0}=0$.
Since the subscripts of $\alpha^{(k)}_I$ and $\beta^{(k)}_I$ are strings of length $n-1$, and there is no ambiguity in the number of $0$-s in \cref{eq:preliminary_alpha}, we keep using this notation in the rest of the proof.

As the next step, we show that all the remaining $\alpha^{(k)}_I$ must be $0$. 
For that, we consider the sequences in $\mathcal{X}^{\mathcal{G}}_{n,1}$,
and, therefore, first calculate the product
\begin{equation}\label{eq:fully_coherent_state}
    U_{1}U_{2}\dots U_{n}\ket{0}^{\otimes n}\eutp \ket{+}\otimes\sum_{J=j_1\dots j_{n-1}\in\{0,1\}^{n-1}} \exp\left(\ii\pi\sum_{l=1}^{n-1}\beta_{0\dots 0j_{l}\dots j_{n-1}}^{(l)}\right)\kb{J}{J}\ket{+}^{\otimes n-1}.
\end{equation}
For better readability, we conduct the following calculation only for the sequence $\iS_n\dots\iS_1\iS_n\in\mathcal{X}^{\mathcal{G}}_{n,1}$ explicitly, while the other cases where the last put is $\iS_k$ and $k\in\{2,\dots,n-1\}$ follow analogously.
In order to keep the calculations tractable, we introduce the following additional notation. 
For $J=j_1\dots j_{n-1}\in\{0,1\}^{n-1}$, we set $\theta(J)\coloneqq\sum_{l=1}^{n-1}\beta^{(l)}_{0\dots 0j_l\dots j_{n-1}}$ and for $I=i_1\dots i_l\in\{0,1\}^{l}$, we denote $\wout{I}{k}\coloneqq i_1\dots i_{k-1}i_{k+1}\dots i_l$, i.e., the sequence $I$ with the $k$-th element removed. 
Also, for $k\in\{2,3,\dots,n\}$ and $I=i_1\dots i_{n-1}\in\{0,1\}^{n-1}$ we define 
\begin{align}\label{eq:Delta}
    \Delta_{I,k} &\coloneqq \theta(i_2i_3\dots i_{k-1} 1 i_{k},\dots i_{n-1})-\theta(i_2i_3 \dots i_{k-1} 0 i_{k}\dots i_{n-1})=\theta(\wout{I}{1}\oplus_{k-1} 1)-\theta(\wout{I}{1}\oplus_{k-1}0),\\
    \varphi(I,k)&\coloneqq\beta^{(k)}_{I}+\theta(i_2i_3\dots i_{k-1}0 i_{k}\dots i_{n-1})=\beta_I^{(k)}+\theta(\wout{I}{1}\oplus_{k-1}0),
\end{align}
where the term $\varphi(I,k)$ will be used to collect some phases that do not affect the proof.
We also use $U_{k,I}$ for the unitary operator of the subchannel, i.e., without spelling it out as $Z^{\alpha_I^{(k)}} S_y Z^{-\alpha_I^{(k)}}$.
Using the new notation and \cref{eq:fully_coherent_state}, we calculate 
\begin{align}\label{app:eq:S_n_U_nU_1dotsU_n}
\begin{split}
    U_{n}U_{1}U_{2}\dots U_{n}\ket{0}^{\otimes n} \eutp &\frac{1}{\sqrt{2}^{n-1}}\left(\sum_{I\in\{0,1\}^{n-1}} \e^{\ii\pi\beta^{(n)}_{I}}\kb{I}{I}\otimes U_{n,I}\right)\left( \ket{+}\otimes\sum_{J\in\{0,1\}^{n-1}} \e^{\ii\pi\theta(J)}\ket{J}\right)\\
    =&\frac{1}{\sqrt{2}^{n-1}}\sum_{I,J\in \{0,1\}^{n-1}}\ket{i_1}\braket{i_1}{+} \e^{\ii\pi(\beta^{(n)}_{I}+\theta(J))}\ket{\wout{I}{1}}\braket{\wout{I}{1}}{\wout{J}{n-1}}\otimes U_{n,I}\ket{j_{n-1}}\\
    =&\frac{1}{\sqrt{2}^{n}}\sum_{I\in\{0,1\}^{n-1}}\left(\e^{\ii\pi\beta^{(n)}_{I}} \ket{I}\otimes \sum_{j\in\{0,1\}} \e^{\ii\pi\theta(\wout{I}{1}j)}U_{n,I}\ket{j}\right)\\
    =&\frac{1}{\sqrt{2}^{n}} \sum_{I\in\{0,1\}^{n-1}}\Bigg(\e^{\ii\pi(\beta_I^{(n)}+\theta(\wout{I}{1}0))}\ket{I}\otimes  U_{n,I}\left(\ket{0}+\e^{\ii\pi(\theta(\wout{I}{1}1)-\theta(\wout{I}{1}0))}\ket{1}\right)\Bigg)\\
    =&\frac{1}{\sqrt{2}^{n-1}}\sum_{I\in\{0,1\}^{n-1}}\e^{\ii\pi\varphi(I,n)}\ket{I}\otimes U_{n,I}Z^{\Delta_{I,n}}\ket{+}.
     \end{split}
    \end{align}
In the general case of $k\in\{2,\dots,n-1\}$, we find analogously
\begin{equation}\label{eq:aux_state}
    U_k U_1U_2 \dots U_n\ket{0}^{\otimes n}\eutp \frac{1}{\sqrt{2}^{n-1}}\sum_{I\in\{0,1\}^{n-1}} \e^{\ii\pi\varphi(I,k)}\ket{I}_{B_k}\otimes \left(U_{k,I} Z^{\Delta_{I,k}}\ket{+}\right)_{A_k}.
\end{equation}
Next, we compute the reduced state of the subsystem $A_k$ by tracing out $B_k$, which a priori could be a mixed state
\begin{align}\label{app:eq:S_n_rho_k}
    \rho_k \coloneqq \frac{1}{2^{n-1}}\sum_{I\in\{0,1\}^{n-1}} U_{k,I} Z^{\Delta_{I,k}}\kb{+}{+} Z^{-\Delta_{I,k}}U_{k,I}^\dagger.
\end{align}
However, from the measurement outcomes $a_\mathcal{M}(\iS_n\iS_{n-1}\dots\iS_1\iS_k)\subseteq\Set{J\in\{0,1\}^n\given j_k=1}$ we find that this partial state must be pure and equal to $\kb{1}{1}$. 
Hence, we establish that the state in \cref{eq:aux_state} must be a product state with respect to the bipartition $A_k\vert B_k$, and by spelling out $U_{k,I}$ in \cref{app:eq:S_n_rho_k} we find
\begin{equation}\label{eq:alpha_Delta_condition}
  Z^{\alpha^{(k)}_{I}}S_yZ^{-\alpha^{(k)}_{I}} Z^{\Delta_{I,k}}\ket{+}\eutp\ket{1},\quad \forall I\in\{0,1\}^{n-1}.  
\end{equation}
From the above, we can conclude that $\alpha_{I}^{(k)}=\Delta_{I,k}$ must hold. 
Since $\Delta_{I,k}$ is independent of $i_1$ (see \cref{eq:Delta}), we know that 
\begin{align}\label{eq:parameter_i1_elimination}
    \alpha^{(k)}_{0I}=\alpha^{(k)}_{1I},\quad \forall I\in\{0,1\}^{n-2}.
\end{align}
For $k=2$, it then follows that $\alpha^{(2)}_{1I}=0$ for all $I\in\{0,1\}^{n-2}$, since previously we achieved $\alpha^{(2)}_{0I}=0, \forall I\in\{0,1\}^{n-2}$ by choosing an appropriate gauge.
Hence, $\alpha^{(2)}_I=0$ for all $I\in \{0,1\}^{n-1}$.
We have therefore successfully eliminated the dependence on $i_1$ for all the $\alpha^{(k)}_I$ parameters using the inputs $\mathcal{X}_{n,1}^{\mathcal{G}}$. This is fundamentally due to $U_1$ being applied last to reach the fully coherent state, before again acting with the different $U_k$ in \cref{eq:aux_state}. 
Therefore, we proceed to iterate the scheme by analyzing the input sequences in $\mathcal{X}_{n,2}^{\mathcal{G}}$, which are given by $\iS_n\iS_{n-1}\dots\iS_2\iS_k$ with $k=3,\dots,n$. 
The analogous calculations and reasoning then show that $\alpha^{(k)}_{00I}=\alpha^{(k)}_{01I}$ for all $I\in\{0,1\}^{n-3}$ and $k=3,\dots,n$. 
The calculations are similar to Eqs.~(\ref{app:eq:S_n_U_nU_1dotsU_n}-\ref{eq:alpha_Delta_condition}), with the difference being that the subsystem $A_1$ is essentially unaffected and remains in the state $\ket{0}$.
We combine this with the constraints in \cref{eq:parameter_i1_elimination}, to find 
\begin{equation}
\alpha^{(k)}_{11I}=\alpha^{(k)}_{01I}=\alpha^{(k)}_{00I}=\alpha^{(k)}_{10I},\quad \forall I\in\{0,1\}^{n-3}, 
\end{equation}
for all $k\geq 3$, which shows that the $\alpha_I^{(k)}$-parameters are independent of the first two bits of their index.
From \cref{eq:preliminary_alpha} we then obtain that $\alpha^{(3)}_I=0$ for all $I\in\{0,1\}^{n-1}$. 
Iterating further, by considering all the sequences of the type $\iS_n\iS_{n-1}\dots\iS_m\iS_k\in\mathcal{X}_{n,m}^{\mathcal{G}}$ for $m=3,\dots,n-1$, we can conclude that $\alpha^{(k)}_I=0$ for all $k$ and $I\in\{0,1\}^{n-1}$.\\

It is left to show that $\beta^{(k)}_I=0$ for all $k$ and $I$.
We further investigate the implications of the condition $\alpha_I^{(k)}=\Delta_{I,k}$.
Since we found that $\alpha^{(k)}_I = 0$, we get $\Delta_{I,k}=\theta(\wout{I}{1}\oplus_{k-1} 1)-\theta(\wout{I}{1}\oplus_{k-1} 0)=0$ for all $I$. 
This condition shows that for each $k=2,\dots,n$, $\theta(J)$ does not depend on the bit $j_{k-1}$ of $J=j_1\dots j_{n-1}$, so altogether we conclude that $\theta(J)$ is actually a constant of $J$. 
In particular, \cref{eq:fully_coherent_state} simplifies to
\begin{equation}\label{eq:fully_coherent_state_identity}
    U_{1}\dots U_{n}\ket{0}^{\otimes n}\eutp\ket{+}^{\otimes n}.
\end{equation}
From $a_\mathcal{M}((\iS_n\iS_{n-1}\dots\iS_1)^2)=\{1^n\}$, which is the sequence we denoted as $\vec{x}_{n,1}$, and \cref{eq:fully_coherent_state_identity}, we can deduce that also 
\begin{equation}\label{eq:fully_coherent_state_inverse_1}
    U_{n}^\dagger\dots U_{1}^\dagger \ket{1}^{\otimes n}\eutp\ket{+}^{\otimes n},
\end{equation} 
must hold.
As the next step, we examine the inputs $\vec{y}_{n,k}$ and the corresponding conditions $a_\mathcal{M}\left((\iS_n\dots\iS_1\iS_k)^2\right) = \{1^{k-1}01^{n-k}\}$. Since $U_k^\dagger$ acts as a local $S_y$ gate on $\ket{0}_{A_k}\otimes \ket{1}^{\otimes n-1}_{B_k}$ we find that with \cref{eq:fully_coherent_state_identity}, this leads to
\begin{equation}\label{app:eq:S_n_proof_U1dotsUnUk}
    U_1\dots U_nU_k\ket{+}^{\otimes n} \eutp U_k^\dagger\ket{0}_{A_k}\otimes\ket{1}^{\otimes n-1}_{B_k}\eutp \ket{-}_{A_k}\otimes\ket{1}^{\otimes n-1}_{B_k}.
\end{equation}
We continue by considering the pairing of the outer sides of \cref{app:eq:S_n_proof_U1dotsUnUk} with $\bra{1}^{\otimes n}$. With \cref{eq:fully_coherent_state_inverse_1}, we find that $\bra{+}^{\otimes n}U_k\ket{+}^{\otimes n} \eutp \frac{1}{\sqrt{2}}$ must hold.
Inserting the explicit form of $U_k$ from \cref{eq:global_channel_form} in the above equation, and keeping in mind that $\bra{+}S_y\ket{+} \eutp \frac{1}{\sqrt{2}}$, we obtain
\begin{equation}
    \sum_{I\in\{0,1\}^{n-1}}\e^{\ii\pi\beta^{(k)}_I} \kb{I}{I}\ket{+}^{\otimes n-1} \eutp \ket{+}^{\otimes n-1}.
\end{equation}

This of course necessitates, that $\beta^{(k)}_I=\beta^{(k)}_0$ for all $I$. 
Since we fixed all $\beta^{(k)}_{0^{n-1}}=0$ beforehand, we then find $U_k = (S_y)_{A_k}\otimes \openone_{B_k} =S_y^{(k)}$ for all $k=1,\dots,n$, which concludes the proof.  
\end{proof}
\begin{reptheorem}{res:universal}[restated]
For every $n\geq 2$, the $n$-qubit model $\mathcal{U}_n = \Cl_n + \{CS^{(1,2)}\}$ can be certified by the set of inputs $\mathcal{X}^u_n$ which has cardinality $\abs{\mathcal{X}^u_n}\in \LandauO(n 2^n)$.
\end{reptheorem}
\begin{proof}
\Cref{res:S_n} shows that the inputs $\mathcal{X}_n$ are regular for the model $\mathcal{S}_n$. A quick accounting shows that $\abs{\mathcal{X}_n}\in\mathcal{O}(n2^n)$.
Therefore, every $2^n$-dimensional model $\mathcal{N}\coloneqq (\rho,\{\Lambda_{j}\}_{j=1}^n,\{M_I\}_{I\in\{0,1\}^n})+\{\Lambda_{(1,j)}\}_{j=2}^n+\{\Lambda_\iH,\Lambda_{\mathrm{cs}}\}$ which is compatible with the inputs $\mathcal{X}_n $ can, without loss of generality, be assumed to be of the form $\mathcal{S}_n + \{\Lambda_{(1,j)}\}_{j=2}^n+\{\Lambda_\iH,\Lambda_{\mathrm{cs}}\}$.
Since $C_\iH X^{(1,2)}$ is diagonal in the $n$-qubit Hadamard basis, we can apply \cref{cor:adding_inputs} with $\mathcal{C}=\{\ket{\mathrm{++}_y},\ket{\mathrm{+}_y+},\ket{\mathrm{-+}_y}\}\otimes\ket{+_y}^{\otimes n-2}$ which satisfies the condition \ref{prop:3}.
Analogously, we can reconstruct all the remaining $C_\iH X^{(1,j)}$ gates for $j\in\{3,\dots,n\}$, and, thus, can further simplify $\mathcal{N}=\mathcal{S}_n+\{C_\iH X^{(1,j)}\}_{j=2}^n+\{\Lambda_\iH,\Lambda_{\mathrm{cs}}\}$.
By \cref{cor:adding_inputs}, adding each $C_\iH X^{(1,j)}$ gate requires $\mathcal{O}(2^n)$ tests and, therefore, the total number of inputs is still of order $\mathcal{O}(n 2^n)$.
It is straightforward to generalize \cref{res:Clifford_2} for all $\mathcal{C}l_n$ with $n\geq 3$, adding $\mathcal{O}(2^n)$ more inputs, and therefore we find $\mathcal{N}=\Cl_n+\{\Lambda_{\mathrm{cs}}\}$. 
Having access to the Clifford group allows us to prepare any state of the computational basis from $\ket{+}^{\otimes n}$.
Additionally, the set $\mathcal{C'}=\{\ket{0+},\ket{\mathrm{+}0},\ket{1+}\}\otimes \ket{+}^{\otimes n-2}$ is also among the attainable states of $\mathcal{C}l_n$. 
We can use these two facts and \cref{cor:adding_inputs} to obtain an input set that certifies the augmented model $\mathcal{C}l_n+\{CS^{(1,2)}\}$. 
Again this adds $\mathcal{O}(2^n)$ inputs, with the total scaling remaining at $\mathcal{O}(n2^n)$.
\end{proof}

\section{Robustness analysis}
\label{app:numerics}
Here we provide more details on the robustness analysis of the \ac{QSQ} protocol. 
In this analysis, we investigate the relation between the probability $\PP(\text{\vbt{pass}})$ of a noisy model $\mathcal{N} = (\tilde\rho,\Set{\tilde\Lambda_x}_{x\in \XX},\Set{\tilde M}_{a\in \AA})$ to pass a single repetition of \cref{protocol}, and the infidelity between $\mathcal{N}$ and the target model $\mathcal{M} = (\rho,\Set{\Lambda_x}_{x\in \XX},\Set{M_a}_{a\in \AA})$, both defined over $\HH\cong \CC^{2^n}$.
We repeat the definition~\eqref{eq:distNM} of the infidelity between two models:
\begin{equation}\label{app:eq:distNM}
    \dist(\mathcal{N},\mathcal{M}) \coloneqq \min_U\max\Set*{1-\Fid(\tilde\rho,U\rho U^\dagger),1-\Fid_\avg(\tilde\Lambda_x,U\circ\Lambda_x\circ U^{-1}),\d_{TV}(\Set{\tilde M_a}_{a\in\AA}, \Set{UM_aU^\dagger}_{a\in\AA})},
\end{equation}
where we have the total variation distance between the measurements
\begin{equation}
    \d_{TV}(\Set{\tilde M_a}_{a\in\AA}, \Set{M_a}_{a\in\AA}) \coloneqq \frac{1}{2}\max_{\sigma\in \DM(\HH)}\Big\{\sum_{a\in \AA}\abs*{\Tr[\sigma(\tilde M_a - M_a)]}\Big\}.
\end{equation}
Moreover, when $\rho$ is pure $\Fid(\tilde\rho,\rho)\coloneqq \Tr[\tilde\rho\rho]$ and similarly, for unitary $\Lambda$ the average gate infidelity becomes
\begin{equation}
    \Fid_\avg(\tilde\Lambda,\Lambda)\coloneqq \frac{1}{d+1} + \frac{1}{d(d+1)}\sum_{i,j=0}^{d-1}\Tr[\Lambda(\ketbra{i}{j})^\dagger\tilde\Lambda(\ketbra{i}{j})].
\end{equation} 
In \cref{app:eq:distNM}, the minimization is taken over (anti-) unitary transformations, the maximization is also taken over all $x\in\XX$, and with a slight abuse of notation, we used the same letter to denote the unitary channel corresponding to $U$
\begin{equation}
    U\circ\Lambda_x\circ U^{-1}(\argdot)\coloneqq U\Lambda_x(U^\dagger (\argdot) U)U^\dagger,
\end{equation}
which, in the case of $U$ being anti-unitary, corresponds to the unitary channel and the complex conjugation applied to $\Lambda_x$.

For the robustness analysis of the soundness property, we need to find an upper bound on $\PP(\text{\vbt{pass}})$ for models $\mathcal{N}$, with $\dist(\mathcal{N},\mathcal{M})\geq \veps$. 
For the case of completeness, we need to show that there is a lower bound on $\PP(\text{\vbt{pass}})$ for models with infidelity less than some threshold value.

\begin{reptheorem}{th:robustness}[restated]
    Let the probability of a model $\mathcal{N}$ to fail a single repetition of the \cref{protocol} for the target model $\mathcal{M}$ be $\PP(\text{\textnormal{\vbt{fail}}})$, and the infidelity between these models be $\dist(\mathcal{N},\mathcal{M})$ as defined in \cref{app:eq:distNM}.
    The \ac{QSQ} protocol is robust in the sense that there exist integer $k_1,k_2\geq 1$, and real $C_1,C_2>0$, such that the following holds
    \begin{equation}\label{eq:app:robust}
    \begin{split}
       C_1\dist(\mathcal{N},\mathcal{M})^{k_1}\leq \PP(\text{\textnormal{\vbt{fail}}})\leq C_2\dist(\mathcal{N},\mathcal{M})^{1/k_2}.
    \end{split}    
    \end{equation}
    Moreover, there exists $C_2\in \LandauO(n2^n)$ for $k_2=2$.
\end{reptheorem}
\begin{proof}
The proof rests on the result from algebraic geometry known as Lojasiewicz's inequality (see e.g., Ref.~\cite[Cor.~2.6.7]{Bochnak1998}), also used in Ref.~\cite{vanDam2007self-testing} to prove the robustness of self-tests.
This result states that if there exist two continuous semi-algebraic functions $f,g: Y\to \RR$, defined over a compact semi-algebraic set $Y\subseteq \RR^m$, such that for all $\vec{y}\in Y$, if $f(\vec{y})=0$ then $g(\vec{y})=0$, then there exist an integer $k\geq 1$ and a real $C>0$, such that $\forall \vec{y}\in Y$, $\abs{g(\vec{y})}^k\leq C\abs{f(\vec{y})}$.

Fix an \ac{ONB} in $\HH$ over which the model $\mathcal{N}$ is defined, and consider a parametrization of $\mathcal{N}$ w.r.t.~this \ac{ONB}, with the parameters $\vec{y}\in \RR^m$ for some large enough $m$.
Let $Y$ be a set of all $\vec{y}$ that correspond to the models $\mathcal{N}$ that are physical.
Since the physicality constraints are represented as a set of \ac{PSD} constraints and linear in $\vec{y}$ set of equations, the set $Y$ can be represented as 
\begin{equation}\label{eq:app:Y_def}
    Y = \Set{\vec{y}\in \RR^m\given P_1(\vec{y})\geq 0,\dots,P_{l_1}(\vec{y})\geq 0,P_{l_1+1}(\vec{y})=0,\dots,P_{l_1+l_2}(\vec{y})=0},
\end{equation}
where $P_i(\vec{y})$ are polynomials in the entries of $\vec{y}$, and the number of constraints $l_1+l_2$ is finite.
This implies that $Y$ is semi-algebraic and closed, since all the inequalities in \cref{eq:app:Y_def} are not strict.
Moreover, since all the operators in $\mathcal{N}$ are bounded, the set $Y$ is bounded itself and, thus, is compact.

Let us take $f(\vec{y})=\PP(\text{\textnormal{\vbt{fail}}})$.
Clearly, $f(\vec{y})$ is a polynomial in the entries of $\vec{y}$.
Hence, it is a continuous and semi-algebraic function, since the graph of $f(\vec{y})$, i.e., the set $\Set{(v,\vec{y})\in \RR^{m+1} \given v=f(\vec{y}), \vec{y}\in Y}$, is semi-algebraic.
The case of $\dist(\mathcal{N},\mathcal{M})$ as a function of $\vec{y}$ is more subtle, because in its definition, we have the optimization over the unitary gauge. 
First, we need to define a new function for the infidelity between two models
\begin{equation}
\begin{split}
    \dist_\Sigma(\mathcal{N},\mathcal{M}) = \frac{1}{3}&\min_{U\in \U(\HH)}\Bigg\{3-\Fid(\tilde\rho,U\rho U^\dagger)-\frac{1}{2^n}\sum_{a\in \AA}\Tr[\tilde M_a U M_aU^\dagger]-\frac{1}{\abs{X}}\sum_{x\in \XX}\Fid_\avg(\tilde\Lambda_x,U\circ\Lambda_x\circ U^\dagger)\Bigg\},
    \end{split}
\end{equation}
where, as compared to \cref{app:eq:distNM}, we took the average over the infidelities instead of the maximum, set the minimization to be only over the unitary transformations, and for the \acp{POVM} took a different figure of merit for their closeness. 
As we show in \cref{lemma:dist}, $\dist_\Sigma(\mathcal{N},\mathcal{M})\leq \dist(\mathcal{N},\mathcal{M})\leq 3\max\Set{2^{n},\abs{\XX}}\sqrt{\dist_\Sigma(\mathcal{N},\mathcal{M})}$, if the optimal gauge in \cref{app:eq:distNM} is unitary.
Clearly, when the optimal gauge is anti-unitary, the same relation holds for the model $\mathcal{M}^\ast$ with the complex conjugation applied to the state, gates, and measurement.

Let us first take $g(\vec{y}) = \dist_\Sigma(\mathcal{N},\mathcal{M})$. 
To prove that $g(\vec{y})$ is semi-algebraic, we use the consequence of Tarski-Seidenberg's theorem (see e.g., Ref.~\cite[Cor.~A.2.4]{Hoermander2005}), namely that a real function defined as $g(\vec{y}) = \inf_{\vec{z}\in Z}\Set{h(\vec{y},\vec{z})}$ where $Z\subseteq \RR^{m'}$ is semi-algebraic if $h: Y\times Z \to \RR$ and $Z$ are semi-algebraic.  
Let $\vec{z}$ be the parameters corresponding to a parametrization of unitary operators w.r.t.~the same \ac{ONB} as above, and $Z$ be the corresponding set for all unitary operators. 
Clearly, there exists a parametrization for which $Z$ is compact and semi-algebraic.
Moreover, $h(\vec{y},\vec{z})$ is again a polynomial in terms of the entries of $\vec{y}$ and $\vec{z}$, which means that its graph is semi-algebraic. 
Finally, due to (a trivial case of) Berge's Maximum Theorem (see e.g., Ref.~\cite[Th.~17.31]{aliprantis2006infinite}) in parameter optimization, the function $g(\vec{y})$ is continuous, since $h(\vec{y},\vec{z})$ is continuous over $Y\times Z$. 

The case of $g(\vec{y})=\dist_\Sigma(\mathcal{N},\mathcal{M}^\ast)$ is clearly analogous to the above. 
We, therefore, apply Lojasiewicz's inequality in both directions, i.e., also for $f(\vec{y})$ and $g(\vec{y})$ swapped, and obtain the relations in \cref{eq:app:robust} first for $\dist_\Sigma(\mathcal{N},\mathcal{M})$ (or $\dist_\Sigma(\mathcal{N},\mathcal{M}^\ast)$) and then due to \cref{lemma:dist} also for $\dist(\mathcal{N},\mathcal{M})$.
This completes the first part of the proof.

Now, we show that the upper bound in \cref{eq:app:robust} holds for $k_2=2$ and $C_2\in \LandauO(n2^n)$.
Let the model $\mathcal{N}$ be such that $\dist(\mathcal{N},\mathcal{M})\leq \veps$, for some $\veps\geq 0$. 
First, we notice that the unitary gauge $U$ does not play any role in the completeness analysis, because every model equivalent to $\mathcal{M}$ leads to the same observed statistics. 
Therefore, we simply set $U=\1$ in the argument below.
From $\dist(\mathcal{N},\mathcal{M})\leq \veps$, we directly have that each of the terms representing the infidelities in \cref{app:eq:distNM} for the state, measurement, and the channels must satisfy the upper bound of $\veps$.
In particular, we use the condition $1-\Fid_\avg(\tilde\Lambda_x,\Lambda_x)\leq \veps$, to conclude that $\norm{\tilde\Lambda_x-\Lambda_x}_\diamond\leq 2\sqrt{2^n(2^n+1)\veps}\in\LandauO(2^n\sqrt{\veps})$ (see e.g.,~\cite{kliesch2021theory}).
Similarly, for the state $\tilde\rho$, we have $\norm{\tilde\rho-\rho}_1\leq 2\sqrt{\veps}$~\cite{Fuchs1999}. 
Then, for any tested sequence $\vec{x} = x_1x_2\dots x_l$ of length $l$, for which the set of target outcomes is $a_\mathcal{M}(\vec{x})$, we have
\begin{equation}
    \begin{split}
        \abs*{\PP(a\in a_\mathcal{M}(\vec{x})\vert \vec{x})-1} & \leq \abs*{\sum_{a\in a_\mathcal{M}(\vec{x})}\left(\Tr[\tilde M_a\tilde \Lambda_{x_l}\circ\cdots\circ \tilde\Lambda_{x_2}\circ\tilde \Lambda_{x_1}(\tilde \rho)-\Tr[M_a \Lambda_{x_l}\circ\cdots\circ\Lambda_{x_2}\circ\Lambda_{x_1}(\rho)]\right)}\\
        & \leq \abs*{\sum_{a\in a_\mathcal{M}(\vec{x})}\left(\Tr[\tilde M_a\tilde \Lambda_{x_l}\circ\cdots\circ \tilde\Lambda_{x_2}\circ\tilde \Lambda_{x_1}(\rho)-\Tr[M_a \Lambda_{x_l}\circ\cdots\circ\Lambda_{x_2}\circ\Lambda_{x_1}(\rho)]\right)} + 2\sqrt{\veps} \\
        & \leq \abs*{\sum_{a\in a_\mathcal{M}(\vec{x})}\left(\Tr[\tilde M_a\tilde \Lambda_{x_l}\circ\cdots\circ \tilde\Lambda_{x_2}\circ \Lambda_{x_1}(\rho)-\Tr[M_a \Lambda_{x_l}\circ\cdots\circ\Lambda_{x_2}\circ\Lambda_{x_1}(\rho)]\right)} + \LandauO(2^n\sqrt{\veps}) \\
        & \leq \cdots \leq  \sum_{a\in a_\mathcal{M}(\vec{x})}\abs*{\Tr[(\tilde M_a-M_a)\Lambda_{x_l}\circ\cdots\circ\Lambda_{x_2}\circ \Lambda_{x_1}(\rho)]} + \LandauO(l 2^n\sqrt{\veps}) = \LandauO(l 2^n\sqrt{\veps}),
    \end{split}
\end{equation}
where the dots indicate the iteration of the previous estimate over the channels $\Lambda_{x_1},\dots,\Lambda_{x_l}$.
Since the maximal length of the tested sequences that we consider in this work is $l\in \LandauO(n)$, we can conclude that $\PP(\text{\vbt{pass}})\geq 1-\LandauO(n2^n\sqrt{\veps})$, i.e., there is a certain level of noise $\veps$ that can be tolerated by the \ac{QSQ} protocol for a given finite number of repetitions $N$. 
In our numerical investigations, we observe that the above general, but pessimistic in scaling, assessment is not tight, and, in fact, one should expect a linear scaling between $\PP(\text{\vbt{pass}})$ and $\dist(\mathcal{N},\mathcal{M})$. 
\end{proof}

\section{Technical Lemmata}
\label{app:lemmas}
In this section, we state and prove some technical results about channels.
Some of these results can be of independent interest, e.g., \cref{corr:unitarity} can be turned into a protocol for certifying unitarity of a channel, and, similarly, \cref{lemma:channel=U} could be useful for certifying the identity channel.
\begin{techlemma}\label{lemma:unitarity_from_subchannels}
Let $\Lambda\in\CPTP(\mathcal{H}_A\otimes \mathcal{H}_B)$, for finite-dimensional Hilbert spaces $\mathcal{H}_A,\mathcal{H}_B$ of dimensions $d_A$ and $d_B$, respectively.
Assume that there exists an orthonormal basis $\{\ket{\psi_i}\}_{i\in[d_A]}$ of $\mathcal{H}_A$, and a set of unitaries $\{U_i\}_{i\in [d_A]}$ on $\HH_B$ such that the following conditions are satisfied:
\begin{enumerate}
    \item[(i)] For all $\ket{\psi_i}$ and $\rho\in\DM(\HH_B)$ we have $ \Lambda(\kb{\psi_i}{\psi_i}\otimes\rho)=\kb{\psi_i}{\psi_i}\otimes U_i\rho U_i^\dagger $.
    \item[(ii)] There exists a set $\mathcal{C}\subseteq \DM(\mathcal{H}_A\otimes\mathcal{H}_B)$ of product states, such that the reduced states on the subsystem $A$ are fully coherent for the basis $\{\ket{\psi_i}\}_{i\in[d_A]}$ (in the sense of \cref{def:coherence_graph}), and $\Lambda(\rho)$ is a pure state for every $\rho\in\mathcal{C}$.
\end{enumerate}
Then the channel $\Lambda$ is unitary, i.e., $\exists U\in \U(\HH_A\otimes\HH_B)$ s.t.,
$\Lambda(\rho)=U\rho U^\dagger$, and there exist phases $\theta_i\in \RR$, such that
\begin{equation} 
    U = \sum_{i\in[d_A]} \e^{\ii\theta_i}\kb{\psi_i}{\psi_i}\otimes U_i.
\end{equation}
\end{techlemma}
\begin{proof}
Let $\{K_j\}_{j\in [n]}$ be a set of Kraus operators for the channel $\Lambda$, where $n\in \NN$ is the Kraus rank.
Set $\ket{i}=\ket{\psi_i}$ for brevity.
From \cref{lemma:marginal_homomorphism}, we have that 
$K_j = \sum_{i\in [d_A]} \kb{i}{i}\otimes K^i_j$. 
Next, we use the condition (i) of the lemma and find that for every $\rho\in \DM(\HH_B)$, it must hold that
\begin{equation}\label{app:eq:lemma_unit_from_sub_1}
    \sum_{j=0}^{n-1} \kb{i}{i}\otimes K^i_j\rho (K^i_j)^\dagger = \kb{i}{i}\otimes U_i\rho U_i^\dagger, \quad \forall i\in [d_A].
\end{equation} 
Since \cref{app:eq:lemma_unit_from_sub_1} holds for all $\rho\in \DM(\HH_B)$, it necessitates that all $K_j^i$ are proportional to $U_i$, i.e., $K_j^{i} = \lambda_{j}^i U_i$, for some appropriate $\lambda_j^i\in \CC$.
From the normalization of $\Lambda$, we have that $\sum_{j=0}^{n-1} \abs{\lambda_j^i}^2 = 1$ hold for all $i$, therefore, we can associate a normalized vector $\bm{\lambda}^i=(\lambda^i_0,\dots,\lambda^i_{n-1})$ to each element of the \ac{ONB} $\{\ket{i}\}_{i\in [d_A]}$.

Assume for a moment that all the vectors $\bm{\lambda}^i$ are equal up to a phase, that is, for each $i\in [d_A]$ we find $\theta_i\in \RR$, such that $\bm{\lambda}^i=\e^{\ii\theta_i}\bm{\lambda}^0$. 
In this case, the Kraus operators take the form 
\begin{equation}
    K_j = \lambda^0_j\sum_{i\in[d_A]}\kb{i}{i}\otimes \e^{\ii\theta_i} U_i,\quad \forall j\in [n].
\end{equation}
In particular, all $K_j$ differ only by the factor $\lambda^0_j$, and therefore the summation over $j$ in the Kraus decomposition factors out, which proves that the channel $\Lambda$ is unitary. 
The remainder of the proof is therefore dedicated to showing that all $\bm{\lambda}^i$ differ only by a phase factor.
Let $\sigma = \sum_{k,l\in [d_A]} c_{kl} \kb{k}{l}$, for some $c_{kl}\in\CC$, such that $\sigma\in \DM(\HH_A)$.
Then for any $\rho\in\DM(\HH_B)$ we have
\begin{align}
\label{eq:aux_1}\Lambda(\sigma\otimes\rho) &= \sum_{j=0}^{n-1} K_j \sigma\otimes\rho K_j^\dagger= \sum_{j=0}^{n-1}\sum_{k,l\in[d_A]}c_{kl} \kb{k}{l}\otimes K_j^{k}\rho {K_j^{l}}^\dagger \\
    &=\sum_{k,l} c_{kl}\kb{k}{l}\otimes \left(\sum_j\lambda_j^k\bar{\lambda}_j^l\right) U_k\rho U_l^\dagger = \sum_{k,l}c_{kl}\langle\bm\lambda^l,\bm\lambda^k\rangle\kb{k}{l}\otimes U_k\rho U_l^\dagger.
\end{align}
In order to use the condition (ii) of the lemma, we calculate the purity of $\Lambda(\sigma\otimes\rho)$, 
\begin{align}
\begin{split}\label{eq:subchannel_reconstruction_lemma}
    \tr[\Lambda(\sigma\otimes\rho)^2] &= \tr\left[\sum_{k,l,k',l'\in[d_A]}\left(c_{kl}c_{k'l'}\langle\bm\lambda^l,\bm\lambda^k\rangle\langle\bm\lambda^{l'},\bm\lambda^{k'}\rangle\right)\ket{k}\braket{l}{k'}\bra{l'}\otimes U_k\rho U_l^\dagger U_{k'}\rho U_{l'}^\dagger\right]\\
    &=\sum_{k,l\in [d_A]}\abs{c_{kl}}^2\abs{\langle\bm\lambda^l,\bm\lambda^k\rangle}^2\tr[\rho^2].
    \end{split}
\end{align}
Clearly, the factors $\abs{\langle\bm\lambda^l,\bm\lambda^k\rangle}$ are bounded by $1$. On the other hand, we also have $\sum_{k,l}\abs{c_{kl}}^2=\tr[\sigma^2]\leq 1$. 
Therefore, the condition $\tr[\Lambda(\sigma\otimes\rho)^2]=1$ requires, aside of the purity of the states $\sigma,\rho$, that $\abs{\langle\bm\lambda^l,\bm\lambda^k\rangle}=1$ for any pair of indices $k,l$ for which $c_{kl}=\bra{k}\sigma\ket{l}\neq 0$.
Now we exploit the assumption (ii) of the Lemma.
For every edge $(k,l)$, $k,l\in [d_A]$ of the coherence graph of $\mathcal{C}$ w.r.t.\ the basis $\{\ket{i}\}_{i\in[d_A]}$, we have a state $\sigma\otimes\rho\in \mathcal{C}$, for which $\bra{k}\sigma\ket{l}\neq 0$ and the output state of the channel is pure.
From the above argument, it then follows that for the corresponding pair $(k,l)$ we have that $\bm\lambda^k\eutp \bm\lambda^l$.
Finally, since the coherence graph of $\mathcal{C}$ is assumed to be connected, all the vectors $\{\bm\lambda^i\}_{i\in [d_A]}$ are equal up to a complex phase.
This finishes the proof.
\end{proof}

\begin{techcorollary}\label{corr:unitarity}
    Let $\Lambda\in\CPTP(\mathcal{H})$ be a channel that maps an \ac{ONB} $\{\ket{\psi_i}\}_{i\in[d]}$ of a $d$-dimensional Hilbert space $\HH$ to another \ac{ONB} $\{\ket{\phi_i}\}_{i\in[d]}$ of $\HH$.
    Let $\mathcal{C}\subseteq\mathcal{D}(\HH)$ be a finite set of states such that 
    \begin{enumerate}
        \item[(i)] $\Lambda(\rho)$ is pure for all $\rho\in\mathcal{C}$, and
        \item[(ii)] $\mathcal{C}$ is fully coherent with respect to the \ac{ONB} $\{\ket{\psi_i}\}_{i\in[d]}$ (see \cref{def:coherence_graph}). 
    \end{enumerate}
    Then the channel $\Lambda$ is unitary.
\end{techcorollary}
\begin{proof}
First, we notice that we can prove this Corollary for $\ket{\phi_i}=\ket{\psi_i}$, and a more general case follows by composing the channel $\Lambda$ with the unitary $\sum_{i\in[d]} \ket{\psi_i}\bra{\phi_i}$, since it does not affect the unitarity.
Then this result follows from \cref{lemma:unitarity_from_subchannels} by taking $\HH_B=\CC$. 
\end{proof}

\begin{techlemma}\label{lemma:channel=U}
    Let $U$ be a unitary on a Hilbert space $\mathcal{H}$ with eigenbasis $\{\ket{\psi_i}\}_{i\in[d]}$ and $\Lambda\in\CPTP(\mathcal{H})$ a channel such that $\Lambda(\psi_i)=\psi_i$ for all $i$. 
    Let there exists a set of pure states $\mathcal{C}\subseteq \DM(\HH)$ such that 
    \begin{enumerate}
        \item[(i)] $\Lambda(\rho) = U\rho U^\dagger$, $\forall \rho\in\mathcal{C}$, and
        \item[(ii)] $\mathcal{C}$ is fully coherent with respect to the \ac{ONB} $\{\ket{\psi_i}\}_{i\in[d]}$ (see \cref{def:coherence_graph}).
    \end{enumerate}
     Then $\Lambda(\rho)=U\rho U^\dagger$ for all $\rho\in\DM(\HH)$.
\end{techlemma}
\begin{proof}
The fact that $\Lambda(\psi_i)=\psi_i$, $\forall i\in [d]$ and that $U\phi U^\dagger=\Lambda(\phi)$ for all $\phi\in\mathcal{C}$ together with the coherence condition show that $\Lambda$ is unitary by \cref{corr:unitarity}. 
Let $V\in \mathrm{U}(\HH)$ be the corresponding unitary operator, i.e., $\Lambda(\argdot) \eqqcolon V(\argdot)V^\dagger$.
We know that $V$ is diagonal in the eigenbasis $\{\ket{\psi_i}\}_{i\in[d]}$ of $U$, since $\Lambda(\psi_i) = \psi_i$, $\forall i\in [d]$. 
Let $\{\mu_i\}_{i\in[d]}$ and $\{\lambda_i\}_{i\in[d]}$ be the eigenvalues of the operators $\overline{\bra{\psi_0}U\ket{\psi_0}}U$ and $\overline{\bra{\psi_0}V\ket{\psi_0}}V$, respectively, i.e., $\lambda_0=\mu_0=1$.
Then, if we prove that $\lambda_i=\mu_i$ for $i\in\{1,2,\dots,d-1\}$, we prove that $U\eutp V$, from which the claim of the lemma follows.

Let $(k,l)$ for $k,l\in [d]$ be an edge of the coherence graph of $\mathcal{C}$ (see \cref{def:coherence_graph}), and let $\rho=\sum_{i,j\in[d]} c_{ij}\kb{\psi_i}{\psi_j}\in\mathcal{C}$ be a state in $\mathcal{C}$ such that $c_{kl}\neq 0$.
The condition $\Lambda(\rho) = U\rho U^\dagger$ for that $\rho$ leads to,
\begin{equation}
    \sum_{i,j\in[d]} c_{ij}\mu_i\bar{\mu}_j \kb{\psi_i}{\psi_j} =
    \sum_{i,j\in [d]} c_{ij} \lambda_i\bar{\lambda}_j\kb{\psi_i}{\psi_j},  
\end{equation}
or equivalently $c_{ij}\mu_i\bar{\mu}_j=c_{ij}\lambda_i\bar{\lambda}_j$, for all $i,j\in [d]$, and, in particular $\mu_k\bar{\mu}_l = \lambda_k\bar{\lambda}_l$, since $c_{kl}\neq 0$.
For the edges $(k,0)$ that connect vertex ``$0$'', we then obtain that simply $\mu_k=\lambda_k$ must hold, since $\mu_0=\lambda_0$.
Consequently, for the edges $(k,l)$ that connect a neighbor ``$k$'' of vertex ``$0$'', we get $\mu_l=\lambda_l$.
Since the graph is connected, we then conclude that $\lambda_i=\mu_i$ for all $i\in [d]$, which finishes the proof.
\end{proof}

\begin{techlemma}\label{lemma:orthogonality}
    Let $\Lambda\in\CPTP(\HH_1,\HH_2)$ be any quantum channel. 
    For any $\rho,\sigma\in \DM(\HH_1)$, such that $\Tr[\Lambda(\rho)\Lambda(\sigma)]=0$ it necessarily holds that $\Tr[\rho\sigma]=0$.
\end{techlemma}
\begin{proof}
    Let $d\coloneqq \dim(\HH_1)$, $\{K_j\}_{j=0}^{n-1}$ be a set of Kraus operators of the channel $\Lambda$, where $n$ is the Kraus rank, and $\rho=\sum_{k\in[d]} \lambda_k\psi_k$, $\sigma=\sum_{l\in[d]}\mu_l\phi_l$ the respective spectral decompositions of $\rho$ and $\sigma$.
    We can express the condition of the lemma as
    \begin{equation}\label{app:eq:lemma_orthogonality}
        0=\Tr[\Lambda(\rho)\Lambda(\sigma)]=\sum_{i,j\in[n]} \Tr[K_i\rho K_i^\dagger K_j\sigma K_j^\dagger] =\sum_{i,j}\sum_{k,l\in[d]}\lambda_k\mu_l\Tr[K_i\psi_k K_i^\dagger K_j\phi_l K_j^\dagger]=\sum_{i,j,k,l}\lambda_k\mu_l\abs{\bra{\psi_k} K_i^\dagger K_j\ket{\phi_l}}^2,
    \end{equation}
    where we omitted the summation limits for brevity.
    From the above it follows that for all $l,k$ with non-zero $\mu_l,\lambda_k$, we must have $\bra{\psi_k}K_i^\dagger K_j\ket{\phi_l}=0$ for all $i$ and $j$, and, in particular, for $i=j$, which results in
    \begin{equation}
        0 = \sum_{i}\bra{\psi_k} K_i^\dagger K_i\ket{\phi_l}=\bra{\psi_k}\sum_{i} K_i^\dagger K_i\ket{\phi_l}=\braket{\psi_k}{\phi_l}.
    \end{equation}
    Therefore, it follows that $\Tr[\rho\sigma]=0$, which finishes the proof.
\end{proof}

\begin{techlemma}\label{lemma:marginal_homomorphism}
    Let $\mathcal{H}_A,\mathcal{H}_B$ be $d_A$- and $d_B$-dimensional Hilbert spaces, respectively, $\mathcal{W}=\{\ket{\psi_i}\}_{i\in[d_A]}$ an orthonormal basis of $\mathcal{H}_A$, and consider a subset of channels 
    \begin{equation}\label{eq:omega}
    \Omega_\mathcal{W} \coloneqq \Set*{\Lambda\in \CPTP(\mathcal{H}_A\otimes \mathcal{H}_B) \given \tr_B\left[\Lambda\left(\psi_i\otimes\rho\right)\right] = \psi_i,\phantom{\Big|} \forall \ket{\psi_i}\in \mathcal{W},\rho\in\DM(\HH_B)}.
    \end{equation}
    Then the following holds.
    \begin{enumerate}
        \item For any choice of Kraus operators $\{K_j\}_{j=0}^{n-1}$ of $\Lambda\in\Omega_{\mathcal{W}}$, there exist suitable $K_j^i\in \LL(\HH_B)$, $i\in[d_A]$ such that \begin{equation}\label{eq:Kraus_form_lemma_intermediate}
       K_j = \sum_{i\in [d_A]} \psi_i\otimes K_j^i.
    \end{equation}
    \item The following interchanging property between $\Lambda$ and the subchannels $\Lambda\vert_B^{\psi_i}$ holds
\begin{equation}\label{eq:important_subchannel_property}
\Lambda(\psi_i\otimes\rho)=\psi_i\otimes \Lambda\vert_B^{\psi_i}(\rho), \quad \forall \rho \in \DM(\HH_B),
\end{equation}
\item The set $ \Omega_\mathcal{W}$ is closed under the composition, and the maps 
    \begin{equation}\label{app:eq:lemma_hom}
    \argdot\vert_B^{\psi_i}:\Omega_\mathcal{W}\to \CPTP(\mathcal{H}_B),\; \Lambda\mapsto \Lambda\vert_B^{\psi_i}
    \end{equation} 
    which map the channels to their respective subchannels (see \cref{def:subchannels}) are homomorphisms.
    \end{enumerate}
\end{techlemma}
Note that it suffices to check the condition inside \cref{eq:omega} on an ONB of $\HH_B$ to conclude that the statement holds also for every $\rho\in\DM(\HH_B)$, as we will see in the proof. 
Although the second and third statements are essentially simple corollaries of the first, due to their importance for the proofs, we state them explicitly.
\begin{proof}
We abbreviate $\ket{i}=\ket{\psi_i}$ wherever it is convenient. 
Let $\Lambda\in\Omega_{\mathcal{W}}$, and let $\{K_j\}_{j=1}^{n}$ be the Kraus operators of $\Lambda$, where $n\in \NN$ is the Kraus rank of $\Lambda$. 
We can express each $K_j$ as
\begin{equation}\label{app:eq:lemma_marg_hom_kraus}
    K_j 
    = \sum_{k,l\in [d_A]} \kb{k}{l}\otimes  K_j^{kl},
\end{equation}
with $K_j^{kl}\in \LL(\HH_B)$, for all $k,l\in [d_A]$.
Pick a basis $\mathcal{V}=\{\ket{\phi_i}\}_{i\in[d_B]}$ for $\HH_B$ and let $\ket{\phi}\in\mathcal{V}$ arbitrary and fixed. We compute
\begin{align}\label{app:eq:lemma_marg_hom_1}
    \Lambda(\kb{i}{i}\otimes\phi)=&\sum_{j=1}^n K_j \kb{i}{i}\otimes\phi K_j^\dagger = \sum_j \sum_{k,l,k',l'\in[d_A]} \left(\kb{k}{l}\otimes K_j^{kl}\right)\left(\kb{i}{i}\otimes\phi\right)\left( \kb{l'}{k'}\otimes(K_j^{k'l'})^\dagger\right)\\
 =& \sum_{j,k,l,k',l'}\ket{k}\braket{l}{i}\braket{i}{l'}\bra{k'}\otimes K_j^{kl}\phi(K_j^{k'l'})^\dagger = \sum_{j,k,k'}\kb{k}{k'}\otimes K_j^{ki}\phi(K_j^{k'i})^\dagger.
\end{align}
We apply the partial trace $\tr_B[\argdot]$ to both sides of \cref{app:eq:lemma_marg_hom_1}, and since $\tr_B[\Lambda(\kb{i}{i}\otimes\phi)]=\kb{i}{i}$ by the assumption $\Lambda\in\Omega_{\mathcal{W}}$, we arrive at
\begin{equation}
    \sum_{j,k,k'}\tr[K_j^{ki}\phi (K_j^{k'i})^\dagger]\kb{k}{k'}=\kb{i}{i},
\end{equation}
which immediately implies that 
\begin{equation}
    \sum_{j=1}^n \tr[K_j^{ki}\phi(K_j^{ki})^\dagger]=0,\text{ for all }k\neq i.
\end{equation}
At the same time $K_j^{ki}\phi (K_j^{ki})^\dagger$ is positive semi-definite, so it has to be that  $K_j^{ki}\phi(K_j^{ki})^\dagger = 0$ for all $k\neq i$. 
Since $\phi\in\mathcal{V}$ was arbitrary, we conclude that $K_j^{ki}=0$ if $k\neq i$.
Thus, the decomposition of the Kraus operator $K_j$ in \cref{app:eq:lemma_marg_hom_kraus} simplifies as  $K_j = \sum_{i\in[d_A]}\psi_i\otimes K_j^{i}$, where $K_j^i\coloneqq K_j^{ii}$, and we get for arbitrary $\rho\in \mathcal{D}(\mathcal{H}_B)$ that
\begin{equation}
\Lambda(\psi_i\otimes\rho)=\psi_i\otimes \sum_{j=1}^n K_j^i\rho (K_j^i)^\dagger.
\end{equation}
With the above identity, the second claim follows immediately by the definition of subchannels (\cref{def:subchannels}). From here, it is also straightforward to conclude that $\Omega_W$ is closed under the composition of channels, i.e., for any $\Lambda,\Lambda'\in\Omega_{\mathcal{W}}$, it holds that $\Lambda\circ\Lambda'\in\Omega_{\mathcal{W}}$.
Finally, we can easily use the identity \cref{eq:important_subchannel_property} to show that the maps in \cref{app:eq:lemma_hom} are homomorphisms.
For $\Lambda,\Lambda'\in\Omega_{\mathcal{W}}$, we calculate
\begin{equation}
(\Lambda\circ\Lambda')\vert_B^{\psi_i}(\rho) = \tr_A[\Lambda\circ\Lambda'(\psi_i\otimes\rho)]=\tr_A\left[\Lambda\left(\psi_i\otimes \Lambda'\vert_B^{\psi_i}(\rho)\right)\right]=\Lambda\vert_B^{\psi_i}\circ \Lambda'\vert_B^{\psi_i}(\rho),
\end{equation}
for arbitrary $\rho\in\DM(\HH_B)$.
This finishes the proof.
\end{proof}

\begin{techlemma}\label{lemma:ONB_channel_purity}
    A channel $\Lambda\in\CPTP(\mathcal{H})$ that maps an \ac{ONB}   $\{\ket{\psi_i}\}_{i\in[d]}$ of a $d$-dimensional Hilbert space $\HH$ to another \ac{ONB} $\{\ket{\phi_i}\}_{i\in[d]}$ of $\HH$ cannot increase the purity of any input quantum state.
    \end{techlemma}
\begin{proof}
Without loss of generality, let us take $\ket{\phi_i} = \ket{\psi_i}, \forall i\in [d]$, otherwise we can apply a unitary $\sum_{i\in[d]} \kb{\psi_i}{\phi_i}$ after $\Lambda$, which does not affect the purity of any output state. 
Since $\Lambda$ acts as the identity on $\psi_i$ for all $i\in[d]$, the Kraus operators $\{K_j\}_{j=0}^{n-1}$ of the channel $\Lambda$ must be diagonal in the basis $\{\ket{\psi_i}\}_{i\in[d]}$, where $n\in \NN$ is the Kraus rank of $\Lambda$.
This is a simple corollary of \cref{lemma:marginal_homomorphism} for the special case of $\mathcal{H}_B=\mathbb{C}$.
Eq.~\eqref{eq:Kraus_form_lemma_intermediate}.
Let therefore $K_j = \sum_{i\in[d]} \lambda_j^{i}\psi_i$ for appropriately chosen $\lambda_j^{i}\in \CC$.
From the normalization condition of $\Lambda$, we get that the vectors $\bm\lambda^i=(\lambda^i_1,\dots,\lambda^i_n)$ are normalized. 
We can then calculate the purity of the output state for an arbitrary $\rho\in\DM(\HH)$,
\begin{align}\begin{split}
    \tr[\Lambda(\rho)^2] &= \sum_{j,j'\in [n]}\,\sum_{i,k,i',k'\in[d]}\tr\left[(\lambda^i_j\psi_i \rho \bar{\lambda}^k_j \psi_k)(\lambda^{i'}_{j'}\psi_{i'}\rho\bar{\lambda}^{k'}_{j'}\psi_{k'})\right]
    =\sum_{i,k\in [d]}\underbrace{\abs{\langle \bm\lambda^i,\bm\lambda^k\rangle}^2}_{\leq 1}\underbrace{\tr\left[\psi_i\rho \psi_k \rho\right]}_{\geq 0} \leq \tr[\rho^2],
\end{split}\end{align}
where in the last step we used the fact that $\{\ket{\psi_i}\}_{i\in[d]}$ is an \ac{ONB}.
\end{proof}

\begin{techlemma}\label{lemma:coherent_state}
Let $\Lambda\in\CPTP(\HH_A\otimes \HH_B)$ be a channel on a composite system described by Hilbert spaces $\HH_A$ and $\HH_B$, and let $\{\ket{\psi_i}\}_{i\in[d_A]}$, $\{\ket{\phi_l}\}_{l\in[d_B]}$ be \acp{ONB} of $\mathcal{H}_A$ and $\mathcal{H}_B$, respectively. 
Assume that $\Lambda$ admits a Kraus representation with Kraus operators $\{K_j\}_{j=0}^{n-1}$, which decompose as
\begin{equation}
    K_j = \sum_{k\in[d_A]}\lambda^{(k)}_j \psi_k\otimes U_k\ ,
\end{equation}
with $\lambda^{(k)}_j\in \CC$, and where the unitaries $U_k\in \U(\HH_B)$ are such that $\bra{\phi_l}U_k\ket{\phi_0}\neq 0, \forall l\in[d_B], k\in [d_A]$. 
Let $\rho\in\DM(\mathcal{H}_A)$ and $\mathcal{I}\subseteq [d_A]$ be such that $\tr[\psi_i\rho]>0$ for all $i\in \mathcal{I}$. 
Then we have
\begin{equation}
\tr\Big[\psi_i\otimes\phi_l \Lambda(\rho\otimes\phi_0)\Big]>0\quad \forall i\in \mathcal{I},l\in[d_B].
\end{equation}
\end{techlemma}
\begin{proof}
We abbreviate $\ket{\psi_i}\otimes\ket{\phi_l}=:\ket{i,l}$. 
Fix $i\in\mathcal{I}$, $l\in[d_B]$, and verify that
\begin{align}\begin{split}
    \tr\left[\kb{i,l}{i,l} \Lambda(\rho\otimes\kb{0}{0})\right]& =\tr\left[\Lambda^\dagger(\kb{i,l}{i,l})\rho\otimes\kb{0}{0}\right] = \sum_{j\in [n]}\Tr\left[K^\dagger_j\kb{i,l}{i,l}K_j\rho\otimes\kb{0}{0}\right]\\
    &=\sum_{j\in [n]}\sum_{k,k'\in[d_A]}\overline{\lambda_j^{(k')}}\lambda_j^{(k)}\tr\left[\left(\kb{k'}{k'}\otimes U_{k'}^\dagger\right)\kb{i,l}{i,l} \left(\kb{k}{k}\otimes  U_k\right) \rho\otimes\kb{0}{0} \right]\\
    &=\sum_{j\in [n]}\overline{\lambda_j^{(i)}}\lambda_j^{(i)}\tr\left[\left(\kb{i}{i}\otimes U_i^\dagger\kb{l}{l}U_i\right)\left(\rho\otimes\kb{0}{0}\right)\right]=\tr\left(\kb{i}{i}\rho\right)\abs{\bra{l}U_i \ket{0}}^2>0\ ,
\end{split}\end{align}
where we used the fact that $\sum_{j\in [n]}\abs*{\lambda^{(i)}_j}^2 = 1$ for all $i\in [d_A]$ due to the normalization of $\Lambda$.
\end{proof}

\begin{techlemma}\label{lemma:dist}
    Let $\mathcal{M} = (\ket{\psi},\Set{U_x}_{x\in\XX},\Set{M_a}_{a\in\AA})$ and $\mathcal{N} = (\rho,\Set{\Lambda_x}_{x\in\XX},\Set{N_a}_{a\in\AA})$ be two models defined over $\HH\cong\CC^d$, with $\AA=\{1,2,\dots,d\}$ and $\Set{M_a}_{a\in \AA}$ being rank-$1$ projective measurement.
    Define the following quantities
    \begin{align}
        \d_1&\coloneqq 1-\Fid(\rho,\kb{\psi}{\psi}), &  \d_2^{(x)}&\coloneqq 1-\Fid_\avg(\Lambda_x,U_x),\;\forall x\in \XX,\\
        \d_3&\coloneqq \frac{1}{2}\max_{\sigma\in\DM(\HH)}\Set*{\sum_{a\in\AA}\abs*{\Tr[\sigma(M_a-N_a)]}}, & \d'_3&\coloneqq 1-\frac{1}{d}\sum_{a\in\AA}\Tr[M_aN_a].
        \end{align}
        Then it holds that
        \begin{equation}\label{app:eq:dist_relation}
            \frac{1}{3}\left(\d_1+\frac{1}{\abs{X}}\sum_{x\in\XX}\d^{(x)}_2+\d'_3\right)\leq \max\Set{\mathrm{d_1},\d_2^{(x)},\d_3}\leq 3\max\Set{d,\abs{\XX}}\sqrt{\frac{1}{3}\left(\d_1+\frac{1}{\abs{X}}\sum_{x\in\XX}\d^{(x)}_2+\d'_3\right)}\ ,
        \end{equation}
        where the maximization is also taken over $x\in\XX$.
\end{techlemma}
\begin{proof}
    First, we prove that $\d'_3\leq d_3\leq \frac{d}{\sqrt{2}}\sqrt{\d'_3}$.
    Since $M_a$ is a rank-$1$ projective measurement, then we can take $\sigma=M_b$ for some $b\in \AA$, and get
    \begin{equation}
        \d_3\geq \frac{1}{2}\sum_{a\in\AA}\abs*{\Tr[M_b(M_a-N_a)]} = \frac{1}{2}\left(1-\Tr[M_bN_b]+\sum_{a\in\AA\setminus b}\Tr[M_bN_a]\right) = 1-\Tr[M_bN_b].
    \end{equation}
    Taking the average over $b\in \AA$, gives $\d_3\geq \d'_3$.
    For the upper bound, we use the relations between the norms
    \begin{equation}\label{app:eq:d3_relation}
        \d_3\leq \frac{1}{2}\sum_{a\in\AA}\norm{M_a-N_a}_\infty\leq \frac{1}{2}\sum_{a\in\AA}\sqrt{1+\Tr[N_a^2]-2\Tr[M_aN_a]}\leq \frac{d}{2}\sqrt{\frac{1}{d}\sum_{a\in\AA}\left(1+\Tr[N_a^2]-2\Tr[M_aN_a]\right)}\leq d\sqrt{\d'_3}\ .
    \end{equation}
    To prove the first inequality in \cref{app:eq:dist_relation}, we use $\d'_3\leq \d_3$ and a trivial upper bound $\d_1,\d_2^{(x)},\d_3\leq \max\Set{\d_1,\d_2^{(x)},\d_3}$ and take the average.
    For the second inequality in \cref{app:eq:dist_relation}, we use the fact that $\d_1,\d_2^{(x)}\in[0,1]$ together with \cref{app:eq:d3_relation} to obtain 
    \begin{equation}
        \max\Set{\d_1,\d_2^{(x)},\d_3}\leq \max\Set*{\sqrt{\d_1},\sqrt{\d_2^{(x)}},d\sqrt{\d'_3}}\leq 3\max\Set{d,\abs{\XX}}\left(\frac{1}{3}\sqrt{\d_1}+\frac{1}{3\abs{\XX}}\sum_{x\in\XX}\sqrt{\d_2^{(x)}}+\frac{1}{3}\sqrt{\d'_3}\right),
    \end{equation}
    from where the claim follows due to the concavity of the square root.
\end{proof}
\end{appendix}

\clearpage
\twocolumngrid
\bibliographystyle{myapsrev4-2}
\bibliography{ref,mk}

\begin{thebibliography}{52}%
\makeatletter
\providecommand \@ifxundefined [1]{%
 \@ifx{#1\undefined}
}%
\providecommand \@ifnum [1]{%
 \ifnum #1\expandafter \@firstoftwo
 \else \expandafter \@secondoftwo
 \fi
}%
\providecommand \@ifx [1]{%
 \ifx #1\expandafter \@firstoftwo
 \else \expandafter \@secondoftwo
 \fi
}%
\providecommand \natexlab [1]{#1}%
\providecommand \enquote  [1]{``#1''}%
\providecommand \bibnamefont  [1]{#1}%
\providecommand \bibfnamefont [1]{#1}%
\providecommand \citenamefont [1]{#1}%
\providecommand \href@noop [0]{\@secondoftwo}%
\providecommand \href [0]{\begingroup \@sanitize@url \@href}%
\providecommand \@href[1]{\@@startlink{#1}\@@href}%
\providecommand \@@href[1]{\endgroup#1\@@endlink}%
\providecommand \@sanitize@url [0]{\catcode `\\12\catcode `\$12\catcode `\&12\catcode `\#12\catcode `\^12\catcode `\_12\catcode `\%12\relax}%
\providecommand \@@startlink[1]{}%
\providecommand \@@endlink[0]{}%
\providecommand \url  [0]{\begingroup\@sanitize@url \@url }%
\providecommand \@url [1]{\endgroup\@href {#1}{\urlprefix }}%
\providecommand \urlprefix  [0]{URL }%
\providecommand \Eprint[0]{\href }%
\providecommand \doibase [0]{https://doi.org/}%
\providecommand \selectlanguage [0]{\@gobble}%
\providecommand \bibinfo  [0]{\@secondoftwo}%
\providecommand \bibfield  [0]{\@secondoftwo}%
\providecommand \translation [1]{[#1]}%
\providecommand \BibitemOpen [0]{}%
\providecommand \bibitemStop [0]{}%
\providecommand \bibitemNoStop [0]{.\EOS\space}%
\providecommand \EOS [0]{\spacefactor3000\relax}%
\providecommand \BibitemShut  [1]{\csname bibitem#1\endcsname}%
\let\auto@bib@innerbib\@empty
\bibitem [{\citenamefont {Eisert}\ \emph {et~al.}(2020)\citenamefont {Eisert}, \citenamefont {Hangleiter}, \citenamefont {Walk}, \citenamefont {Roth}, \citenamefont {Markham}, \citenamefont {Parekh}, \citenamefont {Chabaud},\ and\ \citenamefont {Kashefi}}]{Eisert2020Quantumcertification}%
  \BibitemOpen
  \bibfield  {author} {\bibinfo {author} {\bibfnamefont {J.}~\bibnamefont {Eisert}}, \bibinfo {author} {\bibfnamefont {D.}~\bibnamefont {Hangleiter}}, \bibinfo {author} {\bibfnamefont {N.}~\bibnamefont {Walk}}, \bibinfo {author} {\bibfnamefont {I.}~\bibnamefont {Roth}}, \bibinfo {author} {\bibfnamefont {D.}~\bibnamefont {Markham}}, \bibinfo {author} {\bibfnamefont {R.}~\bibnamefont {Parekh}}, \bibinfo {author} {\bibfnamefont {U.}~\bibnamefont {Chabaud}},\ and\ \bibinfo {author} {\bibfnamefont {E.}~\bibnamefont {Kashefi}},\ }\bibinfo {title} {\emph {Quantum certification and benchmarking}},\ \bibfield  {journal} {\bibinfo  {journal} {Nature Reviews Physics}\ }\textbf {\bibinfo {volume} {2}},\ \href {https://doi.org/10.1038/s42254-020-0186-4} {10.1038/s42254-020-0186-4} (\bibinfo {year} {2020})\BibitemShut {NoStop}%
\bibitem [{\citenamefont {Kliesch}\ and\ \citenamefont {Roth}(2021)}]{kliesch2021theory}%
  \BibitemOpen
  \bibfield  {author} {\bibinfo {author} {\bibfnamefont {M.}~\bibnamefont {Kliesch}}\ and\ \bibinfo {author} {\bibfnamefont {I.}~\bibnamefont {Roth}},\ }\bibinfo {title} {\emph {Theory of quantum system certification}},\ \href {https://link.aps.org/doi/10.1103/PRXQuantum.2.010201} {\bibfield  {journal} {\bibinfo  {journal} {PRX Quantum}\ }\textbf {\bibinfo {volume} {2}} (\bibinfo {year} {2021})}\BibitemShut {NoStop}%
\bibitem [{\citenamefont {Liu}\ \emph {et~al.}(2020)\citenamefont {Liu}, \citenamefont {Shang}, \citenamefont {Yu},\ and\ \citenamefont {Zhang}}]{liu2020efficient}%
  \BibitemOpen
  \bibfield  {author} {\bibinfo {author} {\bibfnamefont {Y.-C.}\ \bibnamefont {Liu}}, \bibinfo {author} {\bibfnamefont {J.}~\bibnamefont {Shang}}, \bibinfo {author} {\bibfnamefont {X.-D.}\ \bibnamefont {Yu}},\ and\ \bibinfo {author} {\bibfnamefont {X.}~\bibnamefont {Zhang}},\ }\bibinfo {title} {\emph {Efficient verification of quantum processes}},\ \href {https://doi.org/10.1103/PhysRevA.101.042315} {\bibfield  {journal} {\bibinfo  {journal} {Physical Review A}\ }\textbf {\bibinfo {volume} {101}},\ \bibinfo {pages} {042315} (\bibinfo {year} {2020})},\ \Eprint{https://arxiv.org/abs/1910.13730} {arXiv:1910.13730 [quant-ph]}\BibitemShut {NoStop}%
\bibitem [{\citenamefont {Mohseni}\ \emph {et~al.}(2008)\citenamefont {Mohseni}, \citenamefont {Rezakhani},\ and\ \citenamefont {Lidar}}]{mohseni2008quantum}%
  \BibitemOpen
  \bibfield  {author} {\bibinfo {author} {\bibfnamefont {M.}~\bibnamefont {Mohseni}}, \bibinfo {author} {\bibfnamefont {A.~T.}\ \bibnamefont {Rezakhani}},\ and\ \bibinfo {author} {\bibfnamefont {D.~A.}\ \bibnamefont {Lidar}},\ }\bibinfo {title} {\emph {Quantum-process tomography: Resource analysis of different strategies}},\ \href {https://doi.org/10.1103/PhysRevA.77.032322} {\bibfield  {journal} {\bibinfo  {journal} {Phys. Rev. A}\ }\textbf {\bibinfo {volume} {77}},\ \bibinfo {pages} {032322} (\bibinfo {year} {2008})},\ \Eprint{https://arxiv.org/abs/quant-ph/0702131} {arXiv:quant-ph/0702131 [quant-ph]}\BibitemShut {NoStop}%
\bibitem [{\citenamefont {Emerson}\ \emph {et~al.}(2005)\citenamefont {Emerson}, \citenamefont {Alicki},\ and\ \citenamefont {\.{Z}yczkowski}}]{EmeAliZyc05}%
  \BibitemOpen
  \bibfield  {author} {\bibinfo {author} {\bibfnamefont {J.}~\bibnamefont {Emerson}}, \bibinfo {author} {\bibfnamefont {R.}~\bibnamefont {Alicki}},\ and\ \bibinfo {author} {\bibfnamefont {K.}~\bibnamefont {\.{Z}yczkowski}},\ }\bibinfo {title} {\emph {Scalable noise estimation with random unitary operators}},\ \href {https://doi.org/10.1088/1464-4266/7/10/021} {\bibfield  {journal} {\bibinfo  {journal} {J. Opt. B}\ }\textbf {\bibinfo {volume} {7}},\ \bibinfo {pages} {S347} (\bibinfo {year} {2005})},\ \Eprint{https://arxiv.org/abs/arXiv:quant-ph/0503243} {arXiv:quant-ph/0503243}\BibitemShut {NoStop}%
\bibitem [{\citenamefont {{Levi}}\ \emph {et~al.}(2007)\citenamefont {{Levi}}, \citenamefont {{Lopez}}, \citenamefont {{Emerson}},\ and\ \citenamefont {{Cory}}}]{LevLopEme07}%
  \BibitemOpen
  \bibfield  {author} {\bibinfo {author} {\bibfnamefont {B.}~\bibnamefont {{Levi}}}, \bibinfo {author} {\bibfnamefont {C.~C.}\ \bibnamefont {{Lopez}}}, \bibinfo {author} {\bibfnamefont {J.}~\bibnamefont {{Emerson}}},\ and\ \bibinfo {author} {\bibfnamefont {D.~G.}\ \bibnamefont {{Cory}}},\ }\bibinfo {title} {\emph {Efficient error characterization in quantum information processing}},\ \href {https://doi.org/10.1103/PhysRevA.75.022314} {\bibfield  {journal} {\bibinfo  {journal} {Phys.\ Rev.\ A}\ }\textbf {\bibinfo {volume} {75}},\ \bibinfo {eid} {022314} (\bibinfo {year} {2007})},\ \Eprint{https://arxiv.org/abs/quant-ph/0608246} {arXiv:quant-ph/0608246 [quant-ph]}\BibitemShut {NoStop}%
\bibitem [{\citenamefont {Emerson}\ \emph {et~al.}(2007)\citenamefont {Emerson}, \citenamefont {Silva}, \citenamefont {Moussa}, \citenamefont {Ryan}, \citenamefont {Laforest}, \citenamefont {Baugh}, \citenamefont {Cory},\ and\ \citenamefont {Laflamme}}]{EmeSilMou07}%
  \BibitemOpen
  \bibfield  {author} {\bibinfo {author} {\bibfnamefont {J.}~\bibnamefont {Emerson}}, \bibinfo {author} {\bibfnamefont {M.}~\bibnamefont {Silva}}, \bibinfo {author} {\bibfnamefont {O.}~\bibnamefont {Moussa}}, \bibinfo {author} {\bibfnamefont {C.}~\bibnamefont {Ryan}}, \bibinfo {author} {\bibfnamefont {M.}~\bibnamefont {Laforest}}, \bibinfo {author} {\bibfnamefont {J.}~\bibnamefont {Baugh}}, \bibinfo {author} {\bibfnamefont {D.~G.}\ \bibnamefont {Cory}},\ and\ \bibinfo {author} {\bibfnamefont {R.}~\bibnamefont {Laflamme}},\ }\bibinfo {title} {\emph {Symmetrized characterization of noisy quantum processes}},\ \href {https://doi.org/10.1126/science.1145699} {\bibfield  {journal} {\bibinfo  {journal} {Science}\ }\textbf {\bibinfo {volume} {317}},\ \bibinfo {pages} {1893} (\bibinfo {year} {2007})},\ \Eprint{https://arxiv.org/abs/0707.0685} {arXiv:0707.0685 [quant-ph]}\BibitemShut {NoStop}%
\bibitem [{\citenamefont {{Knill}}\ \emph {et~al.}(2008)\citenamefont {{Knill}}, \citenamefont {{Leibfried}}, \citenamefont {{Reichle}}, \citenamefont {{Britton}}, \citenamefont {{Blakestad}}, \citenamefont {{Jost}}, \citenamefont {{Langer}}, \citenamefont {{Ozeri}}, \citenamefont {{Seidelin}},\ and\ \citenamefont {{Wineland}}}]{KniLeiRei08}%
  \BibitemOpen
  \bibfield  {author} {\bibinfo {author} {\bibfnamefont {E.}~\bibnamefont {{Knill}}}, \bibinfo {author} {\bibfnamefont {D.}~\bibnamefont {{Leibfried}}}, \bibinfo {author} {\bibfnamefont {R.}~\bibnamefont {{Reichle}}}, \bibinfo {author} {\bibfnamefont {J.}~\bibnamefont {{Britton}}}, \bibinfo {author} {\bibfnamefont {R.~B.}\ \bibnamefont {{Blakestad}}}, \bibinfo {author} {\bibfnamefont {J.~D.}\ \bibnamefont {{Jost}}}, \bibinfo {author} {\bibfnamefont {C.}~\bibnamefont {{Langer}}}, \bibinfo {author} {\bibfnamefont {R.}~\bibnamefont {{Ozeri}}}, \bibinfo {author} {\bibfnamefont {S.}~\bibnamefont {{Seidelin}}},\ and\ \bibinfo {author} {\bibfnamefont {D.~J.}\ \bibnamefont {{Wineland}}},\ }\bibinfo {title} {\emph {Randomized benchmarking of quantum gates}},\ \href {https://doi.org/10.1103/PhysRevA.77.012307} {\bibfield  {journal} {\bibinfo  {journal} {Phys. Rev. A}\ }\textbf {\bibinfo {volume} {77}},\ \bibinfo {eid} {012307} (\bibinfo {year} {2008})},\ \Eprint{https://arxiv.org/abs/0707.0963} {arXiv:0707.0963 [quant-ph]}\BibitemShut {NoStop}%
\bibitem [{\citenamefont {Dankert}\ \emph {et~al.}(2009)\citenamefont {Dankert}, \citenamefont {Cleve}, \citenamefont {Emerson},\ and\ \citenamefont {Livine}}]{DanCleEme09}%
  \BibitemOpen
  \bibfield  {author} {\bibinfo {author} {\bibfnamefont {C.}~\bibnamefont {Dankert}}, \bibinfo {author} {\bibfnamefont {R.}~\bibnamefont {Cleve}}, \bibinfo {author} {\bibfnamefont {J.}~\bibnamefont {Emerson}},\ and\ \bibinfo {author} {\bibfnamefont {E.}~\bibnamefont {Livine}},\ }\bibinfo {title} {\emph {Exact and approximate unitary 2-designs and their application to fidelity estimation}},\ \href {https://doi.org/10.1103/PhysRevA.80.012304} {\bibfield  {journal} {\bibinfo  {journal} {Phys. Rev. A}\ }\textbf {\bibinfo {volume} {80}},\ \bibinfo {pages} {012304} (\bibinfo {year} {2009})},\ \Eprint{https://arxiv.org/abs/quant-ph/0606161} {arXiv:quant-ph/0606161 [quant-ph]}\BibitemShut {NoStop}%
\bibitem [{\citenamefont {Helsen}\ \emph {et~al.}(2022)\citenamefont {Helsen}, \citenamefont {Roth}, \citenamefont {Onorati}, \citenamefont {Werner},\ and\ \citenamefont {Eisert}}]{Helsen2022framework}%
  \BibitemOpen
  \bibfield  {author} {\bibinfo {author} {\bibfnamefont {J.}~\bibnamefont {Helsen}}, \bibinfo {author} {\bibfnamefont {I.}~\bibnamefont {Roth}}, \bibinfo {author} {\bibfnamefont {E.}~\bibnamefont {Onorati}}, \bibinfo {author} {\bibfnamefont {A.}~\bibnamefont {Werner}},\ and\ \bibinfo {author} {\bibfnamefont {J.}~\bibnamefont {Eisert}},\ }\bibinfo {title} {\emph {General framework for randomized benchmarking}},\ \bibfield  {journal} {\bibinfo  {journal} {PRX Quantum}\ }\textbf {\bibinfo {volume} {3}},\ \href {https://doi.org/10.1103/PRXQuantum.3.020357} {10.1103/PRXQuantum.3.020357} (\bibinfo {year} {2022})\BibitemShut {NoStop}%
\bibitem [{\citenamefont {{Heinrich}}\ \emph {et~al.}(2022)\citenamefont {{Heinrich}}, \citenamefont {{Kliesch}},\ and\ \citenamefont {{Roth}}}]{Heinrich22GeneralGuarantees}%
  \BibitemOpen
  \bibfield  {author} {\bibinfo {author} {\bibfnamefont {M.}~\bibnamefont {{Heinrich}}}, \bibinfo {author} {\bibfnamefont {M.}~\bibnamefont {{Kliesch}}},\ and\ \bibinfo {author} {\bibfnamefont {I.}~\bibnamefont {{Roth}}},\ }\href {https://doi.org/10.48550/arXiv.2212.06181} {\bibinfo {title} {\emph {Randomized benchmarking with random quantum circuits}}},\ \Eprint{https://arxiv.org/abs/2212.06181} {arXiv:2212.06181 [quant-ph]} (\bibinfo {year} {2022})\BibitemShut {NoStop}%
\bibitem [{\citenamefont {Merkel}\ \emph {et~al.}(2013)\citenamefont {Merkel}, \citenamefont {Gambetta}, \citenamefont {Smolin}, \citenamefont {Poletto}, \citenamefont {C{\'o}rcoles}, \citenamefont {Johnson}, \citenamefont {Ryan},\ and\ \citenamefont {Steffen}}]{merkel2013self}%
  \BibitemOpen
  \bibfield  {author} {\bibinfo {author} {\bibfnamefont {S.~T.}\ \bibnamefont {Merkel}}, \bibinfo {author} {\bibfnamefont {J.~M.}\ \bibnamefont {Gambetta}}, \bibinfo {author} {\bibfnamefont {J.~A.}\ \bibnamefont {Smolin}}, \bibinfo {author} {\bibfnamefont {S.}~\bibnamefont {Poletto}}, \bibinfo {author} {\bibfnamefont {A.~D.}\ \bibnamefont {C{\'o}rcoles}}, \bibinfo {author} {\bibfnamefont {B.~R.}\ \bibnamefont {Johnson}}, \bibinfo {author} {\bibfnamefont {C.~A.}\ \bibnamefont {Ryan}},\ and\ \bibinfo {author} {\bibfnamefont {M.}~\bibnamefont {Steffen}},\ }\bibinfo {title} {\emph {Self-consistent quantum process tomography}},\ \href {https://doi.org/10.1103/PhysRevA.87.062119} {\bibfield  {journal} {\bibinfo  {journal} {Physical Review A}\ }\textbf {\bibinfo {volume} {87}},\ \bibinfo {pages} {062119} (\bibinfo {year} {2013})},\ \Eprint{https://arxiv.org/abs/1211.0322} {arXiv:1211.0322 [quant-ph]}\BibitemShut {NoStop}%
\bibitem [{\citenamefont {{Blume-Kohout}}\ \emph {et~al.}(2013)\citenamefont {{Blume-Kohout}}, \citenamefont {{King Gamble}}, \citenamefont {{Nielsen}}, \citenamefont {{Mizrahi}}, \citenamefont {{Sterk}},\ and\ \citenamefont {{Maunz}}}]{BluGamNie13}%
  \BibitemOpen
  \bibfield  {author} {\bibinfo {author} {\bibfnamefont {R.}~\bibnamefont {{Blume-Kohout}}}, \bibinfo {author} {\bibfnamefont {J.}~\bibnamefont {{King Gamble}}}, \bibinfo {author} {\bibfnamefont {E.}~\bibnamefont {{Nielsen}}}, \bibinfo {author} {\bibfnamefont {J.}~\bibnamefont {{Mizrahi}}}, \bibinfo {author} {\bibfnamefont {J.~D.}\ \bibnamefont {{Sterk}}},\ and\ \bibinfo {author} {\bibfnamefont {P.}~\bibnamefont {{Maunz}}},\ }\href@noop {} {\bibinfo {title} {\emph {Robust, self-consistent, closed-form tomography of quantum logic gates on a trapped ion qubit}}},\ \Eprint{https://arxiv.org/abs/1310.4492} {arXiv:1310.4492 [quant-ph]} (\bibinfo {year} {2013})\BibitemShut {NoStop}%
\bibitem [{\citenamefont {Nielsen}\ \emph {et~al.}(2021)\citenamefont {Nielsen}, \citenamefont {Gamble}, \citenamefont {Rudinger}, \citenamefont {Scholten}, \citenamefont {Young},\ and\ \citenamefont {Blume-Kohout}}]{Nielsen2021gatesettomography}%
  \BibitemOpen
  \bibfield  {author} {\bibinfo {author} {\bibfnamefont {E.}~\bibnamefont {Nielsen}}, \bibinfo {author} {\bibfnamefont {J.~K.}\ \bibnamefont {Gamble}}, \bibinfo {author} {\bibfnamefont {K.}~\bibnamefont {Rudinger}}, \bibinfo {author} {\bibfnamefont {T.}~\bibnamefont {Scholten}}, \bibinfo {author} {\bibfnamefont {K.}~\bibnamefont {Young}},\ and\ \bibinfo {author} {\bibfnamefont {R.}~\bibnamefont {Blume-Kohout}},\ }\bibinfo {title} {\emph {Gate {S}et {T}omography}},\ \bibfield  {journal} {\bibinfo  {journal} {{Quantum}}\ }\textbf {\bibinfo {volume} {5}},\ \href {https://doi.org/10.22331/q-2021-10-05-557} {10.22331/q-2021-10-05-557} (\bibinfo {year} {2021})\BibitemShut {NoStop}%
\bibitem [{\citenamefont {Brieger}\ \emph {et~al.}(2023)\citenamefont {Brieger}, \citenamefont {Roth},\ and\ \citenamefont {Kliesch}}]{brieger2023compressive}%
  \BibitemOpen
  \bibfield  {author} {\bibinfo {author} {\bibfnamefont {R.}~\bibnamefont {Brieger}}, \bibinfo {author} {\bibfnamefont {I.}~\bibnamefont {Roth}},\ and\ \bibinfo {author} {\bibfnamefont {M.}~\bibnamefont {Kliesch}},\ }\bibinfo {title} {\emph {Compressive gate set tomography}},\ \bibfield  {journal} {\bibinfo  {journal} {PRX Quantum}\ }\textbf {\bibinfo {volume} {4}},\ \href {https://doi.org/10.1103/PRXQuantum.4.010325} {10.1103/PRXQuantum.4.010325} (\bibinfo {year} {2023})\BibitemShut {NoStop}%
\bibitem [{\citenamefont {{Proctor}}\ \emph {et~al.}(2017)\citenamefont {{Proctor}}, \citenamefont {{Rudinger}}, \citenamefont {{Young}}, \citenamefont {{Sarovar}},\ and\ \citenamefont {{Blume-Kohout}}}]{proctor2017WhatRandomizedBenchmarking}%
  \BibitemOpen
  \bibfield  {author} {\bibinfo {author} {\bibfnamefont {T.}~\bibnamefont {{Proctor}}}, \bibinfo {author} {\bibfnamefont {K.}~\bibnamefont {{Rudinger}}}, \bibinfo {author} {\bibfnamefont {K.}~\bibnamefont {{Young}}}, \bibinfo {author} {\bibfnamefont {M.}~\bibnamefont {{Sarovar}}},\ and\ \bibinfo {author} {\bibfnamefont {R.}~\bibnamefont {{Blume-Kohout}}},\ }\bibinfo {title} {\emph {What randomized benchmarking actually measures}},\ \href {https://doi.org/10.1103/PhysRevLett.119.130502} {\bibfield  {journal} {\bibinfo  {journal} {\prl}\ }\textbf {\bibinfo {volume} {119}},\ \bibinfo {eid} {130502} (\bibinfo {year} {2017})},\ \Eprint{https://arxiv.org/abs/1702.01853} {arXiv:1702.01853 [quant-ph]}\BibitemShut {NoStop}%
\bibitem [{\citenamefont {Mayers}\ and\ \citenamefont {Yao}(2004)}]{mayers2003self}%
  \BibitemOpen
  \bibfield  {author} {\bibinfo {author} {\bibfnamefont {D.}~\bibnamefont {Mayers}}\ and\ \bibinfo {author} {\bibfnamefont {A.}~\bibnamefont {Yao}},\ }\href@noop {} {\bibinfo {title} {\emph {Self testing quantum apparatus}}},\ \Eprint{https://arxiv.org/abs/quant-ph/0307205} {arXiv:quant-ph/0307205 [quant-ph]} (\bibinfo {year} {2004})\BibitemShut {NoStop}%
\bibitem [{\citenamefont {{\v{S}}upi{\'{c}}}\ and\ \citenamefont {Bowles}(2020)}]{Supic2020selftestingof}%
  \BibitemOpen
  \bibfield  {author} {\bibinfo {author} {\bibfnamefont {I.}~\bibnamefont {{\v{S}}upi{\'{c}}}}\ and\ \bibinfo {author} {\bibfnamefont {J.}~\bibnamefont {Bowles}},\ }\bibinfo {title} {\emph {Self-testing of quantum systems: a review}},\ \href {https://doi.org/10.22331/q-2020-09-30-337} {\bibfield  {journal} {\bibinfo  {journal} {{Quantum}}\ }\textbf {\bibinfo {volume} {4}},\ \bibinfo {pages} {337} (\bibinfo {year} {2020})}\BibitemShut {NoStop}%
\bibitem [{\citenamefont {Sekatski}\ \emph {et~al.}(2018)\citenamefont {Sekatski}, \citenamefont {Bancal}, \citenamefont {Wagner},\ and\ \citenamefont {Sangouard}}]{sekatski2018certifying}%
  \BibitemOpen
  \bibfield  {author} {\bibinfo {author} {\bibfnamefont {P.}~\bibnamefont {Sekatski}}, \bibinfo {author} {\bibfnamefont {J.-D.}\ \bibnamefont {Bancal}}, \bibinfo {author} {\bibfnamefont {S.}~\bibnamefont {Wagner}},\ and\ \bibinfo {author} {\bibfnamefont {N.}~\bibnamefont {Sangouard}},\ }\bibinfo {title} {\emph {Certifying the building blocks of quantum computers from {B}ell's theorem}},\ \href {https://doi.org/10.1103/PhysRevLett.121.180505} {\bibfield  {journal} {\bibinfo  {journal} {Phys. Rev. Lett.}\ }\textbf {\bibinfo {volume} {121}},\ \bibinfo {pages} {180505} (\bibinfo {year} {2018})}\BibitemShut {NoStop}%
\bibitem [{\citenamefont {Tsirelson}(1993)}]{Tsirelson93}%
  \BibitemOpen
  \bibfield  {author} {\bibinfo {author} {\bibfnamefont {B.~S.}\ \bibnamefont {Tsirelson}},\ }\href@noop {} {\bibfield  {journal} {\bibinfo  {journal} {Some results and problems on quantum {B}ell-type inequalities}\ }\textbf {\bibinfo {volume} {8}},\ \bibinfo {pages} {329 } (\bibinfo {year} {1993})}\BibitemShut {NoStop}%
\bibitem [{\citenamefont {Popescu}\ and\ \citenamefont {Rohrlich}(1992)}]{POPESCU1992which}%
  \BibitemOpen
  \bibfield  {author} {\bibinfo {author} {\bibfnamefont {S.}~\bibnamefont {Popescu}}\ and\ \bibinfo {author} {\bibfnamefont {D.}~\bibnamefont {Rohrlich}},\ }\bibinfo {title} {\emph {Which states violate bell's inequality maximally?}},\ \href {https://doi.org/https://doi.org/10.1016/0375-9601(92)90819-8} {\bibfield  {journal} {\bibinfo  {journal} {Physics Letters A}\ }\textbf {\bibinfo {volume} {169}},\ \bibinfo {pages} {411} (\bibinfo {year} {1992})}\BibitemShut {NoStop}%
\bibitem [{\citenamefont {Ac\'{\i}n}\ \emph {et~al.}(2012)\citenamefont {Ac\'{\i}n}, \citenamefont {Massar},\ and\ \citenamefont {Pironio}}]{Acin2012Randomness}%
  \BibitemOpen
  \bibfield  {author} {\bibinfo {author} {\bibfnamefont {A.}~\bibnamefont {Ac\'{\i}n}}, \bibinfo {author} {\bibfnamefont {S.}~\bibnamefont {Massar}},\ and\ \bibinfo {author} {\bibfnamefont {S.}~\bibnamefont {Pironio}},\ }\bibinfo {title} {\emph {Randomness versus nonlocality and entanglement}},\ \href {https://doi.org/10.1103/PhysRevLett.108.100402} {\bibfield  {journal} {\bibinfo  {journal} {Phys. Rev. Lett.}\ }\textbf {\bibinfo {volume} {108}},\ \bibinfo {pages} {100402} (\bibinfo {year} {2012})}\BibitemShut {NoStop}%
\bibitem [{\citenamefont {Yang}\ and\ \citenamefont {Navascu\'es}(2013)}]{Yang2013Robust}%
  \BibitemOpen
  \bibfield  {author} {\bibinfo {author} {\bibfnamefont {T.~H.}\ \bibnamefont {Yang}}\ and\ \bibinfo {author} {\bibfnamefont {M.}~\bibnamefont {Navascu\'es}},\ }\bibinfo {title} {\emph {Robust self-testing of unknown quantum systems into any entangled two-qubit states}},\ \href {https://doi.org/10.1103/PhysRevA.87.050102} {\bibfield  {journal} {\bibinfo  {journal} {Phys. Rev. A}\ }\textbf {\bibinfo {volume} {87}},\ \bibinfo {pages} {050102} (\bibinfo {year} {2013})}\BibitemShut {NoStop}%
\bibitem [{\citenamefont {Yang}\ \emph {et~al.}(2014)\citenamefont {Yang}, \citenamefont {V\'ertesi}, \citenamefont {Bancal}, \citenamefont {Scarani},\ and\ \citenamefont {Navascu\'es}}]{Yang2014Robust}%
  \BibitemOpen
  \bibfield  {author} {\bibinfo {author} {\bibfnamefont {T.~H.}\ \bibnamefont {Yang}}, \bibinfo {author} {\bibfnamefont {T.}~\bibnamefont {V\'ertesi}}, \bibinfo {author} {\bibfnamefont {J.-D.}\ \bibnamefont {Bancal}}, \bibinfo {author} {\bibfnamefont {V.}~\bibnamefont {Scarani}},\ and\ \bibinfo {author} {\bibfnamefont {M.}~\bibnamefont {Navascu\'es}},\ }\bibinfo {title} {\emph {Robust and versatile black-box certification of quantum devices}},\ \href {https://doi.org/10.1103/PhysRevLett.113.040401} {\bibfield  {journal} {\bibinfo  {journal} {Phys. Rev. Lett.}\ }\textbf {\bibinfo {volume} {113}},\ \bibinfo {pages} {040401} (\bibinfo {year} {2014})}\BibitemShut {NoStop}%
\bibitem [{\citenamefont {Bancal}\ \emph {et~al.}(2015)\citenamefont {Bancal}, \citenamefont {Navascu\'es}, \citenamefont {Scarani}, \citenamefont {V\'ertesi},\ and\ \citenamefont {Yang}}]{Bancal2015Physical}%
  \BibitemOpen
  \bibfield  {author} {\bibinfo {author} {\bibfnamefont {J.-D.}\ \bibnamefont {Bancal}}, \bibinfo {author} {\bibfnamefont {M.}~\bibnamefont {Navascu\'es}}, \bibinfo {author} {\bibfnamefont {V.}~\bibnamefont {Scarani}}, \bibinfo {author} {\bibfnamefont {T.}~\bibnamefont {V\'ertesi}},\ and\ \bibinfo {author} {\bibfnamefont {T.~H.}\ \bibnamefont {Yang}},\ }\bibinfo {title} {\emph {Physical characterization of quantum devices from nonlocal correlations}},\ \href {https://doi.org/10.1103/PhysRevA.91.022115} {\bibfield  {journal} {\bibinfo  {journal} {Phys. Rev. A}\ }\textbf {\bibinfo {volume} {91}},\ \bibinfo {pages} {022115} (\bibinfo {year} {2015})}\BibitemShut {NoStop}%
\bibitem [{\citenamefont {Coopmans}\ \emph {et~al.}(2019)\citenamefont {Coopmans}, \citenamefont {Kaniewski},\ and\ \citenamefont {Schaffner}}]{coopmans2019Robust}%
  \BibitemOpen
  \bibfield  {author} {\bibinfo {author} {\bibfnamefont {T.}~\bibnamefont {Coopmans}}, \bibinfo {author} {\bibfnamefont {J.}~\bibnamefont {Kaniewski}},\ and\ \bibinfo {author} {\bibfnamefont {C.}~\bibnamefont {Schaffner}},\ }\bibinfo {title} {\emph {Robust self-testing of two-qubit states}},\ \href {https://doi.org/10.1103/PhysRevA.99.052123} {\bibfield  {journal} {\bibinfo  {journal} {Phys. Rev. A}\ }\textbf {\bibinfo {volume} {99}},\ \bibinfo {pages} {052123} (\bibinfo {year} {2019})}\BibitemShut {NoStop}%
\bibitem [{\citenamefont {McKague}(2010)}]{mckague2010selftestinggraphstates}%
  \BibitemOpen
  \bibfield  {author} {\bibinfo {author} {\bibfnamefont {M.}~\bibnamefont {McKague}},\ }\href {https://arxiv.org/abs/1010.1989} {\bibinfo {title} {\emph {Self-testing graph states}}} (\bibinfo {year} {2010}),\ \Eprint{https://arxiv.org/abs/1010.1989} {arXiv:1010.1989 [quant-ph]} \BibitemShut {NoStop}%
\bibitem [{\citenamefont {Baccari}\ \emph {et~al.}(2020)\citenamefont {Baccari}, \citenamefont {Augusiak}, \citenamefont {\ifmmode \check{S}\else \v{S}\fi{}upi\ifmmode~\acute{c}\else \'{c}\fi{}}, \citenamefont {Tura},\ and\ \citenamefont {Ac\'{\i}n}}]{Baccari2020Scalable}%
  \BibitemOpen
  \bibfield  {author} {\bibinfo {author} {\bibfnamefont {F.}~\bibnamefont {Baccari}}, \bibinfo {author} {\bibfnamefont {R.}~\bibnamefont {Augusiak}}, \bibinfo {author} {\bibfnamefont {I.}~\bibnamefont {\ifmmode \check{S}\else \v{S}\fi{}upi\ifmmode~\acute{c}\else \'{c}\fi{}}}, \bibinfo {author} {\bibfnamefont {J.}~\bibnamefont {Tura}},\ and\ \bibinfo {author} {\bibfnamefont {A.}~\bibnamefont {Ac\'{\i}n}},\ }\bibinfo {title} {\emph {Scalable bell inequalities for qubit graph states and robust self-testing}},\ \href {https://doi.org/10.1103/PhysRevLett.124.020402} {\bibfield  {journal} {\bibinfo  {journal} {Phys. Rev. Lett.}\ }\textbf {\bibinfo {volume} {124}},\ \bibinfo {pages} {020402} (\bibinfo {year} {2020})}\BibitemShut {NoStop}%
\bibitem [{\citenamefont {Wu}\ \emph {et~al.}(2014)\citenamefont {Wu}, \citenamefont {Cai}, \citenamefont {Yang}, \citenamefont {Le}, \citenamefont {Bancal},\ and\ \citenamefont {Scarani}}]{Wu2014Self}%
  \BibitemOpen
  \bibfield  {author} {\bibinfo {author} {\bibfnamefont {X.}~\bibnamefont {Wu}}, \bibinfo {author} {\bibfnamefont {Y.}~\bibnamefont {Cai}}, \bibinfo {author} {\bibfnamefont {T.~H.}\ \bibnamefont {Yang}}, \bibinfo {author} {\bibfnamefont {H.~N.}\ \bibnamefont {Le}}, \bibinfo {author} {\bibfnamefont {J.-D.}\ \bibnamefont {Bancal}},\ and\ \bibinfo {author} {\bibfnamefont {V.}~\bibnamefont {Scarani}},\ }\bibinfo {title} {\emph {Robust self-testing of the three-qubit $w$ state}},\ \href {https://doi.org/10.1103/PhysRevA.90.042339} {\bibfield  {journal} {\bibinfo  {journal} {Phys. Rev. A}\ }\textbf {\bibinfo {volume} {90}},\ \bibinfo {pages} {042339} (\bibinfo {year} {2014})}\BibitemShut {NoStop}%
\bibitem [{\citenamefont {Šupić}\ \emph {et~al.}(2018)\citenamefont {Šupić}, \citenamefont {Coladangelo}, \citenamefont {Augusiak},\ and\ \citenamefont {Acín}}]{Supic2018Self-testingMPE}%
  \BibitemOpen
  \bibfield  {author} {\bibinfo {author} {\bibfnamefont {I.}~\bibnamefont {Šupić}}, \bibinfo {author} {\bibfnamefont {A.}~\bibnamefont {Coladangelo}}, \bibinfo {author} {\bibfnamefont {R.}~\bibnamefont {Augusiak}},\ and\ \bibinfo {author} {\bibfnamefont {A.}~\bibnamefont {Acín}},\ }\bibinfo {title} {\emph {Self-testing multipartite entangled states through projections onto two systems}},\ \href {https://doi.org/10.1088/1367-2630/aad89b} {\bibfield  {journal} {\bibinfo  {journal} {New Journal of Physics}\ }\textbf {\bibinfo {volume} {20}},\ \bibinfo {pages} {083041} (\bibinfo {year} {2018})}\BibitemShut {NoStop}%
\bibitem [{\citenamefont {Šupić}\ \emph {et~al.}(2023)\citenamefont {Šupić}, \citenamefont {Bowles}, \citenamefont {Renou}, \citenamefont {Acín},\ and\ \citenamefont {Hoban}}]{Supic2023Quantumnetworks}%
  \BibitemOpen
  \bibfield  {author} {\bibinfo {author} {\bibfnamefont {I.}~\bibnamefont {Šupić}}, \bibinfo {author} {\bibfnamefont {J.}~\bibnamefont {Bowles}}, \bibinfo {author} {\bibfnamefont {M.-O.}\ \bibnamefont {Renou}}, \bibinfo {author} {\bibfnamefont {A.}~\bibnamefont {Acín}},\ and\ \bibinfo {author} {\bibfnamefont {M.~J.}\ \bibnamefont {Hoban}},\ }\bibinfo {title} {\emph {Quantum networks self-test all entangled states}},\ \href {https://doi.org/10.1038/s41567-023-01945-4} {\bibfield  {journal} {\bibinfo  {journal} {Nature Physics}\ }\textbf {\bibinfo {volume} {19}},\ \bibinfo {pages} {670–675} (\bibinfo {year} {2023})},\ \Eprint{https://arxiv.org/abs/2201.05032} {arXiv:2201.05032 [quant-ph]}\BibitemShut {NoStop}%
\bibitem [{\citenamefont {McKague}\ \emph {et~al.}(2012)\citenamefont {McKague}, \citenamefont {Yang},\ and\ \citenamefont {Scarani}}]{McKague2012RobustSinglet}%
  \BibitemOpen
  \bibfield  {author} {\bibinfo {author} {\bibfnamefont {M.}~\bibnamefont {McKague}}, \bibinfo {author} {\bibfnamefont {T.~H.}\ \bibnamefont {Yang}},\ and\ \bibinfo {author} {\bibfnamefont {V.}~\bibnamefont {Scarani}},\ }\bibinfo {title} {\emph {Robust self-testing of the singlet}},\ \href {https://doi.org/10.1088/1751-8113/45/45/455304} {\bibfield  {journal} {\bibinfo  {journal} {Journal of Physics A: Mathematical and Theoretical}\ }\textbf {\bibinfo {volume} {45}},\ \bibinfo {pages} {455304} (\bibinfo {year} {2012})}\BibitemShut {NoStop}%
\bibitem [{\citenamefont {Šupić}\ \emph {et~al.}(2016)\citenamefont {Šupić}, \citenamefont {Augusiak}, \citenamefont {Salavrakos},\ and\ \citenamefont {Acín}}]{Supic2016chainedBI}%
  \BibitemOpen
  \bibfield  {author} {\bibinfo {author} {\bibfnamefont {I.}~\bibnamefont {Šupić}}, \bibinfo {author} {\bibfnamefont {R.}~\bibnamefont {Augusiak}}, \bibinfo {author} {\bibfnamefont {A.}~\bibnamefont {Salavrakos}},\ and\ \bibinfo {author} {\bibfnamefont {A.}~\bibnamefont {Acín}},\ }\bibinfo {title} {\emph {Self-testing protocols based on the chained bell inequalities}},\ \href {https://doi.org/10.1088/1367-2630/18/3/035013} {\bibfield  {journal} {\bibinfo  {journal} {New Journal of Physics}\ }\textbf {\bibinfo {volume} {18}},\ \bibinfo {pages} {035013} (\bibinfo {year} {2016})}\BibitemShut {NoStop}%
\bibitem [{\citenamefont {Bamps}\ and\ \citenamefont {Pironio}(2015)}]{Bamps2015SOS}%
  \BibitemOpen
  \bibfield  {author} {\bibinfo {author} {\bibfnamefont {C.}~\bibnamefont {Bamps}}\ and\ \bibinfo {author} {\bibfnamefont {S.}~\bibnamefont {Pironio}},\ }\bibinfo {title} {\emph {Sum-of-squares decompositions for a family of clauser-horne-shimony-holt-like inequalities and their application to self-testing}},\ \href {https://doi.org/10.1103/PhysRevA.91.052111} {\bibfield  {journal} {\bibinfo  {journal} {Phys. Rev. A}\ }\textbf {\bibinfo {volume} {91}},\ \bibinfo {pages} {052111} (\bibinfo {year} {2015})}\BibitemShut {NoStop}%
\bibitem [{\citenamefont {Wagner}\ \emph {et~al.}(2020)\citenamefont {Wagner}, \citenamefont {Bancal}, \citenamefont {Sangouard},\ and\ \citenamefont {Sekatski}}]{Wagner2020DICquantuminstruments}%
  \BibitemOpen
  \bibfield  {author} {\bibinfo {author} {\bibfnamefont {S.}~\bibnamefont {Wagner}}, \bibinfo {author} {\bibfnamefont {J.-D.}\ \bibnamefont {Bancal}}, \bibinfo {author} {\bibfnamefont {N.}~\bibnamefont {Sangouard}},\ and\ \bibinfo {author} {\bibfnamefont {P.}~\bibnamefont {Sekatski}},\ }\bibinfo {title} {\emph {Device-independent characterization of quantum instruments}},\ \href {https://doi.org/10.22331/q-2020-03-19-243} {\bibfield  {journal} {\bibinfo  {journal} {Quantum}\ }\textbf {\bibinfo {volume} {4}},\ \bibinfo {pages} {243} (\bibinfo {year} {2020})}\BibitemShut {NoStop}%
\bibitem [{\citenamefont {Miklin}\ and\ \citenamefont {Oszmaniec}(2020)}]{Miklin2020a}%
  \BibitemOpen
  \bibfield  {author} {\bibinfo {author} {\bibfnamefont {N.}~\bibnamefont {Miklin}}\ and\ \bibinfo {author} {\bibfnamefont {M.}~\bibnamefont {Oszmaniec}},\ }\bibinfo {title} {\emph {A universal scheme for robust self-testing in the prepare-and-measure scenario}},\ \href {https://api.semanticscholar.org/CorpusID:211677348} {\bibfield  {journal} {\bibinfo  {journal} {Quantum}\ }\textbf {\bibinfo {volume} {5}},\ \bibinfo {pages} {424} (\bibinfo {year} {2020})}\BibitemShut {NoStop}%
\bibitem [{\citenamefont {Devetak}\ and\ \citenamefont {Winter}(2005)}]{Devetak2005Distillation}%
  \BibitemOpen
  \bibfield  {author} {\bibinfo {author} {\bibfnamefont {I.}~\bibnamefont {Devetak}}\ and\ \bibinfo {author} {\bibfnamefont {A.}~\bibnamefont {Winter}},\ }\bibinfo {title} {\emph {Distillation of secret key and entanglement from quantum states}},\ \href {https://doi.org/10.1098/rspa.2004.1372} {\bibfield  {journal} {\bibinfo  {journal} {Proceedings of the Royal Society A: Mathematical, Physical and Engineering Sciences}\ }\textbf {\bibinfo {volume} {461}},\ \bibinfo {pages} {207–235} (\bibinfo {year} {2005})}\BibitemShut {NoStop}%
\bibitem [{\citenamefont {Rybotycki}\ \emph {et~al.}(2025)\citenamefont {Rybotycki}, \citenamefont {Białecki}, \citenamefont {Batle},\ and\ \citenamefont {Bednorz}}]{rybotycki2024violation}%
  \BibitemOpen
  \bibfield  {author} {\bibinfo {author} {\bibfnamefont {T.}~\bibnamefont {Rybotycki}}, \bibinfo {author} {\bibfnamefont {T.}~\bibnamefont {Białecki}}, \bibinfo {author} {\bibfnamefont {J.}~\bibnamefont {Batle}},\ and\ \bibinfo {author} {\bibfnamefont {A.}~\bibnamefont {Bednorz}},\ }\bibinfo {title} {\emph {Violation of no-signaling on a public quantum computer}},\ \href {https://doi.org/https://doi.org/10.1002/qute.202400661} {\bibfield  {journal} {\bibinfo  {journal} {Advanced Quantum Technologies}\ }\textbf {\bibinfo {volume} {8}},\ \bibinfo {pages} {2400661} (\bibinfo {year} {2025})}\BibitemShut {NoStop}%
\bibitem [{\citenamefont {Mahadev}(2018)}]{mahadev2018classical}%
  \BibitemOpen
  \bibfield  {author} {\bibinfo {author} {\bibfnamefont {U.}~\bibnamefont {Mahadev}},\ }\bibfield  {title} {\bibinfo {title} {\emph {Classical verification of quantum computations}},\ }in\ \href {https://doi.org/10.1109/FOCS.2018.00033} {\emph {\bibinfo {booktitle} {2018 IEEE 59th Annual Symposium on Foundations of Computer Science (FOCS)}}}\ (\bibinfo {year} {2018})\ pp.\ \bibinfo {pages} {259--267}\BibitemShut {NoStop}%
\bibitem [{\citenamefont {Metger}\ and\ \citenamefont {Vidick}(2021)}]{metger2021self}%
  \BibitemOpen
  \bibfield  {author} {\bibinfo {author} {\bibfnamefont {T.}~\bibnamefont {Metger}}\ and\ \bibinfo {author} {\bibfnamefont {T.}~\bibnamefont {Vidick}},\ }\bibinfo {title} {\emph {Self-testing of a single quantum device under computational assumptions}},\ \href {https://doi.org/10.22331/q-2021-09-16-544} {\bibfield  {journal} {\bibinfo  {journal} {{Quantum}}\ }\textbf {\bibinfo {volume} {5}},\ \bibinfo {pages} {544} (\bibinfo {year} {2021})}\BibitemShut {NoStop}%
\bibitem [{\citenamefont {Stricker}\ \emph {et~al.}(2022)\citenamefont {Stricker}, \citenamefont {Carrasco}, \citenamefont {Ringbauer}, \citenamefont {Postler}, \citenamefont {Meth}, \citenamefont {Edmunds}, \citenamefont {Schindler}, \citenamefont {Blatt}, \citenamefont {Zoller}, \citenamefont {Kraus},\ and\ \citenamefont {Monz}}]{stricker2022towards}%
  \BibitemOpen
  \bibfield  {author} {\bibinfo {author} {\bibfnamefont {R.}~\bibnamefont {Stricker}}, \bibinfo {author} {\bibfnamefont {J.}~\bibnamefont {Carrasco}}, \bibinfo {author} {\bibfnamefont {M.}~\bibnamefont {Ringbauer}}, \bibinfo {author} {\bibfnamefont {L.}~\bibnamefont {Postler}}, \bibinfo {author} {\bibfnamefont {M.}~\bibnamefont {Meth}}, \bibinfo {author} {\bibfnamefont {C.}~\bibnamefont {Edmunds}}, \bibinfo {author} {\bibfnamefont {P.}~\bibnamefont {Schindler}}, \bibinfo {author} {\bibfnamefont {R.}~\bibnamefont {Blatt}}, \bibinfo {author} {\bibfnamefont {P.}~\bibnamefont {Zoller}}, \bibinfo {author} {\bibfnamefont {B.}~\bibnamefont {Kraus}},\ and\ \bibinfo {author} {\bibfnamefont {T.}~\bibnamefont {Monz}},\ }\href@noop {} {\bibinfo {title} {\emph {Towards experimental classical verification of quantum computation}}},\ \Eprint{https://arxiv.org/abs/2203.07395} {arXiv:2203.07395 [quant-ph]} (\bibinfo {year} {2022})\BibitemShut {NoStop}%
\bibitem [{\citenamefont {{Ferracin}}\ \emph {et~al.}(2019)\citenamefont {{Ferracin}}, \citenamefont {{Kapourniotis}},\ and\ \citenamefont {{Datta}}}]{Ferracin18AccreditingOutputsOf}%
  \BibitemOpen
  \bibfield  {author} {\bibinfo {author} {\bibfnamefont {S.}~\bibnamefont {{Ferracin}}}, \bibinfo {author} {\bibfnamefont {T.}~\bibnamefont {{Kapourniotis}}},\ and\ \bibinfo {author} {\bibfnamefont {A.}~\bibnamefont {{Datta}}},\ }\bibinfo {title} {\emph {Accrediting outputs of noisy intermediate-scale quantum computing devices}},\ \href {https://doi.org/10.1088/1367-2630/ab4fd6} {\bibfield  {journal} {\bibinfo  {journal} {New J. Phys.}\ }\textbf {\bibinfo {volume} {21}},\ \bibinfo {pages} {113038} (\bibinfo {year} {2019})},\ \Eprint{https://arxiv.org/abs/1811.09709} {arXiv:1811.09709 [quant-ph]}\BibitemShut {NoStop}%
\bibitem [{\citenamefont {{Ferracin}}\ \emph {et~al.}(2021)\citenamefont {{Ferracin}}, \citenamefont {{Merkel}}, \citenamefont {{McKay}},\ and\ \citenamefont {{Datta}}}]{Ferracin21ExperimentalAccreditationOf}%
  \BibitemOpen
  \bibfield  {author} {\bibinfo {author} {\bibfnamefont {S.}~\bibnamefont {{Ferracin}}}, \bibinfo {author} {\bibfnamefont {S.~T.}\ \bibnamefont {{Merkel}}}, \bibinfo {author} {\bibfnamefont {D.}~\bibnamefont {{McKay}}},\ and\ \bibinfo {author} {\bibfnamefont {A.}~\bibnamefont {{Datta}}},\ }\bibinfo {title} {\emph {Experimental accreditation of outputs of noisy quantum computers}},\ \href {https://doi.org/10.1103/PhysRevA.104.042603} {\bibfield  {journal} {\bibinfo  {journal} {\pra}\ }\textbf {\bibinfo {volume} {104}},\ \bibinfo {eid} {042603} (\bibinfo {year} {2021})},\ \Eprint{https://arxiv.org/abs/2103.06603} {arXiv:2103.06603 [quant-ph]}\BibitemShut {NoStop}%
\bibitem [{\citenamefont {van Dam}\ \emph {et~al.}(2007)\citenamefont {van Dam}, \citenamefont {Magniez}, \citenamefont {Mosca},\ and\ \citenamefont {Santha}}]{vanDam2007self-testing}%
  \BibitemOpen
  \bibfield  {author} {\bibinfo {author} {\bibfnamefont {W.}~\bibnamefont {van Dam}}, \bibinfo {author} {\bibfnamefont {F.}~\bibnamefont {Magniez}}, \bibinfo {author} {\bibfnamefont {M.}~\bibnamefont {Mosca}},\ and\ \bibinfo {author} {\bibfnamefont {M.}~\bibnamefont {Santha}},\ }\bibinfo {title} {\emph {Self-testing of universal and fault-tolerant sets of quantum gates}},\ \href {https://doi.org/10.1137/S0097539702404377} {\bibfield  {journal} {\bibinfo  {journal} {SIAM Journal on Computing}\ }\textbf {\bibinfo {volume} {37}},\ \bibinfo {pages} {611} (\bibinfo {year} {2007})},\ \Eprint{https://arxiv.org/abs/https://doi.org/10.1137/S0097539702404377} {https://doi.org/10.1137/S0097539702404377}\BibitemShut {NoStop}%
\bibitem [{\citenamefont {Wigner}(1931)}]{wigner1931gruppentheorie}%
  \BibitemOpen
  \bibfield  {author} {\bibinfo {author} {\bibfnamefont {E.}~\bibnamefont {Wigner}},\ }\href {https://doi.org/10.1007/978-3-663-02555-9} {\emph {\bibinfo {title} {Gruppentheorie und ihre Anwendung auf die Quantenmechanik der Atomspektren}}}\ (\bibinfo  {publisher} {Vieweg+Teubner Verlag},\ \bibinfo {year} {1931})\BibitemShut {NoStop}%
\bibitem [{\citenamefont {N{\"{o}}ller}\ \emph {et~al.}(2025)\citenamefont {N{\"{o}}ller}, \citenamefont {Miklin}, \citenamefont {Kliesch},\ and\ \citenamefont {Gachechiladze}}]{noller2025classical}%
  \BibitemOpen
  \bibfield  {author} {\bibinfo {author} {\bibfnamefont {J.}~\bibnamefont {N{\"{o}}ller}}, \bibinfo {author} {\bibfnamefont {N.}~\bibnamefont {Miklin}}, \bibinfo {author} {\bibfnamefont {M.}~\bibnamefont {Kliesch}},\ and\ \bibinfo {author} {\bibfnamefont {M.}~\bibnamefont {Gachechiladze}},\ }\bibinfo {title} {\emph {Classical certification of quantum gates under the dimension assumption}},\ \href {https://doi.org/10.22331/q-2025-08-08-1825} {\bibfield  {journal} {\bibinfo  {journal} {{Quantum}}\ }\textbf {\bibinfo {volume} {9}},\ \bibinfo {pages} {1825} (\bibinfo {year} {2025})}\BibitemShut {NoStop}%
\bibitem [{\citenamefont {Schroeder}\ \emph {et~al.}(2025)\citenamefont {Schroeder}, \citenamefont {Vieira}, \citenamefont {Nöller}, \citenamefont {Miklin},\ and\ \citenamefont {Gachechiladze}}]{schroeder2025certifying}%
  \BibitemOpen
  \bibfield  {author} {\bibinfo {author} {\bibfnamefont {A.}~\bibnamefont {Schroeder}}, \bibinfo {author} {\bibfnamefont {L.~B.}\ \bibnamefont {Vieira}}, \bibinfo {author} {\bibfnamefont {J.}~\bibnamefont {Nöller}}, \bibinfo {author} {\bibfnamefont {N.}~\bibnamefont {Miklin}},\ and\ \bibinfo {author} {\bibfnamefont {M.}~\bibnamefont {Gachechiladze}},\ }\bibinfo {title} {\emph {Certifying quantum gates via automata advantage}},\ \href {https://arxiv.org/abs/2510.09575} {\bibfield  {journal} {\bibinfo  {journal} {arXiv preprint}\ } (\bibinfo {year} {2025})},\ \Eprint{https://arxiv.org/abs/2510.09575} {2510.09575}\BibitemShut {NoStop}%
\bibitem [{git()}]{git}%
  \BibitemOpen
  \href@noop {} {\bibinfo {title} {\emph {Git{H}ub repository: {MATLAB} programs used for the numerical experiments}}},\ \bibinfo {howpublished} {\url{https://github.com/nikolai-miklin/QSQ_numerics}}\BibitemShut {NoStop}%
\bibitem [{\citenamefont {Bochnak}\ \emph {et~al.}(1998)\citenamefont {Bochnak}, \citenamefont {Coste},\ and\ \citenamefont {Roy}}]{Bochnak1998}%
  \BibitemOpen
  \bibfield  {author} {\bibinfo {author} {\bibfnamefont {J.}~\bibnamefont {Bochnak}}, \bibinfo {author} {\bibfnamefont {M.}~\bibnamefont {Coste}},\ and\ \bibinfo {author} {\bibfnamefont {M.-F.}\ \bibnamefont {Roy}},\ }\href {https://doi.org/10.1007/978-3-662-03718-8} {\emph {\bibinfo {title} {Real Algebraic Geometry}}}\ (\bibinfo  {publisher} {Springer Berlin Heidelberg},\ \bibinfo {year} {1998})\BibitemShut {NoStop}%
\bibitem [{\citenamefont {H\"{o}rmander}(2005)}]{Hoermander2005}%
  \BibitemOpen
  \bibfield  {author} {\bibinfo {author} {\bibfnamefont {L.}~\bibnamefont {H\"{o}rmander}},\ }\href {https://doi.org/10.1007/b138375} {\emph {\bibinfo {title} {The Analysis of Linear Partial Differential Operators II}}}\ (\bibinfo  {publisher} {Springer Berlin Heidelberg},\ \bibinfo {year} {2005})\BibitemShut {NoStop}%
\bibitem [{\citenamefont {Aliprantis}\ and\ \citenamefont {Border}(2006)}]{aliprantis2006infinite}%
  \BibitemOpen
  \bibfield  {author} {\bibinfo {author} {\bibfnamefont {C.~D.}\ \bibnamefont {Aliprantis}}\ and\ \bibinfo {author} {\bibfnamefont {K.~C.}\ \bibnamefont {Border}},\ }\href {https://doi.org/10.1007/3-540-29587-9} {\emph {\bibinfo {title} {Infinite dimensional analysis: a hitchhiker's guide}}}\ (\bibinfo  {publisher} {Springer},\ \bibinfo {year} {2006})\BibitemShut {NoStop}%
\bibitem [{\citenamefont {Fuchs}\ and\ \citenamefont {van~de Graaf}(1999)}]{Fuchs1999}%
  \BibitemOpen
  \bibfield  {author} {\bibinfo {author} {\bibfnamefont {C.}~\bibnamefont {Fuchs}}\ and\ \bibinfo {author} {\bibfnamefont {J.}~\bibnamefont {van~de Graaf}},\ }\bibinfo {title} {\emph {{Cryptographic distinguishability measures for quantum-mechanical states}}},\ \href {https://doi.org/10.1109/18.761271} {\bibfield  {journal} {\bibinfo  {journal} {IEEE Trans. Inf. Theory}\ }\textbf {\bibinfo {volume} {45}},\ \bibinfo {pages} {1216} (\bibinfo {year} {1999})}\BibitemShut {NoStop}%
\end{thebibliography}%

\end{document}